\documentclass[a4paper,UKenglish,cleveref,thm-restate,openany]{lipics-v2019}

\graphicspath{{./figures/}}%

\usepackage{amsfonts}
\usepackage[page,header]{appendix}
\usepackage{titletoc}
\usepackage{titleps}
\usepackage{amssymb}
\usepackage{setspace}
\usepackage{wrapfig}
\usepackage[framemethod=TikZ]{mdframed}

\usepackage[T1]{fontenc}
\usepackage[utf8]{inputenc} %

\usepackage{etoolbox}
\usepackage{xifthen}
\usepackage{xspace}

\usepackage[all]{hypcap}

\usepackage{nicefrac}

\usepackage{upquote}

\usepackage{stmaryrd}
\SetSymbolFont{stmry}{bold}{U}{stmry}{m}{n}
\usepackage{anyfontsize}

\usepackage{proof}

\usepackage{mathtools} %

\usepackage{mathrsfs}
\usepackage[most]{tcolorbox}
\usepackage{xcolor}

\usepackage{graphicx}
\usepackage{url}
\usepackage{color}
\usepackage{comment}
\usepackage{rotating}
\usepackage{listings}
\usepackage{multicol}
\usepackage{multirow}
\usepackage[noend]{algpseudocode}
\usepackage{tikz}
\usepackage{fancybox}
\usetikzlibrary{positioning,calc,backgrounds,fit,shapes,arrows.meta,cd}
\usepackage{import}
\usepackage{makecell}
\usepackage{relsize}

\usepackage{listings}

\usepackage{csquotes}

\usepackage{booktabs}

\usepackage[nomargin,inline,draft]{fixme}
\fxsetup{theme=color,mode=multiuser}
\FXRegisterAuthor{RN}{aRN}{RN} %
\FXRegisterAuthor{NY}{aNY}{NY} %
\FXRegisterAuthor{KI}{aKI}{KI} %

\usepackage{stfloats}
\fnbelowfloat

\usepackage[group-four-digits,per-mode=fraction]{siunitx}

\usepackage{adjustbox} %
\usepackage{arydshln} %

\usepackage{pdfcomment}

\hypersetup{hidelinks,
  unicode,
  colorlinks=true,
  allcolors=black,
  pdfstartview=Fit,
  breaklinks=true}

\makeatletter
\g@addto@macro{\UrlBreaks}{\UrlOrds}
\makeatother

\SetExpansion
  [ context = sloppy,
    stretch = 30,
    shrink = 60,
    step = 5 ]
  { encoding = {OT1,T1,TS1} }
  { }

\newtoggle{techreport}%
\toggletrue{techreport}%

\title{Multiparty Session Programming with\texorpdfstring{\\}{ }Global Protocol Combinators}

\titlerunning{Multiparty Session Programming with Global Protocol Combinators}

\author{Keigo Imai\footnote{Corresponding author}}
       {Gifu University, Japan}
       {keigoi@gifu-u.ac.jp}
       {https://orcid.org/0000-0003-1602-8473}
       {}%

\author{Rumyana Neykova}
       {Brunel University London, UK}
       {Rumyana.Neykova@brunel.ac.uk}
       {https://orcid.org/0000-0002-3925-8557}
       {}%

\author{Nobuko Yoshida}
       {Imperial College London, UK}
       {n.yoshida@imperial.ac.uk}
       {https://orcid.org/0000-0002-2755-7728}
       {}%

\author{Shoji Yuen}
       {Nagoya University, Japan}
       {yuen@i.nagoya-u.ac.jp}
       {https://orcid.org/0000-0003-2642-0647}
       {}%

\authorrunning{K. Imai, R. Neykova, N. Yoshida and S. Yuen}

\Copyright{Keigo Imai, Rumyana Neykova, Nobuko Yoshida and Shoji Yuen}



\iftoggle{techreport}{}{%
  \relatedversion{A full version of the paper \cite{IRYY2020TechReport}  is available at \url{https://arxiv.org/abs/FIXME.Todo}}%
}

\supplement{A source code repository for the accompanying artifact is available at \url{https://github.com/keigoi/ocaml-mpst/}}

\acknowledgements{
  We thank Jacques Garrigue and Oleg Kiselyov for their comments on an
  early version of this paper.
  Our work is partially supported by the
  first author’s visitor funding to Imperial College London and Brunel
  University London supported by Gifu University, VeTSS, JSPS KAKENHI
  Grant Numbers JP17H01722, JP17K19969 and JP17K12662, Short-term
  visiting Fellowship S19068, EPSRC Doctoral Prize Fellowship, and
  EPSRC EP/K011715/1, EP/K034413/1, EP/L00058X/1, EP/N027833/1,
  EP/N028201/1, EP/T006544/1 and EP/T014709/1.
}

\nolinenumbers

\EventEditors{Robert Hirschfeld and Tobias Pape}
\EventNoEds{2}
\EventLongTitle{34th European Conference on Object-Oriented Programming (ECOOP 2020)}
\EventShortTitle{ECOOP 2020}
\EventAcronym{ECOOP}
\EventYear{2020}
\EventDate{July 13--17, 2020}
\EventLocation{Berlin, Germany}
\EventLogo{}
\SeriesVolume{166}
\ArticleNo{9}
\newcommand{\cfig}[1]{Fig.~\ref{#1}}

\newcommand{\grmeq}{{::=}}
\newcommand{\grmor}{{\;\text{\Large$\mid$}\;}}

\newcommand\TableStrut{\rule{0pt}{1.25\normalbaselineskip}\rule[-1\normalbaselineskip]{0pt}{0pt}}
\newcommand\TableStrutDouble{\rule{0pt}{1.75\normalbaselineskip}\rule[-1.50\normalbaselineskip]{0pt}{0pt}}
\newcommand\TableStrutTriple{\rule{0pt}{2.25\normalbaselineskip}\rule[-2.00\normalbaselineskip]{0pt}{0pt}}

\newcommand{\OCaml}{OCaml\xspace}
\newcommand{\ourlibrary}{\texttt{ocaml-mpst}\xspace}
\newcommand{\Sec}{\textsection}

\newcommand{\channelvector}{channel vector\xspace}

\newcommand{\Rptp}[1]{\ensuremath{{\color{roleColor}\roleFmt{#1}}}}

\newcommand{\labAuth}{{\labFmt{auth}}}

\newcommand{\labOk}{{\labFmt{ok}}}
\newcommand{\labCancel}{{\labFmt{cancel}}}

\newcommand{\OCAMLMPST}{\mustL}

\newcommand{\myparagraph}[1]{\smallskip\noindent\textit{\bf #1\ }}

\newcommand*{\thmstart}{\setlength{\abovedisplayskip}{2pt}\setlength{\belowdisplayskip}{2pt}\rm}

\newtheorem{mycounter}{Dummy}[section]{\bfseries}{\itshape}
\newtheorem{THM}[mycounter]{Theorem}{\bfseries}{\itshape}
\newtheorem{DEF}[mycounter]{Definition}{\bfseries}{\itshape}
{\bfseries}{\itshape}
\newtheorem{COL}[mycounter]{Corollary}{\bfseries}{\itshape}
\newtheorem{LEM}[mycounter]{Lemma}{\bfseries}{\itshape}
\newtheorem{PROP}[mycounter]{Proposition}{\bfseries}{\itshape}
{\bfseries}{\itshape}
\newtheorem{EX}[mycounter]{Example}{\bfseries}{\itshape}
\Crefname{THM}{Theorem}{Theorems}
\Crefname{DEF}{Definition}{Definitions}
\Crefname{CONV}{Convention}{Conventions}
\Crefname{COL}{Corollary}{Corollaries}
\Crefname{LEM}{Lemma}{Lemmas}
\Crefname{PROP}{Proposition}{Propositions}
\Crefname{REM}{Remark}{Remarks}
\Crefname{EX}{Example}{Examples}

\newcommand{\CASE}{{\bfseries Case}\xspace}

\newcommand{\todo}[1]{[{\color{red}{\textbf{Todo:} #1}}]}

\newcommand{\precameraready}[1]{{#1}}
\definecolor{modify}{rgb}{0.0, 0.0, 0.0}%
\definecolor{modified}{rgb}{0.0, 0.0, 0.0}%
\makeatletter
\newsavebox{\@brx}
\newcommand{\llangle}[1][]{\savebox{\@brx}{\(\m@th{#1\langle}\)}%
  \mathopen{\copy\@brx\kern-0.5\wd\@brx\usebox{\@brx}}}
\newcommand{\rrangle}[1][]{\savebox{\@brx}{\(\m@th{#1\rangle}\)}%
  \mathclose{\copy\@brx\kern-0.5\wd\@brx\usebox{\@brx}}}
\makeatother

\newcommand{\elip}{\ifmmode\mathinner{\ldotp\kern-0.2em\ldotp\kern-0.2em\ldotp}\else.\kern-0.13em.\kern-0.13em.\fi}
\newcommand{\elipc}{\ifmmode\mathinner{\cdotp\kern-0.2em\cdotp\kern-0.2em\cdotp}\else.\kern-0.13em.\kern-0.13em.\fi}
\newcommand{\eqdeff}{\ensuremath{\stackrel{\mathrm{def}}{=}}}

\newcommand{\roleC}{{\color{roleColor}\roleFmt{c}}}
\newcommand{\roleA}{{\color{roleColor}\roleFmt{a}}}
\newcommand{\roleS}{{\color{roleColor}\roleFmt{s}}}

\newcommand{\labLoop}{\labFmt{loop}}
\newcommand{\labStop}{\labFmt{stop}}

\newcommand{\dlabel}[1]{\textcolor{lipicsGray}{\sffamily\bfseries\mathversion{bold}#1}}
\newcommand{\Rone}{\dlabel{R1}\xspace}
\newcommand{\Rtwo}{\dlabel{R2}\xspace}
\newcommand{\Rthree}{\dlabel{R3}\xspace}
\newcommand{\Mone}{\dlabel{M1}\xspace}
\newcommand{\Mtwo}{\dlabel{M2}\xspace}

\newcommand{\PSet}{\!\mathscr{P}\!}
\newcommand{\ASigma}{ \mathbb{A}}

\newcommand{\ifempty}[3]{%
  \ifthenelse{\isempty{#1}}{#2}{#3}%
}%

\newcommand{\dom}[1]{{\color{black}\operatorname{dom}\!\left({#1}\right)}}%
\newcommand{\fn}[1]{\operatorname{fn}\!\left({#1}\right)}%
\newcommand{\fv}[1]{\operatorname{fv}\!\left({#1}\right)}%
\newcommand{\notImplies}{\mathrel{{\kern 0.5em}{\not{\kern -0.5em}\implies}}}%
\newcommand{\notImpliedBy}{\mathrel{{\kern 1em}{\not{\kern -1em}\impliedby}}}%
\newcommand{\coloncolonequals}{\Coloneqq}%
\newcommand{\bnfsep}{\mathbin{\;\big|\;}}%

\def\ie{i.e.\@\xspace}%
\definecolor{ruleColor}{rgb}{0.1, 0.3, 0.1}%
\newcommand{\inferrule}[1]{{\color{ruleColor}\text{\scriptsize [#1]}}}%
\newcommand{\inference}[3][]{\infer[\ifempty{#1}{}{\inferrule{#1}}]{#3}{#2}}%
\newcommand{\cinference}[3][]{\infer=[\ifempty{#1}{}{\inferrule{#1}}]{#3}{#2}}%
\newcommand{\setenum}[1]{\mathord{{\color{black}\left\{#1\right\}}}}%
\definecolor{groundColor}{rgb}{0.38, 0.25, 0.32}%

\newcommand{\bind}[2]{\nicefrac{#2}{#1}}
\newcommand{\substenum}[1]{\mathord{{\color{black}\left\{{#1}\right\}}}}
\newcommand{\subst}[2]{\substenum{\bind{#1}{#2}}}
\newcommand{\mapsubst}[2]{{#1}{\mapsto}{#2}}

\definecolor{roleColor}{rgb}{0.5, 0.0, 0.0}%
\newcommand{\roleCol}[1]{{\color{roleColor}#1}}%
\newcommand{\roleUnivSet}{\roleCol{\mathfrak{R}}}%
\newcommand{\roleSet}{\roleCol{\mathbb{R}}}%
\newcommand{\roleFmt}[1]{\boldsymbol{\roleCol{\mathtt{#1}}}}%
\newcommand{\roleP}[1][]{%
  \ifempty{#1}{{\color{roleColor}\roleFmt{p}}}{{\color{roleColor}\roleFmt{p}_{{#1}}}}%
}%
\newcommand{\roleQ}[1][]{%
  \ifempty{#1}{{\color{roleColor}\roleFmt{q}}}{{\color{roleColor}\roleFmt{q}_{#1}}}%
}%
\newcommand{\roleR}[1][]{%
  \ifempty{#1}{{\color{roleColor}\roleFmt{r}}}{{\color{roleColor}\roleFmt{r}_{\!#1}}}%
}%
\newcommand{\labFmt}[2][]{\color{roleColor}{\ifempty{#1}{\mathtt{#2}}{\mathtt{#2}\textsubscript{#1}}}}%

\definecolor{mpColor}{rgb}{0, 0, 0}%
\newcommand{\mpFmt}[1]{{\color{mpColor}#1}}%

\newcommand{\mpLab}[1][]{%
  \mpFmt{\ifempty{#1}{\labFmt{m}}{{\labFmt{m}}_{\mathnormal #1}}}%
}%
\newcommand{\mpLabi}[1][]{%
  \mpFmt{\ifempty{#1}{\labFmt{m}'}{\labFmt{m}'_{\mathnormal #1}}}%
}%
\newcommand{\mpLabFmt}[1]{\mpFmt{\labFmt{#1}}}%

\newcommand{\mpChanRole}[2]{\mpFmt{{#1}[{#2}]}}%

\newcommand{\mpNil}{\mpFmt{\mathbf{0}}}%
\newcommand{\mpSeq}{\mathbin{\mpFmt{\!.\!}}}%
\newcommand{\mpChoice}[3]{%
  \mpFmt{%
    \mpLabFmt{#1}\ifempty{#2}{}{({#2})}\ifempty{#3}{}{\vphantom{x}\mpSeq {#3}}%
  }%
}%
\newcommand{\mpChoiceNoBind}[3]{%
  \mpFmt{%
    \mpLabFmt{#1}\ifempty{#2}{}{\langle{#2}\rangle}\ifempty{#3}{}{\vphantom{x}\mpSeq {#3}}%
  }%
}%
\newcommand{\mpBranch}[6]{%
  \mpFmt{%
    {#1}[\roleFmt{#2}] \mathbin{\!\ifempty{\sum}{\&}{\sum_{#3}}\!}%
    \ifempty{#3}{%
      \mpChoice{#4}{#5}{#6}%
    }{%
      \mpChoice{#4}{#5}{#6}%
    }%
  }%
}%
\newcommand{\mpSel}[5]{%
  \mpFmt{%
    {#1}[\roleFmt{#2}] \mathbin{\!\oplus\!}%
    \mpChoiceNoBind{#3}{#4}{#5}%
  }%
}%
\newcommand{\mpPar}{\mathbin{\mpFmt{\mid}}}%
\newcommand{\mpBigPar}[2]{\mathbin{\mpFmt{\big|_{#1}}{#2}}}%
\newcommand{\mpRes}[2]{\mpFmt{\left(\mathbf{\nu}{#1}\right){#2}}}%
\newcommand{\mpJustDef}[3]{%
  \mpFmt{{#1}(#2) = {#3}}%
}%
\newcommand{\mpDef}[4]{%
  \mpFmt{\mathbf{def}\;\mpJustDef{#1}{#2}{#3}\;\mathbf{in}\;{#4}}%
}%
\newcommand{\mpDefAbbrev}[2]{%
  \mpFmt{\mathbf{def}\;{#1}\;\mathbf{in}\;{#2}}%
}%
\newcommand{\mpCall}[2]{\mpFmt{{#1}\!\left\langle{#2}\right\rangle}}%
\newcommand{\mpCtx}[1][]{\mpFmt{\ifempty{#1}{\mathbb{C}}{\mathbb{C}_{#1}}}}%
\newcommand{\mpCtxi}[1][]{\mpFmt{\ifempty{#1}{\mathbb{C}'}{\mathbb{C}'_{#1}}}}%
\newcommand{\mpCtxHole}{[\,]}%
\newcommand{\mpCtxApp}[2]{{#1}\!\left[{#2}\right]}%

\newcommand{\mpC}[1][]{\mpFmt{\ifempty{#1}{c}{c_{#1}}}}%
\newcommand{\mpD}[1][]{\mpFmt{\ifempty{#1}{d}{d_{#1}}}}%
\newcommand{\mpS}[1][]{\mpFmt{\ifempty{#1}{s}{s_{#1}}}}%
\newcommand{\mpSi}[1][]{\mpFmt{\ifempty{#1}{s'}{s'_{#1}}}}%
\newcommand{\mpX}[1][]{\mpFmt{\ifempty{#1}{X}{X_{#1}}}}%
\newcommand{\mpY}[1][]{\mpFmt{\ifempty{#1}{Y}{Y_{#1}}}}%
\newcommand{\mpP}[1][]{\mpFmt{\ifempty{#1}{P}{P_{#1}}}}%
\newcommand{\mpPi}[1][]{\mpFmt{\ifempty{#1}{P'}{P'_{#1}}}}%
\newcommand{\mpPii}[1][]{\mpFmt{\ifempty{#1}{P''}{P''_{#1}}}}%
\newcommand{\mpQ}[1][]{\mpFmt{\ifempty{#1}{Q}{Q_{#1}}}}%
\newcommand{\mpQi}[1][]{\mpFmt{\ifempty{#1}{Q'}{Q'_{#1}}}}%
\newcommand{\mpDefD}[1][]{\mpFmt{\ifempty{#1}{D}{D_{#1}}}}%
\newcommand{\mpMove}{\to}%
\newcommand{\mpMoveStar}{\mathrel{\mpMove{}^{\!\!\!*}}}%
\newcommand{\mpNotMoveP}[1]{\mpFmt{#1}\!\not{\mpMove}}%

\definecolor{gtColor}{rgb}{0.43, 0.21, 0.1}%
\newcommand{\gtFmt}[1]{{\color{gtColor}#1}}%
\newcommand{\gtMsgFmt}[1]{\gtFmt{\labFmt{#1}}}%
\newcommand{\gtLab}[1][]{%
  \ifempty{#1}{\gtMsgFmt{m}}{{\color{gtColor}\gtMsgFmt{m}_{#1}}}%
}%

\newcommand{\gtG}[1][]{\gtFmt{\ifempty{#1}{G}{G_{#1}}}}%

\newcommand{\gtSeq}{\mathbin{\gtFmt{.}}}%

\newcommand{\gtCommRaw}[3]{%
  \gtFmt{%
    {#1} {\to} {#2}{:}%
    \left\{%
      {#3}%
    \right\}%
  }%
}%
\newcommand{\gtComm}[6]{%
  \gtFmt{%
    \gtCommRaw{#1}{#2}{%
      \gtCommChoice{#4}{#5}{#6}%
    }_{#3}%
  }%
}%
\newcommand{\gtCommChoice}[3]{%
  \gtFmt{%
    \gtMsgFmt{#1}\ifempty{#2}{}{({#2})}%
    \ifempty{#3}{}{\vphantom{x}{\gtSeq}{#3}}%
  }%
}%

\newcommand{\gtEnd}{\gtFmt{\mathbf{end}}}%

\newcommand{\gtRec}[2]{\gtFmt{\mu{#1}.{#2}}}%
\newcommand{\gtRecVarBase}{\gtFmt{\mathbf{t}}}%
\newcommand{\gtRecVar}[1][]{\gtFmt{\ifempty{#1}{\gtRecVarBase}{\gtRecVarBase_{#1}}}}%

\newcommand{\gtRoles}[1]{{\color{black} \operatorname{roles}(\gtFmt{#1})}}%

\newcommand{\gtProj}[2]{%
  {\color{stColor}\gtFmt{#1} {\upharpoonright} #2}%
}%

\definecolor{stColor}{rgb}{0.1, 0.4, 0.1}%
\newcommand{\stFmt}[1]{{\color{stColor}#1}}%

\newcommand{\stChoice}[2]{\stFmt{#1}\ifempty{#2}{}{\stFmt{({#2})}}}%

\newcommand{\stSeq}{\mathbin{\!\stFmt{.}\!}}%
\newcommand{\stIntC}{\mathbin{\stFmt{\oplus}}}%
\newcommand{\stIntSum}[3]{\roleFmt{#1}\stFmt{\oplus_{#2}{#3}}}%
\newcommand{\stExtC}{\mathbin{\stFmt{\&}}}%
\newcommand{\stExtSum}[3]{\roleFmt{#1}\stFmt{\&_{#2}{#3}}}%
\newcommand{\stRec}[2]{\stFmt{\mu{#1}.{#2}}}%
\newcommand{\stEnd}{\stFmt{\mathbf{end}}}%

\newcommand{\stLab}[1][]{\mpFmt{\ifempty{#1}{\labFmt{m}}{\labFmt{m}_{#1}}}}%

\newcommand{\stS}[1][]{\stFmt{\ifempty{#1}{S}{S_{#1}}}}%
\newcommand{\stSi}[1][]{\stFmt{\ifempty{#1}{S'}{S'_{#1}}}}%

\newcommand{\stT}[1][]{\stFmt{\ifempty{#1}{T}{T_{#1}}}}%
\newcommand{\stTi}[1][]{\stFmt{\ifempty{#1}{T'}{T'_{#1}}}}%

\newcommand{\stRecVarBase}{\stFmt{\mathbf{t}}}%
\newcommand{\stRecVar}[1][]{\stFmt{\ifempty{#1}{\stRecVarBase}{\stRecVarBase_{#1}}}}%
\newcommand{\stRecVari}[1][]{\stFmt{\ifempty{#1}{\stRecVar'}{\stRecVar'_{#1}}}}%
\newcommand{\stMerge}[2]{\stFmt{\bigsqcap_{#1}{#2}}}%
\newcommand{\stBinMerge}{\mathbin{\stFmt{\sqcap}}}%

\newcommand{\stSub}{\mathrel{\stFmt{\leqslant}}}%
\newcommand{\stNotSub}{\mathrel{\stFmt{\not\leqslant}}}%
\definecolor{ptColor}{rgb}{0.20, 0.29, 0.09}%
\newcommand{\stEnv}[1][]{\stFmt{\ifempty{#1}{\Gamma}{\Gamma_{#1}}}}%
\newcommand{\stEnvi}[1][]{\stFmt{\ifempty{#1}{\Gamma'}{\Gamma'_{#1}}}}%
\newcommand{\stEnvEmpty}{\stFmt{\emptyset}}%
\newcommand{\stEnvMap}[2]{\stFmt{\mpFmt{#1}\mathbin{\!:\!}{#2}}}%
\newcommand{\stEnvComp}{\mathpunct{\stFmt{,}}}%
\newcommand{\stEnvApp}[2]{\stFmt{#1\!\left(\mpFmt{#2}\right)}}%

\newcommand{\stEnvMove}{\mathrel{\stFmt{\to}}}%
\newcommand{\stEnvMoveP}[1]{{#1}\!\!\stEnvMove}%

\newcommand{\stEnvEndPred}{\operatorname{end}}%
\newcommand{\stEnvEndP}[1]{\stEnvEndPred(\stFmt{#1})}%

\newcommand{\stEnvDFPred}{\operatorname{df}}%
\newcommand{\stEnvDFP}[1]{\stEnvDFPred(\stFmt{#1})}%

\newcommand{\stEnvLivePlusPred}{\operatorname{live^{+}}}%
\newcommand{\stEnvLivePlusP}[1]{\stEnvLivePlusPred(\stFmt{#1})}%

\newcommand{\stEnvSafePred}{\operatorname{safe}}%
\newcommand{\stEnvSafeP}[1]{\stEnvSafePred(\stFmt{#1})}%

\newcommand{\stEnvEntails}[3]{%
  \stFmt{#1} \vdash \stFmt{\mpFmt{#2} \mathbin{\!:\!} {#3}}%
}%
\newcommand{\mpEnv}[1][]{\stFmt{\ifempty{#1}{\Theta}{\Theta_{#1}}}}%
\newcommand{\mpEnvEmpty}{\stFmt{\emptyset}}%
\newcommand{\mpEnvMap}[2]{\stFmt{\mpFmt{#1}{:}\stFmt{#2}}}%
\newcommand{\mpEnvComp}{\mathpunct{\stFmt{,}}}%
\newcommand{\mpEnvApp}[2]{\stFmt{#1}\!\left(\mpFmt{#2}\right)}%

\newcommand{\mpEnvEntails}[3]{%
  \stFmt{#1} \vdash \stFmt{\mpFmt{#2} \mathbin{\!:\!} \stFmt{#3}}%
}%

\newcommand{\stJudge}[3]{%
  \stFmt{\ifempty{#1}{#2}{{#1} \cdot {#2}}%
  \mathrel{\mpFmt{\vdash}} \mpFmt{#3}}%
}%
\newcommand{\stEnvQMoveModQ}[1]{\mathrel{\stFmt{\rightarrow}_{{#1}}}}%
\newcommand{\stEnvQNotMoveModQP}[2]{%
  {#2}/{\kern -0.5em}{\stEnvQMoveModQ{#1}}}%

\newcommand{\predP}[1][]{\ifempty{#1}{\varphi}{\varphi_{#1}}}%
\newcommand{\predPApp}[2][]{\ifempty{#1}{\predP}{\predP[{#1}]}\!\left({#2}\right)}%

\newcommand{\muCol}[1]{{\color{red}#1}}%
\newcommand{\muWordEmpty}[1][]{\muCol{\epsilon}}%
\newcommand{\iruleMPRedComm}{R-Comm}%
\newcommand{\iruleMPRedCall}{R-Def}%
\newcommand{\iruleMPRedCtx}{R-Ctx}%

\newcommand{\iruleSubEnd}{Sub-$\stEnd$}%
\newcommand{\iruleSubBranch}{Sub-$\stExtC$}%
\newcommand{\iruleSubSel}{Sub-$\stIntC$}%
\newcommand{\iruleSubRecL}{Sub-$\stFmt{\mu}$L}%
\newcommand{\iruleSubRecR}{Sub-$\stFmt{\mu}$R}%

\newcommand{\iruleMPEnd}{T-$\stEnvEndPred$}%
\newcommand{\iruleMPX}{T-$\mpFmt{X}$}%
\newcommand{\iruleMPSub}{T-Sub}%
\newcommand{\iruleMPNil}{T-$\mpNil$}%
\newcommand{\iruleMPDef}{T-$\mpFmt{\mathbf{def}}$}%
\newcommand{\iruleMPCall}{T-$\mpX$}%
\newcommand{\iruleMPPar}{T-$\mpPar$}%
\newcommand{\iruleMPSafeRes}{T-$\mpFmt{\mathbf{\nu}}$}%
\newcommand{\iruleMPBranch}{T-$\mpFmt{\&}$}%
\newcommand{\iruleMPSel}{T-$\mpFmt{\oplus}$}%
\newcommand{\iruleMPInit}{T-Init}%

\newcommand{\iruleSafeComm}{S-${\stIntC}{\stExtC}$}%
\newcommand{\iruleSafeRec}{S-$\stFmt{\mu}$}%
\newcommand{\iruleSafeMove}{S-$\stEnvMove$}%

\newcommand{\SP}{\mkern2.5mu}

\newcommand\HOLE{{\oeFmt{[\ ]}}}
\newcommand\UNDERSCORE{\mbox{\underline{\hspace{0.7em}}}}
\newcommand\UNL{\mbox{\underline{\hspace{0.7em}}}}

\definecolor{oeColor}{rgb}{0, 0, 0}%
\definecolor{otColor}{rgb}{0, 0, 0.9}%
\newcommand{\oeFmt}[1]{{\color{oeColor}#1}}%
\newcommand{\otFmt}[1]{{\color{otColor}#1}}%

\newcommand{\ocLab}[1][]{%
  \oeFmt{\ifempty{#1}{\color{roleColor}\mathtt{l}}{{\color{roleColor}\mathtt{l}}_{\mathnormal #1}}}%
}%
\newcommand{\oLabFmt}[1]{\oeFmt{#1}}%

\newcommand{\ocUnfold}[1]{%
  {\oeFmt{\operatorname{unfold}^{\ast}\!\left({#1}\right)}}}%

\newcommand{\ocUnfoldOne}[1]{%
  {\oeFmt{\operatorname{unfold}\!\left({#1}\right)}}}%
\newcommand{\ocUnfoldN}[2]{%
  {\oeFmt{\operatorname{unfold}^{\color{black}{#1}}\!\left({#2}\right)}}}%
\newcommand{\OC}{\oeFmt{\mathbb{E}}}

\newcommand{\ocUnit}{\oeFmt{\texttt{()}}}%
\newcommand{\ocUnitPayload}{\oeFmt{\texttt{()}}}%

\newcommand{\ocIntChoice}[3]{{#1}{=}{\ocPair{#2}{#3}}}%
\newcommand{\ocIntChoiceSmall}[3]{{#1}{=}{\ocPairSmall{#2}{#3}}}%
\newcommand{\ocExtChoice}[3]{\ocRecvWrapElem{#1}{#2}{#3}}
\newcommand{\ocExtChoiceSmall}[3]{\ocRecvWrapElemSmall{#1}{#2}{#3}}

\newcommand{\ocIntSum}[3]{\ocRecord{#1}{\left\langle{#3}\right\rangle_{#2}}{}}%
\newcommand{\ocIntSumSmall}[3]{\oeFmt{\langle{#1}{=}\langle{#3}\rangle_{#2}\rangle}}%
\newcommand{\ocExtSum}[3]{\ocRecord{#1}{\left[#3\right]_{#2}}{}}%
\newcommand{\ocExtSumSmall}[3]{\oeFmt{\langle{#1}{=}{{[{#3}]}}_{#2}\rangle}}%

\newcommand{\oeNil}{\oeFmt{\bullet}}%
\newcommand\ocRec[2]{\oeFmt{\mu{#1}{\mathbin{.}}{#2}}}
\newcommand\ocPair[2]{\oeFmt{\left({#1}{,}{#2}\right)}}
\newcommand\ocPairSmall[2]{\oeFmt{({#1}{,}{#2})}}
\newcommand\ocTuple[1]{\oeFmt{\left({#1}\right)}}
\newcommand\ocVariant[2]{\oeFmt{\left[{#1}{=}{#2}\right]}}
\newcommand\ocVariantSmall[2]{\oeFmt{[{#1}{=}{#2}]}}

\newcommand\ocVariantPairPat[3]{\oeFmt{{#1}{({#2}{,}{#3})}}}
\newcommand\ocTupOpen{(}
\newcommand\ocTupClose{)}
\newcommand\ocRecordOpen{\oeFmt{\langle}}
\newcommand\ocRecordClose{\oeFmt{\rangle}}
\newcommand\ocRecordEntry[2]{\oeFmt{{#1}{=}{#2}}}
\newcommand\ocRecordSep{\oeFmt{,}}
\newcommand\ocRecord[3]{\oeFmt{\left\langle{#1}{=}{#2}\right\rangle_{#3}}}
\newcommand\ocRecordSmall[3]{\oeFmt{\langle{#1}{=}{#2}\rangle_{#3}}}
\newcommand\ocWrapperElem[2]{\oeFmt{{#1}{{\mathtt{@}}}{#2}}}
\newcommand\ocWrapper[3]{\oeFmt{\left[\ocWrapperElem{#1}{#2}\right]_{#3}}}
\newcommand{\ocWrapperApp}[2]{{#1}[{#2}]}
\newcommand\ocRecvWrapElem[3]{{#1}{=}{\ocPair{\underline{#2}}{#3}}}
\newcommand\ocRecvWrapElemSmall[3]{{#1}{=}{\ocPairSmall{\underline{#2}}{#3}}}

\newcommand\ocRecvWrapSmall[4]{\ocVariantSmall{#1}{\ocPairSmall{\underline{#2}}{#3}}_{#4}}

\newcommand{\otFix}[2]{\otFmt{{\mathrm{tfix}}{\left({#1},{#2}\right)}}}
\newcommand{\ocFix}[2]{{{\mathrm{fix}}{\left({#1},{#2}\right)}}}
\newcommand{\ocFixSmall}[2]{{{\mathrm{fix}}{({#1},{#2})}}}

\newcommand{\ocMerge}[2]{\oeFmt{\mathlarger{\bigsqcup}_{#1}{#2}}}%
\newcommand{\ocBinMerge}[1][]{
  \ifempty{#1}{
    \mathbin{\oeFmt{\sqcup}}
   }{
    \mathbin{\oeFmt{\sqcup}}_{\otFmt{#1}}
   }
}%

\newcommand{\oeChoice}[4]{%
  \oeFmt{%
    \oLabFmt{#1}({#2},{#3})\triangleright{#4}%
  }%
}%
\newcommand{\oeMatch}[6]{
  \oeFmt{%
  \mathbf{match}\SP#1\SP\mathbf{with}\SP\{\oeChoice{#2}{#3}{#4}{#5}\}_{#6}
  }%
}%
\newcommand{\oeMatchRaw}[2]{
  \oeFmt{%
  \mathbf{match}\SP#1\SP\mathbf{with}\SP\{#2\}
  }%
}%
\newcommand{\oeRecvKwd}{\mathbf{recv}}
\newcommand{\oeWithKwd}{\mathbf{with}}
\newcommand{\oeBranchMult}[7]{
  \oeFmt{%
  \oeMatchKwd\SP\oeRecvKwd\SP#1\#\roleFmt{#2}\SP\SP\oeWithKwd\left\{\oeChoice{#3}{#4}{#5}{#6}\right\}_{#7}
  }%
}%
\newcommand{\oeBranchRaw}[3]{
  \oeFmt{%
  \oeMatchKwd\SP\oeRecvKwd\SP#1\#\roleFmt{#2}\SP\SP\oeWithKwd\left\{#3\right\}
  }%
}%
\newcommand{\oeBranchSingle}[4]{
  \oeFmt{%
    \mathbf{let}\SP{#1}\SP{=}\SP\textbf{recv}\SP\SP#2\#\roleFmt{#3}\SP\SP\mathbf{in}\SP{#4}
  }%
}%
\newcommand{\oeRecv}[3]{
  \oeFmt{%
    \mathbf{let}\SP{#1}\SP{=}\SP\textbf{recv}\SP\SP#2\SP\SP\mathbf{in}\SP{#3}
  }%
}%
\newcommand{\oeRecvBare}[1]{
  \oeFmt{%
    \textbf{recv}\SP\SP#1
  }%
}%
\newcommand{\oeSel}[6]{
  \oeFmt{%
    \mathbf{let}\SP{#1}\SP{=}\SP\textbf{send}\SP\SP#2\#\roleFmt{#3}\#{#4}\SP\SP{#5}\SP\mathbf{in}\SP{#6}
  }%
}%
\newcommand{\oeSend}[4]{
  \oeFmt{%
    \mathbf{let}\SP{#1}\SP{=}\SP\textbf{send}\SP\SP#2\SP\SP{#3}\SP\mathbf{in}\SP{#4}
  }%
}%
\newcommand{\oeSendBare}[2]{
  \oeFmt{%
    \textbf{send}\SP\SP#1\SP\SP{#2}
  }%
}%

\newcommand{\oeInit}[3]{
  \oeFmt{%
    \mathbf{let}\SP{#1}\SP{=}\SP{#2}\SP\mathbf{in}\SP{#3}%
  }%
}%
\newcommand{\oeLetKwd}{\oeFmt{\mathbf{let}}}
\newcommand{\oeLetrecKwd}{\oeFmt{\mathbf{letrec}}}
\newcommand{\oeMatchKwd}{\oeFmt{\mathbf{match}}}
\newcommand{\oeInKwd}{\oeFmt{\mathbf{in}}}
\newcommand{\oePar}{\mathbin{\oeFmt{\mid}}}%
\newcommand{\oeRes}[2]{\oeFmt{(\mathbf{\nu}{#1}){#2}}}%
\newcommand{\oeJustDef}[3]{%
  \oeFmt{{#1}({#2}) = {#3}}%
}%
\newcommand{\oeLetrecAbbrev}[2]{%
  \oeFmt{\mathbf{letrec}\;{#1}\;\mathbf{in}\;{#2}}%
}%
\newcommand{\oeLetrec}[4]{%
  \oeFmt{\mathbf{letrec}\;\oeJustDef{#1}{#2}{#3}\;\mathbf{in}\;{#4}}%
}%

\newcommand{\oeCall}[2]{\oeFmt{{#1}\!\left\langle{#2}\right\rangle}}%
\newcommand{\oeCallDec}[2]{\oeFmt{{#1}\!\left({#2}\right)}}%

\newcommand{\ocChi}[1][]{\color{black}{\ifempty{#1}{\chi}{\chi_{#1}}}}%
\newcommand{\ocEmpty}{\color{black}{\emptyset}}
\newcommand{\ocL}[1][]{\oeFmt{\ifempty{#1}{l}{l_{#1}}}}%
\newcommand{\ocC}[1][]{\oeFmt{\ifempty{#1}{c}{c_{#1}}}}%
\newcommand{\ocV}[1][]{\oeFmt{\ifempty{#1}{v}{v_{#1}}}}%
\newcommand{\ocCi}[1][]{\oeFmt{\ifempty{#1}{c'}{c'_{#1}}}}%
\newcommand{\ocD}[1][]{\oeFmt{\ifempty{#1}{d}{d_{#1}}}}%
\newcommand{\ocX}[1][]{\oeFmt{\ifempty{#1}{x}{x_{#1}}}}%
\newcommand{\ocXi}[1][]{\oeFmt{\ifempty{#1}{x'}{x'_{#1}}}}%
\newcommand{\ocY}[1][]{\oeFmt{\ifempty{#1}{y}{y_{#1}}}}%
\newcommand{\ocZ}[1][]{\oeFmt{\ifempty{#1}{z}{z_{#1}}}}%
\newcommand{\ocH}[1][]{\oeFmt{\ifempty{#1}{h}{h_{#1}}}}%
\newcommand{\ocHi}[1][]{\oeFmt{\ifempty{#1}{h'}{h'_{#1}}}}%
\newcommand{\ocS}[1][]{\oeFmt{\ifempty{#1}{s}{s_{#1}}}}%
\newcommand{\ocSi}[1][]{\oeFmt{\ifempty{#1}{s'}{s'_{#1}}}}%
\newcommand{\oeRecVarBase}{\oeFmt{X}}%
\newcommand{\oeRecVar}[1][]{\oeFmt{\ifempty{#1}{\oeRecVarBase}{\oeRecVarBase_{#1}}}}%
\newcommand{\oeRecVari}[1][]{\oeFmt{\ifempty{#1}{\oeRecVar'}{\oeRecVar'_{#1}}}}%
\newcommand{\oeE}[1][]{\oeFmt{\ifempty{#1}{e}{e_{#1}}}}%
\newcommand{\oeEi}[1][]{\oeFmt{\ifempty{#1}{e'}{e'_{#1}}}}%
\newcommand{\oeEii}[1][]{\oeFmt{\ifempty{#1}{e''}{e''_{#1}}}}%

\newcommand{\oeD}[1][]{\oeFmt{\ifempty{#1}{D}{D_{#1}}}}%
\newcommand{\oeDi}[1][]{\oeFmt{\ifempty{#1}{D'}{D'_{#1}}}}%
\newcommand{\otUnit}{\otFmt{\bullet}}%

\newcommand{\otPair}[2]{\otFmt{{#1}{\times}{#2}}}%
\newcommand{\otTimes}{\mathbin{\otFmt{\scriptstyle\times}}}
\newcommand{\otVariant}[3]{\otFmt{\left[{#1} \UNDERSCORE {#2}\right]_{#3}}}

\newcommand{\otVariantMany}[1]{\otFmt{\left[#1\right]}}
\newcommand{\otVariantElem}[2]{\otFmt{{#1} \UNDERSCORE {#2}}}
\newcommand{\otRecord}[3]{\otFmt{\left\langle{#1}\,{:}\,{#2}\right\rangle_{#3}}}

\newcommand{\otChanTag}{\otFmt{\sharp}}%
\newcommand{\otChan}[1]{\otChanTag{#1}}%
\newcommand{\otInpTag}{\otFmt{?}}%
\newcommand{\otInp}[1]{\otInpTag{#1}}
\newcommand{\otOutTag}{\otFmt{!}}%
\newcommand{\otOut}[1]{\otOutTag{#1}}

\newcommand{\otIntChoice}[3]{{#1}{:}\otFmt{{!{#2}}{\times}{#3}}}%
\newcommand{\otExtChoice}[3]{{#1}{\UNDERSCORE}\otFmt{{#2}{\times}{#3}}}%

\newcommand{\otIntSumBig}[3]{\otFmt{\Bigl\langle{#1}{:}\bigl\langle{#3}\bigr\rangle_{#2}\Bigr\rangle}}%
\newcommand{\otIntSum}[3]{\otFmt{\Bigl\langle{#1}{:}\langle{#3}\rangle_{#2}\Bigr\rangle}}%
\newcommand{\otIntSumSmall}[3]{\otFmt{\langle{#1}{:}\langle{#3}\rangle_{#2}\rangle}}%
\newcommand{\otExtSumBig}[3]{\otFmt{\left\langle{#1}{:}{{{?}\left[{#3}\right]}}_{#2}\right\rangle}}%
\newcommand{\otExtSum}[3]{\otFmt{\Bigl\langle{#1}{:}{{{?}[{#3}]}}_{#2}\Bigr\rangle}}%
\newcommand{\otExtSumSmall}[3]{\otFmt{\langle{#1}{:}{{{?}[{#3}]}}_{#2}\rangle}}%
\newcommand{\otVar}[2]{\otFmt{[{#1}]}_{#2}}%

\newcommand{\otRec}[2]{\otFmt{\mu{#1}{.}{#2}}}%

\newcommand{\otT}[1][]{\otFmt{\ifempty{#1}{T}{T_{#1}}}}%
\newcommand{\otTi}[1][]{\otFmt{\ifempty{#1}{T'}{T'_{#1}}}}%
\newcommand{\otTii}[1][]{\otFmt{\ifempty{#1}{T''}{T''_{#1}}}}%
\newcommand{\otTiii}[1][]{\otFmt{\ifempty{#1}{T'''}{T'''_{#1}}}}%

\newcommand{\otO}{\otFmt{\bullet}}%
\newcommand{\otS}[1][]{\otFmt{\ifempty{#1}{S}{S_{#1}}}}%
\newcommand{\otSi}[1][]{\otFmt{\ifempty{#1}{S'}{S'_{#1}}}}%

\newcommand{\otH}[1][]{\otFmt{\ifempty{#1}{H}{H_{#1}}}}%

\newcommand{\otRecVarBase}{\otFmt{\mathbf{t}}}%
\newcommand{\otRecVar}[1][]{\otFmt{\ifempty{#1}{\otRecVarBase}{\otRecVarBase_{#1}}}}%
\newcommand{\otRecVari}[1][]{\otFmt{\ifempty{#1}{\otRecVar'}{\otRecVar'_{#1}}}}%
\newcommand{\otRolesSet}[1]{\operatorname{roles}(\otFmt{#1})}%

\newcommand{\otSub}{\mathrel{\otFmt{\leqslant}}}%
\newcommand{\otEnv}[1][]{\otFmt{\ifempty{#1}{\Gamma}{\Gamma_{#1}}}}%
\newcommand{\otEnvi}[1][]{\otFmt{\ifempty{#1}{\Gamma'}{\Gamma'_{#1}}}}%
\newcommand{\otEnvii}[1][]{\otFmt{\ifempty{#1}{\Gamma''}{\Gamma''_{#1}}}}%
\newcommand{\otEnvEmpty}{\otFmt{\emptyset}}%
\newcommand{\otEnvMap}[2]{{#1}\mathbin{\!:\!}\otFmt{#2}}%
\newcommand{\otEnvComp}{\mathpunct{\otFmt{,}}}%

\newcommand{\ogtEnvEntails}[3]{%
  \otFmt{#1} \vdash_\roleSet {#2} \mathbin{:} \otFmt{#3}%
}%
\newcommand{\ogtEnvEntailsEx}[4]{%
  \otFmt{#2} \vdash_{#1} {#3} \mathbin{:} \otFmt{#4}%
}%
\newcommand{\otEnvEntails}[3]{%
  \otFmt{#1} \vdash {#2} \mathbin{:} \otFmt{#3}%
}%
\newcommand{\oeEnv}[1][]{\otFmt{\ifempty{#1}{\Theta}{\Theta_{#1}}}}%
\newcommand{\oeEnvi}[1][]{\otFmt{\ifempty{#1}{\Theta'}{\Theta'_{#1}}}}%
\newcommand{\oeEnvEmpty}{\otFmt{\emptyset}}%
\newcommand{\oeEnvMap}[2]{\otFmt{\otFmt{#1}{:}\otFmt{#2}}}%
\newcommand{\oeEnvComp}{\mathpunct{\otFmt{,}}}%

\newcommand{\iruleORedComm}{{\sc Ored-Comm}}%
\newcommand{\iruleORedInit}{{\sc Ored-Init}}%
\newcommand{\iruleORedMatch}{{\sc Ored-Match}}%
\newcommand{\iruleORedRec}{{\sc Ored-Rec}}%
\newcommand{\iruleORedCtx}{{\sc Ored-Ctx}}%
\newcommand{\iruleORedCong}{{\sc Ored-$\equiv$}}%

\newcommand{\otJudge}[3]{%
  \otFmt{\ifempty{#1}{#2}{{#1} \cdot {#2}}%
  \mathrel{\otFmt{\vdash}} \otFmt{#3}}%
}%
\newcommand{\otEnvQMoveModQ}[1]{\mathrel{\otFmt{\rightarrow}_{{#1}}}}%
\newcommand{\otEnvQNotMoveModQP}[2]{%
  {#2}/{\kern -0.5em}{\otEnvQMoveModQ{#1}}}%

\newcommand{\otWrapper}[2]{{#1}[{#2}]}%

\newcommand{\iruleOTCName}{{\sc Otc}-$\ocS$}%
\newcommand{\iruleOTCVar}{{\sc Otc}-$\ocX$}%
\newcommand{\iruleOTCUnit}{{\sc Otc}-$\ocUnit$}%
\newcommand{\iruleOTCSub}{{\sc Otc-Sub}}%
\newcommand{\iruleOTCTuple}{{\sc Otc-Tup}}%
\newcommand{\iruleOTCVariant}{{\sc Otc-Variant}}%
\newcommand{\iruleOTCRecord}{{\sc Otc-Record}}%
\newcommand{\iruleOTCWrapper}{{\sc Otc-Wrapper}}%
\newcommand{\iruleOTCWrapInp}{{\sc Otc-WrapInp}}%
\newcommand{\iruleOTCRec}{{\sc Otc}-$\oeFmt{\mu}$}%

\newcommand{\iruleOTGComm}{{\sc Otg-Comm}}
\newcommand{\iruleOTGChoice}{{\sc Otg-Choice}}
\newcommand{\iruleOTGRec}{{\sc Otg}-$\ogtRecKwd$}
\newcommand{\iruleOTGRecVar}{{\sc Otg}-$\ogtRecVar$}
\newcommand{\iruleOTGEnd}{{\sc Otg}-$\ogtEnd$}
\newcommand{\iruleOTGSub}{{\sc Otg-Sub}}

\newcommand{\iruleOTX}{{\sc Ot}-$\oeRecVar$}%
\newcommand{\iruleOTNil}{{\sc Ot}-$\oeNil$}%
\newcommand{\iruleOTLetrec}{{\sc Ot}-$\oeFmt{\mathbf{letrec}}$}%
\newcommand{\iruleOTCall}{{\sc Ot}-$\oeRecVar$}%
\newcommand{\iruleOTPar}{{\sc Ot}-$\oePar$}%
\newcommand{\iruleOTRes}{{\sc Ot}-$\oeFmt{\mathbf{\nu}}$}%
\newcommand{\iruleOTMatch}{{\sc Ot}-$\oeFmt{\mathbf{match}}$}%
\newcommand{\iruleOTRecv}{{\sc Ot}-$\oeFmt{\mathtt{recv}}$}%
\newcommand{\iruleOTSel}{{\sc Ot}-$\oeFmt{\oplus}$}%
\newcommand{\iruleOTInit}{{\sc Ot}-Init}%

\newcommand{\iruleOSubRcdDepth}{{\sc Osub-RcdDepth}}
\newcommand{\iruleOSubVar}{{\sc Osub-Var}}

\newcommand{\iruleOSubTuple}{{\sc Osub-Tup}}
\newcommand{\iruleOSubOut}{{\sc Osub-Out}}
\newcommand{\iruleOSubInp}{{\sc Osub-Inp}}
\newcommand{\iruleOSubRecL}{{\sc Osub-$\otFmt{\mu}$L}}
\newcommand{\iruleOSubRecR}{{\sc Osub-$\otFmt{\mu}$R}}
\newcommand{\iruleOSubInpCh}{{\sc Osub-InpCh}}
\newcommand{\iruleOSubOutCh}{{\sc Osub-OutCh}}
\newcommand{\iruleOSubUnit}{{\sc Osub-$\otUnit$}}
\definecolor{gtColor}{rgb}{0.43, 0.21, 0.1}%
\newcommand{\gocaml}[1][]{\ogtFmt{\ifempty{#1}{\boldsymbol{\mathtt{g}}}{\boldsymbol{\mathtt{g}}_{#1}}}}%

\newcommand{\ogtFmt}[1]{{\color{gtColor}#1}}%

\newcommand{\ogtG}[1][]{\ogtFmt{\ifempty{#1}{\boldsymbol{\mathtt{g}}}{\boldsymbol{\mathtt{g}}_{#1}}}}%
\newcommand{\ogtTo}{\rightarrow}

\newcommand{\ogtChoiceKwd}{\ogtFmt{\texttt{choice}}}
\newcommand{\ogtChoice}[3]{%
  \ogtFmt{%
    \ogtChoiceKwd\,{#1}\left\{{#2}\right\}_{#3}%
  }}
\newcommand{\ogtComm}[6]{%
\ifempty{#3}{
  \ogtFmt{%
    \texttt{(}{#1}\,\ogtTo\,{#2}\texttt{)}\,\ifempty{#5}{#4}{{#4}{:}{#5}}\ {#6}
  }
}
{
  \ogtFmt{%
    \ogtChoiceKwd\,{#1}\left\{
    \texttt{(}{#1}\,\ogtTo\,{#2}\texttt{)}\,\ifempty{#5}{#4}{{#4}{:}{#5}}\ {#6}\right\}_{#3}%
  }%
}}%
\newcommand{\ogtEnd}{\ogtFmt{\texttt{finish}}}%
\newcommand{\ogtCommSingle}[5]{%
  \ogtFmt{%
    \texttt{(}{#1}\,\ogtTo\,{#2}\texttt{)}\,\ifempty{#4}{#3}{{#3}{:}{#4}}\ {#5}
  }
}
\newcommand{\ogtCommBin}[8]{%
  \ogtFmt{%
    \ogtChoiceKwd\,{#1}\left\{\ogtCommSingle{#1}{#2}{#3}{#4}{#5}{\texttt{,}\,} \ogtCommSingle{#1}{#2}{#6}{#7}{#8}\right\}
  }
}
\newcommand{\ogtChoiceRaw}[2]{%
  \ogtFmt{%
    \ogtChoiceKwd\,{#1}\,\left\{{#2}\right\}
  }
}

\newcommand{\ogtRecKwd}{\ogtFmt{\texttt{fix}}}
\newcommand{\ogtRec}[2]{\ogtFmt{\ogtRecKwd\,{#1}\,\texttt{-{}>}\,{#2}}}%
\newcommand{\ogtRecVarBase}{\gtFmt{x}}%
\newcommand{\ogtRecVar}[1][]{\gtFmt{\ifempty{#1}{\ogtRecVarBase}{\ogtRecVarBase_{#1}}}}%
\newcommand{\ogtRecVari}[1][]{\gtFmt{\ifempty{#1}{\ogtRecVarBase'}{\ogtRecVarBase'_{#1}}}}%

\newcommand\GCCV[1]{\mpFmt{\llbracket}{#1}\mpFmt{\rrbracket}^{\mpS}_{\roleSet}}
\newcommand\GCCVR[2]{\mpFmt{\llbracket}{#1}\mpFmt{\rrbracket}_{\mpFmt{#2}}}

\newcommand{\NTH}[2]{#1{\mpFmt{({#2})}}}

\newcommand{\BASIC}[2]{\mathrm{Basic}({#1},{#2})}

\newcommand{\ffv}[1]{\operatorname{ffv}\!\left({#1}\right)}%
\newcommand{\dfv}[1]{\operatorname{dfv}\!\left({#1}\right)}%

\newcommand{\gAuth}{\ogtG[{\mathrm{Auth}}]}
\newcommand{\eAuth}{\oeE[{\mathrm{Auth}}]}

\newcommand{\blueI}{{\color{blue}{i}}}
\newcommand{\blueIi}{{\color{blue}{i'}}}

\newcommand{\blueJi}{{\color{blue}{j'}}}
\newcommand{\must}{\ourlibrary}
\newcommand{\mustL}{{\sf{MiO}}\xspace}
\newcommand{\globalcombinators}{global combinators}

\definecolor{dkblue}{rgb}{0,0.1,0.5}
\definecolor{dkgreen}{rgb}{0,0.4,0.1}
\definecolor{dkred}{rgb}{0.4,0,0}

\definecolor{linkColor}{rgb}{0,0,0.5}
\definecolor{pblue}{rgb}{0.13,0.13,1}
\definecolor{purple}{rgb}{0.5,0,0.5}
\definecolor{beige}{rgb}{0,0.5,0.5}
\definecolor{pred}{rgb}{0.9,0,0}
\definecolor{WhiteSmoke}{rgb}{0.96, 0.96, 0.96}
\definecolor{mygray}{gray}{0.6}

\newcommand{\CODESIZETINY}{\footnotesize}
\newcommand{\CODESIZE}{\small}  %
\newcommand{\CODESIZENORMAL}{\small}  %
\newcommand{\LSTCODESIZE}{\fontsize{8pt}{9.5pt}}  %
\newcommand{\CODESTYLE}{\ttfamily}
\newcommand{\CODE}[1]{{\CODESIZE\CODESTYLE #1}}

\mdfdefinestyle{cdescription}{%
  hidealllines=true,%
  skipbelow= \baselineskip,
  roundcorner= 2pt,
  splittopskip=\topskip,skipbelow=\baselineskip,%
   skipabove=\baselineskip
}

\mdfdefinestyle{adescription}{%
  hidealllines=true,%
  backgroundcolor= red!15,
   innermargin =-1cm,
   outermargin =-1cm,
  skipbelow= \baselineskip,
  roundcorner= 2pt,
  splittopskip=\topskip,skipbelow=\baselineskip,%
   skipabove=\baselineskip
}

\surroundwithmdframed[style=adescription]{adescription}

\lstdefinelanguage{SOCaml}{
    morekeywords= {val, let, new, lazy, match, with, rec, open, module, namespace, type, of, member, and, for, while, true, false, in, do, fun, return, yield, try, mutable, if, then, else, cloud, async, static, use, abstract, interface, inherit, finally, Thread, begin, end},
	otherkeywords={ let!, return!, do!, yield!, use!, var, from, select, where, order, by, match, goto},   sensitive=false,
	morekeywords = [3]{choice_at, finish, fix, out, inp, dlin, channel},
    keywordstyle = [3]\color{purple},
    morekeywords = [4]{c, a, s},
    morecomment=[s][\color{dkgreen}]{{(*}{*)}},
    moredelim=[is][\color{gtColor}]{--!}{!--},
    moredelim=[is][\color{otColor}]{--?}{?--},
    moredelim=**[is][\btHL]{£}{£},
    moredelim=**[is][\bfseries\color{dkred}]{&}{&},
    moredelim=**[is][\btHL]{~}{~},
    moredelim=**[is][\berrHL]{€}{€},
    morestring=[b]",
    stringstyle=\color{dkred}
}

\lstdefinestyle{BASE}{
  mathescape=false,
  keywordstyle=\CODESIZENORMAL\CODESTYLE\color{dkblue}, %
  alsoletter={},
  basicstyle={},
  emph={},
  emphstyle={},
  breaklines=false,
  literate={},
}

\lstdefinestyle{SCRIBBASE}%
{%
	style=BASE,
	language=Java,
	morekeywords=[1]{
		protocol, role, choice, at, or, from, to, rec, par, and, interrupt, by, finish, continue, global, local, self, interruptible, with, as, catches, type, sig, aux, explicit, connect, wrap, disconnect, accept, wrapClient, wrapServer},
  moredelim=**[is][{\color{red}\lstset{style=ERROR}}]{@!}{@},
  moredelim=**[is][\only<1>{\color{red}\lstset{style=ERROR}}]{@!1}{@},
  moredelim=**[is][\only<2>{\color{red}\lstset{style=ERROR}}]{@!2}{@},
  moredelim=**[is][\only<3>{\color{red}\lstset{style=ERROR}}]{@!3}{@},
  moredelim=**[is][\only<4>{\color{red}\lstset{style=ERROR}}]{@!4}{@},
  moredelim=**[is][\only<5>{\color{red}\lstset{style=ERROR}}]{@!5}{@},
  moredelim=**[is][{\color{dkgreen}\lstset{style=ASSIST}}]{@?}{@},
  moredelim=**[is][\only<1>{\color{dkgreen}\lstset{style=ASSIST}}]{@?1}{@},
  moredelim=**[is][\only<2>{\color{dkgreen}\lstset{style=ASSIST}}]{@?2}{@},
  moredelim=**[is][\only<3>{\color{dkgreen}\lstset{style=ASSIST}}]{@?3}{@},
  moredelim=**[is][\only<4>{\color{dkgreen}\lstset{style=ASSIST}}]{@?4}{@},
  moredelim=**[is][\only<5>{\color{dkgreen}\lstset{style=ASSIST}}]{@?5}{@},
  moredelim=**[is][\only<-2>{\color{red}\lstset{style=ERROR}}]{@!-2}{@},
  moredelim=**[is][\only<2->{\color{red}\lstset{style=ERROR}}]{@!2-}{@}
}

\lstdefinestyle{SCRIBBG}%
{%
	style=SCRIBBASE,
	basicstyle=\CODESTYLE\CODESIZE\color{black!40},
	keywordstyle=[1]{\color{dkblue!40}},
	keywordstyle=[2]{\color{dkgreen!40}},
  commentstyle=\itshape\color{purple!40},
	identifierstyle=\color{black!40},
	stringstyle=\color{teal!40},
 	emphstyle=\color{dkblue!40}
}

\lstdefinestyle{LSTSCRIBBG}%
{%
	style=SCRIBBG,
	basicstyle=\LSTCODESIZE\CODESTYLE\color{black!40},
  moredelim=**[is][\color{black}\lstset{style=SCRIB}]{@-}{@},
  moredelim=**[is][\only<1>{\color{black}\lstset{style=LSTSCRIB}}]{@1}{@},
  moredelim=**[is][\only<2>{\color{black}\lstset{style=LSTSCRIB}}]{@2}{@},
  moredelim=**[is][\only<3>{\color{black}\lstset{style=LSTSCRIB}}]{@3}{@},
  moredelim=**[is][\only<4>{\color{black}\lstset{style=LSTSCRIB}}]{@4}{@},
  moredelim=**[is][\only<5>{\color{black}\lstset{style=LSTSCRIB}}]{@5}{@},
  moredelim=**[is][\only<6>{\color{black}\lstset{style=LSTSCRIB}}]{@6}{@},
  moredelim=**[is][\only<7>{\color{black}\lstset{style=LSTSCRIB}}]{@7}{@},
  moredelim=**[is][\only<8>{\color{black}\lstset{style=LSTSCRIB}}]{@8}{@},
  moredelim=**[is][\only<9>{\color{black}\lstset{style=LSTSCRIB}}]{@9}{@},
  moredelim=**[is][\only<10>{\color{black}\lstset{style=LSTSCRIB}}]{@10}{@},
  moredelim=**[is][\only<1->{\color{black}\lstset{style=LSTSCRIB}}]{@1-}{@},
  moredelim=**[is][\only<2->{\color{black}\lstset{style=LSTSCRIB}}]{@2-}{@},
  moredelim=**[is][\only<3->{\color{black}\lstset{style=LSTSCRIB}}]{@3-}{@},
  moredelim=**[is][\only<4->{\color{black}\lstset{style=LSTSCRIB}}]{@4-}{@},
  moredelim=**[is][\only<-2>{\color{black}\lstset{style=LSTSCRIB}}]{@-2}{@}
}

\lstdefinestyle{SCRIB}%
{%
	style=SCRIBBASE,
	basicstyle=\CODESTYLE\CODESIZE\color{black},
	keywordstyle=[1]{\color{dkblue}},
	keywordstyle=[2]{\color{dkgreen}},
  commentstyle=\itshape\color{purple},
	identifierstyle=\color{black},
	stringstyle=\color{teal},
 	emphstyle=\color{dkblue}
}

\lstdefinestyle{LSTSCRIB}%
{%
	style=SCRIB,
	basicstyle=\LSTCODESIZE\CODESTYLE\color{black},
  moredelim=**[is][\only<1>{\color{black!40}\lstset{style=LSTSCRIBBG}}]{@1}{@},
  moredelim=**[is][\only<2>{\color{black!40}\lstset{style=LSTSCRIBBG}}]{@2}{@},
  moredelim=**[is][\only<3>{\color{black!40}\lstset{style=LSTSCRIBBG}}]{@3}{@},
  moredelim=**[is][\only<4>{\color{black!40}\lstset{style=LSTSCRIBBG}}]{@4}{@},
  moredelim=**[is][\only<5>{\color{black!40}\lstset{style=LSTSCRIBBG}}]{@5}{@},
  moredelim=**[is][\only<6>{\color{black!40}\lstset{style=LSTSCRIBBG}}]{@6}{@},
  moredelim=**[is][\only<7>{\color{black!40}\lstset{style=LSTSCRIBBG}}]{@7}{@},
  moredelim=**[is][\only<8>{\color{black!40}\lstset{style=LSTSCRIBBG}}]{@8}{@},
  moredelim=**[is][\only<9>{\color{black!40}\lstset{style=LSTSCRIBBG}}]{@9}{@},
  moredelim=**[is][\only<-2>{\color{black!40}\lstset{style=LSTSCRIBBG}}]{@-2}{@},
  moredelim=**[is][\only<1->{\color{black!40}\lstset{style=LSTSCRIBBG}}]{@1-}{@},
  moredelim=**[is][\only<2->{\color{black!40}\lstset{style=LSTSCRIBBG}}]{@2-}{@},
  moredelim=**[is][\only<3->{\color{black!40}\lstset{style=LSTSCRIBBG}}]{@3-}{@},
  moredelim=**[is][\only<4->{\color{black!40}\lstset{style=LSTSCRIBBG}}]{@4-}{@},
  moredelim=**[is][\only<1-2>{\color{black!40}\lstset{style=LSTSCRIBBG}}]{@1-2}{@}
}

\lstdefinestyle{ERROR}%
{%
	style=SCRIBBASE,
	basicstyle=\CODESTYLE\CODESIZE\color{red},
	keywordstyle=[1]{\color{red}},
	keywordstyle=[2]{\color{red}},
	identifierstyle=\color{red}
}

\lstdefinestyle{ASSIST}%
{%
	style=SCRIBBASE,
	basicstyle=\CODESTYLE\CODESIZE\color{green},
	keywordstyle=[1]{\color{dkgreen}},
	keywordstyle=[2]{\color{dkgreen}},
	identifierstyle=\color{dkgreen}
}

\lstdefinestyle{CONST}%
{%
	style=SCRIBAPIBASE,
	basicstyle=\CODESTYLE\CODESIZE\color{dkblue},
	keywordstyle=[1]{\color{dkblue}},
	keywordstyle=[2]{\color{dkblue}},
	identifierstyle=\color{dkblue}
}

\lstdefinestyle{SCRIBAPIBASE}%
{%
	language=Java,
	morekeywords=[1]{},
	morekeywords=[2]{connect, accept, send, receive, enterScope, start, end, invite, async, sync, branch, getOp},
	morekeywords=[3]{val},
	basicstyle=\CODESTYLE\CODESIZE\color{black},
	keywordstyle=[1]{\bfseries\color{purple}},
	keywordstyle=[2]{\color{purple}},
	keywordstyle=[3]{\color{dkblue}},
	keywordstyle=[4]{\bfseries\itshape\color{dkblue}},
  commentstyle=\itshape\color{dkgreen},
  moredelim=**[is][\bfseries\itshape\color{dkblue}]{@=}{@},  %
	moredelim=**[is][\only<1>{\bfseries\itshape\color{dkblue}}]{@=1}{@},
  moredelim=**[is][\only<2>{\bfseries\itshape\color{dkblue}}]{@=2}{@},
  moredelim=**[is][\only<3>{\bfseries\itshape\color{dkblue}}]{@=3}{@},
  moredelim=**[is][\only<4>{\bfseries\itshape\color{dkblue}}]{@=4}{@},
  moredelim=**[is][\only<5>{\bfseries\itshape\color{dkblue}}]{@=5}{@},
	moredelim=**[is][\only<1-2>{\bfseries\itshape\color{dkblue}}]{@=1-2}{@},
	moredelim=**[is][\only<4->{\bfseries\itshape\color{dkblue}}]{@=4-}{@},
  moredelim=**[is][\only<1>{\color{red}\lstset{style=ERROR}}]{@!1}{@},
  moredelim=**[is][\only<2>{\color{red}\lstset{style=ERROR}}]{@!2}{@},
  moredelim=**[is][\only<3>{\color{red}\lstset{style=ERROR}}]{@!3}{@},
  moredelim=**[is][\only<4>{\color{red}\lstset{style=ERROR}}]{@!4}{@},
  moredelim=**[is][\only<5>{\color{red}\lstset{style=ERROR}}]{@!5}{@},
  moredelim=**[is][{\color{dkgreen}\lstset{style=ASSIST}}]{@?}{@},
  moredelim=**[is][\only<1>{\color{dkgreen}\lstset{style=ASSIST}}]{@?1}{@},
  moredelim=**[is][\only<2>{\color{dkgreen}\lstset{style=ASSIST}}]{@?2}{@},
  moredelim=**[is][\only<3>{\color{dkgreen}\lstset{style=ASSIST}}]{@?3}{@},
  moredelim=**[is][\only<4>{\color{dkgreen}\lstset{style=ASSIST}}]{@?4}{@},
  moredelim=**[is][\only<5>{\color{dkgreen}\lstset{style=ASSIST}}]{@?5}{@}
}

\lstdefinestyle{SCRIBAPIBG}%
{%
	style=SCRIBAPIBASE,
	basicstyle=\CODESTYLE\CODESIZE\color{black!40},
  keywordstyle=\color{red!40},
	keywordstyle=[1]{\bfseries\color{purple!40}},
	keywordstyle=[2]{\color{purple!40}},
	keywordstyle=[3]{\color{dkblue!40}},
	keywordstyle=[4]{\bfseries\itshape\color{dkblue!40}},
  commentstyle=\color{dkgreen!40},
	identifierstyle=\color{black!40},
	stringstyle=\color{teal!40},
 	emphstyle=\color{dkblue!40}
}

\lstdefinestyle{LSTSCRIBAPIBG}%
{%
	style=SCRIBAPIBG,
	basicstyle=\LSTCODESIZE\CODESTYLE\color{black!40},
  moredelim=**[is][\color{black}\lstset{style=LSTSCRIBAPI}]{@-}{@},
  moredelim=**[is][\only<1>{\color{black}\lstset{style=LSTSCRIBAPI}}]{@1}{@},
  moredelim=**[is][\only<2>{\color{black}\lstset{style=LSTSCRIBAPI}}]{@2}{@},
  moredelim=**[is][\only<3>{\color{black}\lstset{style=LSTSCRIBAPI}}]{@3}{@},
  moredelim=**[is][\only<4>{\color{black}\lstset{style=LSTSCRIBAPI}}]{@4}{@},
  moredelim=**[is][\only<5>{\color{black}\lstset{style=LSTSCRIBAPI}}]{@5}{@},
  moredelim=**[is][\only<6>{\color{black}\lstset{style=LSTSCRIBAPI}}]{@6}{@},
  moredelim=**[is][\only<7>{\color{black}\lstset{style=LSTSCRIBAPI}}]{@7}{@},
  moredelim=**[is][\only<8>{\color{black}\lstset{style=LSTSCRIBAPI}}]{@8}{@},
  moredelim=**[is][\only<9>{\color{black}\lstset{style=LSTSCRIBAPI}}]{@9}{@},
  moredelim=**[is][\only<1->{\color{black}\lstset{style=LSTSCRIBAPI}}]{@1-}{@},
  moredelim=**[is][\only<2->{\color{black}\lstset{style=LSTSCRIBAPI}}]{@2-}{@},
  moredelim=**[is][\only<3->{\color{black}\lstset{style=LSTSCRIBAPI}}]{@3-}{@},
  moredelim=**[is][\only<4->{\color{black}\lstset{style=LSTSCRIBAPI}}]{@4-}{@},
  moredelim=**[is][\only<1-2>{\color{black}\lstset{style=LSTSCRIBAPI}}]{@1-2}{@},  %
  moredelim=**[is][\only<2-3>{\color{black}\lstset{style=LSTSCRIBAPI}}]{@2-3}{@},
  moredelim=**[is][\only<3-4>{\color{black}\lstset{style=LSTSCRIBAPI}}]{@3-4}{@}
}

\lstdefinestyle{SCRIBAPI}%
{%
	style=SCRIBAPIBASE,
	identifierstyle=\color{black},
	stringstyle=\color{teal},
 	emphstyle=\color{dkblue}
}

\lstdefinestyle{LSTSCRIBAPI}%
{%
	style=SCRIBAPI,
	basicstyle=\LSTCODESIZE\CODESTYLE\color{black},
  moredelim=**[is][\color{black!40}\lstset{style=LSTSCRIBAPIBG}]{@1}{@},
  moredelim=**[is][\only<2>{\color{black!40}\lstset{style=LSTSCRIBAPIBG}}]{@2}{@},
  moredelim=**[is][\only<3>{\color{black!40}\lstset{style=LSTSCRIBAPIBG}}]{@3}{@},
  moredelim=**[is][\only<4>{\color{black!40}\lstset{style=LSTSCRIBAPIBG}}]{@4}{@},
  moredelim=**[is][\only<5>{\color{black!40}\lstset{style=LSTSCRIBAPIBG}}]{@5}{@},
  moredelim=**[is][\only<6>{\color{black!40}\lstset{style=LSTSCRIBAPIBG}}]{@6}{@},
  moredelim=**[is][\only<7>{\color{black!40}\lstset{style=LSTSCRIBAPIBG}}]{@7}{@},
  moredelim=**[is][\only<8>{\color{black!40}\lstset{style=LSTSCRIBAPIBG}}]{@8}{@},
  moredelim=**[is][\only<9>{\color{black!40}\lstset{style=LSTSCRIBAPIBG}}]{@9}{@},
  moredelim=**[is][\only<1->{\color{black!40}\lstset{style=LSTSCRIBAPIBG}}]{@1-}{@},
  moredelim=**[is][{\color{dkgreen}\lstset{style=ASSIST}}]{@?}{@},
}

\lstdefinestyle{PYTHONAPI}%
{
	language=python,
	showstringspaces=false,
	formfeed=\newpage,
	tabsize=2,
	basicstyle=\CODESTYLE\CODESIZE,
	keywordstyle=\CODESTYLE\CODESIZE,
	commentstyle=\color{purple}\CODESTYLE\CODESIZE,
	stringstyle=\color{teal}\CODESTYLE\CODESIZE,
 	emphstyle=\color{dkblue}\bfseries\CODESIZE,
	morekeywords={models, lambda, forms, def, class}
	keywordstyle=\color{dkblue},
	emph={%
		access,and,as,break,class,continue,def,del,elif,else,%
		except,exec,finally,for,from,global,if,import,in,is,%
		lambda,not,or,pass,print,raise,return,try,while,assert,with%
	}
}
\makeatletter

\newenvironment{btHighlight}[1][]
{\begingroup\tikzset{bt@Highlight@par/.style={#1}}\begin{lrbox}{\@tempboxa}}
{\end{lrbox}\bt@HL@box[bt@Highlight@par]{\@tempboxa}\endgroup}

\newcommand\btHL[1][]{%
  \begin{btHighlight}[#1]\bgroup\aftergroup\bt@HL@endenv%
}

\def\bt@HL@endenv{%
  \end{btHighlight}%
  \egroup
}
\newcommand{\bt@HL@box}[2][]{%
  \tikz[#1]{%
    \pgfpathrectangle{\pgfpoint{1pt}{0pt}}{\pgfpoint{\wd #2}{\ht #2}}%
    \pgfusepath{use as bounding box}%
    \node[anchor=base west, fill=orange!30,outer sep=0pt,inner xsep=1pt, inner ysep=0pt, rounded corners=3pt, minimum height=\ht\strutbox+1pt,#1]{\raisebox{1pt}{\strut}\strut\usebox{#2}};
  }%
}

\makeatother

\makeatletter

\newenvironment{berrHighlight}[1][]
{\begingroup\tikzset{berr@Highlight@par/.style={#1}}\begin{lrbox}{\@tempboxa}}
{\end{lrbox}\berr@HL@box[berr@Highlight@par]{\@tempboxa}\endgroup}

\newcommand\berrHL[1][]{%
  \begin{berrHighlight}[#1]\bgroup\aftergroup\berr@HL@endenv%
}

\def\berr@HL@endenv{%
  \end{berrHighlight}%
  \egroup
}

\newcommand{\berr@HL@box}[2][]{%
  \tikz[#1]{%
    \pgfpathrectangle{\pgfpoint{1pt}{0pt}}{\pgfpoint{\wd #2}{\ht #2}}%
    \pgfusepath{use as bounding box}%
    \node[anchor=base west, draw=red, fill = WhiteSmoke,,outer sep=0pt,inner xsep=1pt, inner ysep=0pt, rounded corners=3pt, minimum height=\ht\strutbox+1pt,#1]{\raisebox{1pt}{\strut}\strut\usebox{#2}};
  }%
}

\makeatother

\lstdefinelanguage{Scribble}{
  keywords={int},
  morecomment=[l]{//},
  morecomment=[s]{/*}{*/},
  morecomment=[s]{@"}{"},
  morestring=[b]",
  tabsize=2
}

\lstnewenvironment{SCRIBBLELISTING}%
{
\lstset{
  language=Scribble,
  xleftmargin= 2.5em,
  xrightmargin= 2.5em,
  framexleftmargin = 0.5em,
  numbers=left,
  showstringspaces=false,
  formfeed=\newpage,
  tabsize=2,
  commentstyle=\color{orange}\itshape,
  basicstyle= \SCRMATHFONT,
  keywordstyle=\color{purple}\itshape,
  emph={global, protocol, role, par, and, from, to, choice, at, or, rec, continue},
  morekeywords = [2]{A,B,C,D,E, S, P},
  keywordstyle = [2]\color{beige}\itshape,
  emphstyle=\color{dkblue}\bfseries,
  basicstyle=\CODESTYLE\scriptsize,
  escapeinside={*@}{@*},
  moredelim=**[is][\btHL]{`}{`},
}
}
{
}

\lstdefinestyle{myocaml}{
  style=BASE,
  language=SOCaml,
  tabsize=2,
  numbersep=1mm,
  keywordstyle = \color{dkblue},
  basicstyle=\CODESIZETINY\CODESTYLE,
  morekeywords = [3]{choice_at, finish},
  keywordstyle = [3]\color{purple},
  keywordstyle = [4]\color{dkred},
  commentstyle = \color{dkgreen}\itshape,
  alsoletter = {:,->,(,), _ },
  escapeinside={^}{^},
  mathescape = true,
  literate=
    {-->}{{\color{purple}-{}->}}1
    {:?}{$\color{blue}{:? \quad}$}1
    {<}{<}1
    {\%}{{\color{dkred}\%}}1
     {@@}{{\color{dkblue}@@}}1
    {<@}{{\color{blue}<@} \quad}1
    {@>}{\quad{\color{blue}@}>}1
    {>}{>}1
    {._}{{$\,$}}1
    {,,}{{$\,\elipc\,$}}1
    {**}{{,}}1
    {???}{{\%}}1
}

\lstnewenvironment{OCAMLLISTING}%
{\lstset{style=myocaml}}
{}

\lstnewenvironment{OCAMLLISTINGTINY}%
{\lstset{style=myocaml,basicstyle=\scriptsize\CODESTYLE}}
{}

\lstnewenvironment{BIGOCAMLLISTING}%
{\lstset{style=tCODEstyle}}
{}

\lstdefinestyle{tCODEstyle}{
    style=BASE,
    language = SOCaml,
    alsoletter = {_},
    alsoletter = {'},
    basicstyle=\CODESIZE\CODESTYLE,
    emph={string,str,int, real, bool, unit},
    morekeywords = [2]{dynlin, lin, data},
    keywordstyle = [2]\color{gray},
    emphstyle=\color{dkred},
    keywordstyle=\CODESIZE\CODESTYLE\color{dkred},
    mathescape=true,
    breaklines=true,
    literate=
    {-->}{{\color{purple}-{}->}}1
    {:?}{$\color{dkblue}{:? }$}1
    {<}{<}1
     {(}{(}1
    {\%}{{\color{blue}\%}}1
     {@@}{{\color{dkblue}@@}}1
    {<@}{{\color{blue}<@} \quad}1
    {@>}{\quad{\color{blue}@}>}1
    {>}{>}1
    {._}{{$\,$}}1
}

\definecolor{mygray}{RGB}{245,245,245}
\lstnewenvironment{OCAMLLISTINGDARK}%
{
\lstset{
  language=SOCaml,
  alsoletter ={_},
  alsoletter ={.},
  tabsize=2,
  numbersep=1mm,
  keywordstyle = \color{dkblue},
  basicstyle=\tiny\CODESTYLE,
  morekeywords = [3]{choice_at, finish},
  keywordstyle = [3]\color{purple},
  keywordstyle = [4]\color{dkred},
  commentstyle = \color{dkgreen}\itshape,
  alsoletter = {:,->,.,(,)},
  escapeinside={^}{^},
  mathescape = true,
  moredelim=**[is][\btHL]{~}{~},
  literate=
    {-->}{{\color{purple}-{}->}}1
    {:?}{$\color{blue}{:? \quad}$}1
    {<}{<}1
     {(}{(}1
    {\%}{{\color{blue}\%}}1
     {@@}{{\color{dkblue}@@}}1
    {<@}{{\color{blue}<@} \quad}1
    {@>}{\quad{\color{blue}@}>}1
    {>}{>}1
    {._}{{$\,$}}1
    {,,}{{$\,\elipc\,$}}1
    {**}{{,}}1
    {_1}{$_\textsf{\scriptsize 1}$}1
}
}
{
}

\lstdefinestyle{oCODEstyle}{
    style=BASE,
    language = SOCaml,
    alsoletter = {_},
    basicstyle=\CODESIZENORMAL\CODESTYLE,
    emph={string,str,int, real, bool, unit},
    emphstyle=\color{dkred},
    keywordstyle=\CODESIZENORMAL\CODESTYLE\color{dkblue},
    mathescape=true,
    breaklines=true,
    literate=
    {-->}{{\color{purple}-{}->}}1
    {:?}{$\color{dkblue}{:? }$}1
    {<}{<}1
     {(}{(}1
    {\%}{{\color{blue}\%}}1
    {@@}{{\color{dkblue}@@}}1
    {<@}{{\color{blue}<@} \quad}1
    {@>}{\quad{\color{blue}@}>}1
    {>}{>}1
    {._}{{$\,$}}1
    {,,}{{$\,\elipc\,$}}1
    {**}{{,}}1
}
\newcommand{\oCODEEsc}[1]{\lstinline[style=oCODEstyle]!#1!}
\newcommand{\oCODE}{\lstinline[style=oCODEstyle]}

\mdfdefinestyle{myframedcodestyle}{
  skipabove=4pt,
  skipbelow=2pt,
  leftmargin=0pt,
  rightmargin=0pt,%
  innertopmargin=0pt,%
  innerbottommargin=0pt,
  innerleftmargin=-1em,
  innerrightmargin=0pt,
}

\lstdefinestyle{myocamltight}{
    style=myocaml,
    xleftmargin=2em,
    xrightmargin=2em,
    belowskip=2pt,
    aboveskip=4pt,
    framexleftmargin=0pt,
    framexrightmargin=0pt,
    framextopmargin=0pt,
    framexbottommargin=0pt,
    resetmargins=true,
  }

\lstnewenvironment{LISTING}%
{
  \lstset{
    style=myocamltight,
    xleftmargin=1em,
    xrightmargin=1em,
    belowskip=0pt,
  }
}{}

\lstnewenvironment{LISTINGNORMAL}%
{
  \lstset{
    style=myocamltight,
    xleftmargin=1em,
    xrightmargin=1em,
    belowskip=0pt,
    basicstyle=\CODESIZENORMAL\CODESTYLE,
    keywordstyle=\CODESIZENORMAL\CODESTYLE\color{dkblue},
  }
}{}

\lstnewenvironment{LISTINGBOX}%
{
  \lstset{
    style=myocamltight,
    frame=single,
  }
}{}

\newcommand{\mrg}{\mathit{mrg}}

\DeclareMathSizes{10}{10}{8}{6}

\begin{document}

\maketitle
\startcontents[sections]
\begin{abstract}
Multiparty Session Types (MPST) is a typing discipline for
communication protocols. It ensures the absence of communication
errors and deadlocks for well-typed communicating processes. The state-of-the-art
implementations of the MPST theory rely on (1) \emph{runtime
 linearity checks} to ensure correct usage of communication
channels and (2) external domain-specific languages for
specifying and verifying multiparty protocols.

To overcome these limitations, we propose a library for programming with \textit{global combinators} -- a set of functions for writing and verifying multiparty protocols in \OCaml. 
Local behaviours for \textit{all} processes in a protocol are inferred
\textit{at once} from a global combinator.
We formalise global combinators 
and prove a sound realisability of global combinators -- a 
well-typed global combinator derives a set of local
types, by which typed endpoint programs can ensure
type and communication safety. Our approach
enables fully-static verification and implementation 
of the whole protocol, from
the protocol specification to the process implementations, to happen
in the same language. 

We compare our implementation to 
untyped and continuation-passing style implementations, and 
demonstrate its expressiveness by implementing 
a plethora of protocols.   
We show our library can interoperate with existing libraries and services, 
implementing DNS (Domain Name Service) protocol and the OAuth (Open Authentication) protocol.
\end{abstract}

\section{Introduction}
\label{sec:intro}
\textbf{Multiparty Session Types.} %
Multiparty Session Types (MPST) \cite{HYC08,MPST,DBLP:journals/jacm/HondaYC16} is a theoretical framework that stipulates how to write, verify and ensure correct implementations of communication protocols.
The methodology of programming with MPST (depicted in \cfig{fig:overview}(a)) starts from a  communication protocol (a \emph{global type}) which specifies the behaviour of a system of interacting processes. %
The local behaviour (a \emph{local type})
for each endpoint process is then algorithmically {\em projected} from the
protocol. Finally, each endpoint process is implemented in an endpoint host language and type-checked against its respective local type by a session
typing system. The guarantee of session types is that a system of  well-typed endpoint
processes \textit{does not go wrong}, i.e it does not exhibit
communication errors such as reception errors, orphan messages or deadlocks, and satisfies
session fidelity, i.e. the local behaviour of each process follows the global specification.

The theoretical MPST framework ensures desirable safety properties. In practice, 
session types implementations that enforce these properties \textit{statically}, i.e at compile-time,  are limited to binary (two party protocols)
\cite{pucella08session,orchard16effects,lindley16embedding,padovani17context}. 
Extending binary session types
implementations
to multiparty interactions, which support static linearity checks
(i.e., linear usage of channels),
is non-trivial, and poses two implementation challenges. 

\textbf{(C1) How global types can be specified and verified in a general-purpose programming language?} Checking compatibility of two communicating processes relies on \textit{duality}, i.e., when one process performs an action, the other performs a complementary (dual) action. Checking the compatibility of multiple processes is more complicated, and 
relies on the existence of a \textit{well-formed} global protocol and  the syntax-directed procedure of \textit{projection}, 
which derives local types from a global specification. A global protocol is considered \textit{well-formed}, 
if local types can be derived via projection. 
Since global types are far from the types of a “mainstream” programming language, 
state-of-the-art MPST implementations
\cite{HY2016,DBLP:conf/cc/NeykovaHYA18,scalas17linear,CHJNY2019} use
external domain-specific protocol description languages and tools 
(e.g. the Scribble toolchain \cite{scribble})
to specify global types and to implement the verification procedure of projection.  
The usage of external tools for protocol description and verification 
widens the gap between the specification and 
its implementations and makes it more difficult to locate protocol violations
in the program, i.e. the correspondence between an error in the program and
the protocol is less apparent.
\textbf{(C2) How to implement safe multiparty communication over binary channels?} 
The theory of MPST requires processes to communicate over multiparty channels – channels that carry messages between two or more parties; their types stipulate the 
precise sequencing of the communication between multiple processes. Additionally, multiparty channels has to be used linearly, i.e exactly once. 
In practice, however, (1) communication channels are binary, i.e a TCP socket for example connects only two parties, and hence its type can describe interactions between two entities only; (2) most languages do not support typing of linear resources.  Existing MPST implementations 
\cite{HY2016,DBLP:conf/cc/NeykovaHYA18,scalas17linear,CHJNY2019}  apply two workarounds. To preserve the order of interactions  when implementing a multiparty protocol over binary channels, existing works use code generation (e.g. \cite{scribble}) and generate local types (APIs) for several (nominal) programming languages. Note that although the interactions order is preserved, most of these implementations \cite{HY2016,DBLP:conf/cc/NeykovaHYA18,CHJNY2019} still require type-casts on the underlying channels, compromising type safety of the host type system.
To ensure linear usage of multiparty channels, \textit{runtime checks} are inserted to detect if a channel has been used more than once. 
This is because the  type systems of their respective host languages do not provide static linearity checking mechanism.

\begin{figure}[t]
\begin{centering}
\hspace{1.5mm}
\begin{minipage}{0.49\textwidth}
\vspace{7mm}
\includegraphics[scale=0.31]{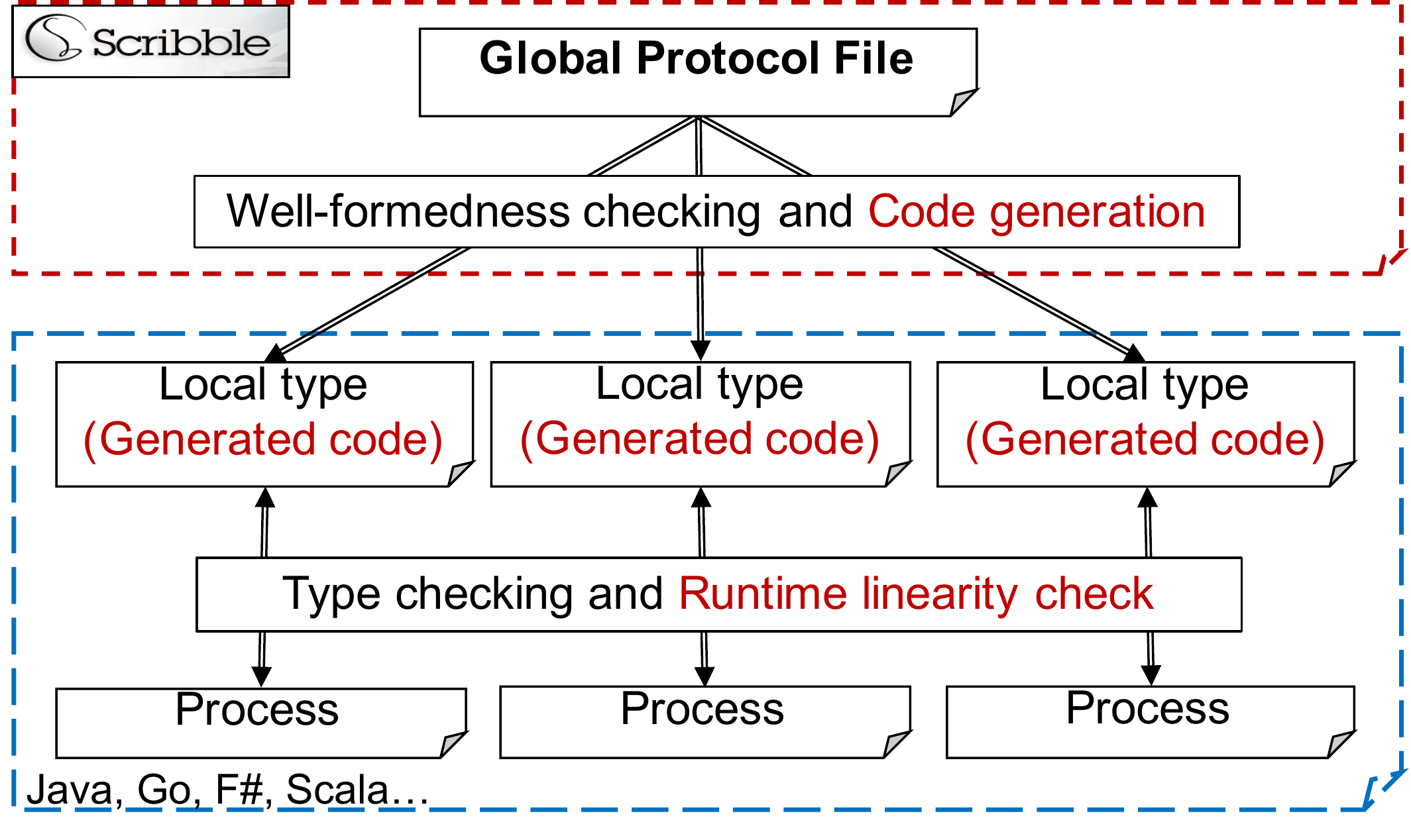}
\end{minipage}
\hspace{1mm}
\begin{minipage}{0.48\textwidth}
\includegraphics[scale=0.31]{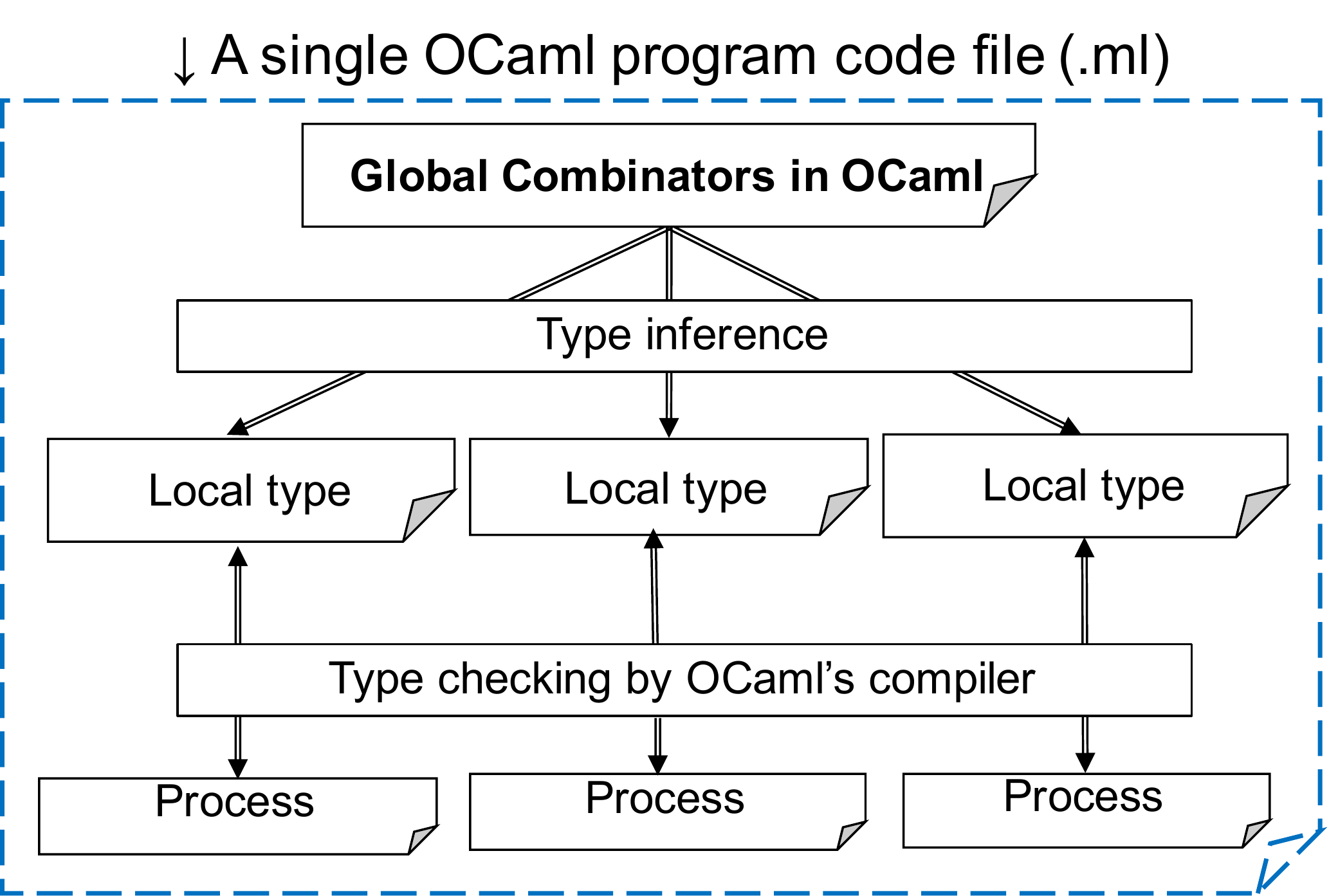}
\end{minipage}
\end{centering}
\caption{(a) State-of-the-art MPST implementations and (b) \must methodology }
\label{fig:overview}
\end{figure}

\textbf{Our approach.} This paper presents a library for programming MPST protocols in \OCaml
that solves the above challenges.
Our library, \must, allows to specify, verify and implement
MPST protocols in a single language, OCaml. %
Specifically, we address {\bf(C1)} by developing
\textit{\globalcombinators}, an embedded DSL (EDSL) for writing global
types in \OCaml. We address {\bf(C2)} by
encoding multiparty channels into \textit{channel
  vectors}  -- a data structure, storing a nested sequence of binary channels. Moreover,  
\must verifies \textit{statically} the linear usage of communication channels,
using OCaml's strong typing system and supports session delegation. 

The key device in our approach is the discovery that in a system with
variant and record types, checking compatibility of local types
coincides with existence of least upper bound
w.r.t.~subtyping relation. This realisation enables a fully
static MPST implementation, i.e.,~static checking not only on local but also on global types in a general purpose language.

Programming with \must  (\cfig{fig:overview}(b))
closely follows the ``top-down'' methodology of MPST, but
differs from the traditional MPST framework in
\cfig{fig:overview}(a).
To use our library, a programmer specifies the global protocol with a set of global combinators.
The \OCaml typechecker verifies  correctness of the global protocol
and infers local types from global combinators. 
A developer implements the endpoint processes using our \must API. Finally, the \OCaml type checker verifies that the API is used according to the inferred type.
The benefits of \must are that it is (1) {\em lightweight} -- it does
not depend on any external code-generation mechanism, verification of
global protocols is reduced to typability of global combinators; (2)
{\em fully-static} -- our embedding integrates with recent techniques for static checking of binary session types and linearly-typed lists \cite{imai18sessionscp,linocaml}, which we adopt to
implement multiparty session channels and session delegation;
(3) {\em usable} -- we can auto-detect and correct protocol violations in
the program, guided by OCaml programming environments like Merlin \cite{DBLP:journals/pacmpl/BourRS18};
(4) {\em extensible} -- while most MPST implementations rely on a
nominal typing, we embed session types  in \OCaml's {\em structural}
types, and preserve session subtyping \cite{DBLP:journals/jlp/GhilezanJPSY19};
and (5) {\em expressive} --  we can
type strictly more processes than \cite{scalas16lightweight} 
(see \Sec~\ref{sec:related}).

\myparagraph{Contributions.} Contributions and the outline of the paper are as follows:
\begin{description}
\item[\Sec~\ref{sec:overview}] gives an overview of programming with \must,  a
  library in \OCaml for specification, verification and
  implementations of communication protocols.

\item[\Sec~\ref{sec:formalism}]
formalises global combinators, presents their typing
system, and proves a \emph{sound realisability of
global combinator}, i.e. 
a set of local types 
inferred from a global combinator can type 
a channel which embeds a set of endpoint behaviours as OCaml data structures. 

\item[\Sec~\ref{sec:implementation}] 
discusses the design and implementation of global combinators. 

\item[\Sec~\ref{sec:linear:channels}] summarises
the \must communication library and explains how we utilise 
advanced features/libraries in OCaml to enable dynamic/static linearity checking on channels.

\item[\Sec~\ref{sec:evaluation}] evaluates \must. 
We compare \must\ with several different implementations and 
demonstrate the expressiveness of  \must by showing implementations 
of MPST examples, as well as a variety of real-world protocols. 
We demonstrate our library can interoperate with existing libraries and services, namely we implement DNS (Domain Name Service) and the OAuth (Open Authentication) protocols on top of existing libraries.
\end{description}
We discuss related work in \Sec~\ref{sec:related} and
conclude with future work in \Sec~\ref{sec:conclusion}.
Full proofs, omitted definitions and examples can be found in
\iftoggle{techreport}{Appendix}{\cite{IRYY2020TechReport}}.
Our implementation, {\bfseries\ourlibrary} is available at {\tt \url{https://github.com/keigoi/ocaml-mpst}}
including benchmark programs and results.

\section{Overview of OCaml Programming with Global Combinators}
\label{sec:overview}
This section gives an overview of multiparty
session programming in \must by examples. It starts from 
declaration of global combinators, followed by endpoint
implementations.  We also demonstrate how 
errors can be reported by an OCaml programming environment like Merlin \cite{DBLP:journals/pacmpl/BourRS18}.
In the end of this section, we show the syntax of global combinators and the constructs 
of \must API in \cfig{fig:synopsis}. 
The detailed explanation of the implementations of the constructs 
is deferred to \S~\ref{sec:implementation}.

\myparagraph{From global combinators to communication programs.} %
We illustrate
\emph{\globalcombinators} starting from a simple authentication protocol
(based on OAuth 2.0 \cite{rfc6749}).
A full version of the protocol is implemented and discussed in
\S~\ref{sec:evaluation}.  %
\cfig{fig:full:impl} shows the complete \OCaml implementation of the protocol,
from the protocol specification (using \globalcombinators)
to the endpoint implementations (using \must API). %

\begin{figure}[t]
{\lstset{numbers=left}
\begin{OCAMLLISTING}
let oAuth = (s --> c) login @@._._(c --> a) pwd @@._._(a --> s) auth @@._._finish (* global protocol*)^\label{line:full:global}^
\end{OCAMLLISTING}}
\hrule
{\lstset{firstnumber=last}
{\lstset{numbers=left,multicols=2}
\begin{OCAMLLISTING}
(* The client process *)
let cliThread () = ^\label{line:cli:start}^
  let ch = get_ch c oAuth in^\label{line:cli:getch}^
  let `login(x, ch) = recv ch#role_S in ^\label{line:cli:recv}^
  let ch = send ch#role_A#pwd "pass" in  ^\label{line:cli:send}^
  close ch ^\label{line:cli:close}^

(* The service process *)
let srvThread () =
  let ch = get_ch s oAuth in
  let ch = send ch#role_C#login "Hi" in
  let `auth(_,ch) = recv ch#role_A in
  close ch

(* The authenticator process *)
let authThread () =
  let ch = get_ch a oAuth in
  let `pwd(code,ch) = recv ch#role_C in
  let ch = send ch#role_S#auth true in
  close ch ^\label{line:full:end}^

(* start all processes *)
let () =
  List.iter Thread.join [^\label{line:full:threadstart}^
    Thread.create cliThread ();
    Thread.create srvThread ();
    Thread.create authThread ()]  ^\label{line:full:threadend}^
\end{OCAMLLISTING}}}
\caption{Global protocol and local implementations for OAuth protocol \protect\footnotemark}
\label{fig:full:impl}
\end{figure}

\footnotetext{We use a simplified syntax that support the in-built communication transport of Ocaml. For the full syntax of the library that is parametric on the transport, see the  \href{https://github.com/keigoi/ocaml-mpst/blob/master/instructions.md\#&note-on-syntax-discrepancies}{repository}.}
%
%
%
%
%
%

%
%

%
%
The protocol consists of three parties, a service
\Rptp{s}, a client \Rptp{c}, and an authenticator \Rptp{a}.
The interactions between the parties (hereafter also called {\em roles}) proceed as follows:
(1) the service \Rptp{s} sends to the client
\Rptp{c} a \oCODE{login} message containing a greeting (of type
\oCODE{string}); (2) the client then continues by sending its
password (\oCODE{pwd}) (of type \oCODE{string}) to the authenticator
\Rptp{a}; and (3) finally the authenticator \Rptp{a} notifies \Rptp{s}, by sending an \oCODE{auth} message (of type \oCODE{bool}), whether the client
access is authorised.

The global protocol \oCODE{oAuth} in Line~\ref{line:full:global} is specified using two global combinators, \oCODE{-->} and
\oCODE{finish}.
The former represents a
point-to-point communication between two roles, while the latter signals the end of a protocol.
The operator \oCODE{@@} is a {\color{modify}right-associative} function application {\color{modify}operator} to eliminate parentheses, i.e., \lstinline!(c._-->._a)._pwd._@@._$\mathit{exp}$! is equivalent to
\lstinline!(c._-->._a)._pwd._($\mathit{exp}$)!, where \oCODE{-->} works as a four-ary function
which takes roles \oCODE{c} and \oCODE{a} and label \oCODE{pwd} and continuation $\mathit{exp}$.
We assume that \oCODE{login},
\oCODE{pwd} and \oCODE{auth} are predefined by the user as {\em
  label objects} with their {\em payload types} of \oCODE{string}, \oCODE{string} and
\oCODE{bool}, respectively\footnote{
  To be precise, the labels are \emph{polymorphic} on their payload types which are instantiated at the point where they are used.
}.
Similarly, \Rptp{s}, \Rptp{c} and \Rptp{a}
are predefined {\em role objects}.
We elaborate on how to define these custom labels and roles in \S~\ref{sec:implementation}.

The execution of the \oCODE{oAuth} expression returns a tuple of three {\em channel vectors} -- one for each role in the global combinator. 
Each element of the tuple can be extracted using an index, encoded in role objects (\oCODE{c}, \oCODE{s}, and \oCODE{a}).
Intuitively, the role object \oCODE{c} stores a functional pointer that points to the first element of the tuple, \oCODE{s} points to the second, and \oCODE{a}  to the third element. 
The types of the extracted channel vectors reflect the local behaviour that each role, specified in the protocol, should implement. 
Channel vectors are objects that hide the \textit{actual bare communication channels} shared between every two communicating processes. 

Lines~\ref{line:cli:start}--\ref{line:full:end}
present the implementations for all three processes specified in the global protocol.
We explain the implementation for the client -- \oCODE{cliThread} (Lines~\ref{line:cli:start}--\ref{line:cli:close}).  
Other processes are similarly implemented. 
Line~\ref{line:cli:getch} extracts the channel vector that encapsulates the behaviour of the \Rptp{c}lient, 
i.e the first element of \oCODE{oAuth}. 
This is done by using the function \lstinline!get_ch! (provided by our library) applied to the role object \oCODE{c} and the expression \oCODE{oAuth}.  %
%
%

%
%
%
%
%
%
%
%
%
%
%

%
%
%
%
%
%
%
%
%
%
%
%
%
%
%
%
%
%
%
%
%
%

%

Our library provides two main communication primitives, namely \oCODE{send} and \oCODE{recv}.
To statically check communication structures using types,
we exploit OCaml's {\em structural} types of objects and polymorphic variants
(rather than their nominal counterparts of records and ordinary variants).
In Line~\ref{line:cli:recv},
  \oCODE{ch#role_S} is an invocation of method \oCODE{role_S} on an object \oCODE{ch}.
  The \oCODE{recv} primitive waits on a {\em bare channel} returned by the method invocation.
The returned value is matched against a variant tag indicating the input label \oCODE{`login} with 
the pair of the payload value \oCODE{x} and a continuation \oCODE{ch} (shadowing the previous usage of \oCODE{ch}).
Then, on Line~\ref{line:cli:send}, two method calls on \oCODE{ch} are performed, e.g \oCODE{ch#role_A#pwd}, which extract a communication channel for sending a password (\oCODE{pwd}) to the
\Rptp{a}uthenticator. This channel is passed 
 to the \oCODE{send} primitive, along with the payload value \oCODE{"pass"}.
Then, \oCODE{let} rebinds the name \oCODE{ch} to the continuation returned by \oCODE{send} and on
Line~\ref{line:cli:close} the channel is closed.
Each operation is guided by the host OCaml type system, via {\em channel vector type}.
For example, the \oCODE{c}lient channel \oCODE{ch} extracted in Line~\ref{line:cli:getch} has a channel vector type (inferred by \OCaml type checker)
\lstinline!<role_S:._[`login._of._string._*._$t$]._inp>!
which denote reception (suffixed by \oCODE{inp}) from server of a \oCODE{login} label, then continuing to $t$, where
$t$ is \lstinline!<role_A:<pwd:(string,close)._out>>!
denoting sending (\oCODE{out}) to authenticator of a \oCODE{pwd} label, followed by closing.
Note that the type \oCODE{<f:._t>} denotes an \OCaml object with a
field \oCODE{f} of type \oCODE{t};  \oCODE{[`m._of._t]} is a (polymorphic)
variant type having a tag \oCODE{m} of type  \oCODE{t}.
Finally, in Lines~\ref{line:full:threadstart}--\ref{line:full:threadend} all processes are started in new threads.

\begin{figure}[t]
\begin{subfigure}[t]{0.49\textwidth}
{\lstset{numbers=left}
\begin{OCAMLLISTING}
let oAuth2 () =^\label{line:ex2globalbegin}^
  (choice_at s (to_s login_cancel)^\label{line:ex2choice}^
    (s, oAuth ()) ^\label{line:auth2:choice1}^
    (s, (s --> c) cancel @@
        (c --> a) quit @@ ^\label{line:auth2:choice2}^
        finish))^\label{line:ex2resultend}^
\end{OCAMLLISTING}}
\caption{Protocol With Branching\label{fig:auth2a}}
\end{subfigure}
\begin{subfigure}[t]{0.49\textwidth}
{\lstset{numbers=left}
\begin{OCAMLLISTING}
let oAuth3 () =^\label{line:oauth3globalbegin}^
  fix (fun repeat ->
  (choice_at s (to_s oauth2_retry)^\label{line:oauth3choice}^
    (s, oAuth2 ()
    (s, (s --> c) retry @@
        repeat))^\label{line:oauth3resultend}^
\end{OCAMLLISTING}}
\caption{Protocol With Branching \& Recursion\label{fig:auth2b}}
\end{subfigure}
\caption{Extended \lstinline!oAuth! protocols}
\label{fig:auth2}
\end{figure}
\myparagraph{On the expressiveness of  well-typed global protocols.}
\cfig{fig:auth2} shows two global protocols that extend \oCODE{oAuth} with new behaviours.
In  \cfig{fig:auth2a}, the global combinator \oCODE{choice_at} specifies a branching
behaviour at role \Rptp{s}.  
In the first case (Line~\ref{line:auth2:choice1}),
the protocol proceeds with 
protocol \oCODE{oAuth}. In the second case (Line~\ref{line:auth2:choice2}) the service sends \oCODE{cancel},
to the client, and the client sends a \oCODE{quit} message to the authenticator.
The deciding role, \Rptp{s},  is explicit in each branch. 
The choice combinator requires a user-defined
\oCODE{(to_s login_cancel)} (Line~\ref{line:ex2choice}) that specifies
concatenation of two objects for sending in branches. Its implementation
is straightforward (see \S~\ref{sec:implementation}). The protocol
\oCODE{oAuth3} in \cfig{fig:auth2b} reuses \oCODE{oAuth2} and
further elaborates its behaviour by offering a retry option. It
demonstrates a recursive specification where the \oCODE{fix} combinator binds the protocol itself to variable \oCODE{repeat}.

The implementation of the corresponding client code for \cfig{fig:auth2a} is shown on 
\cfig{fig:client:oAuth2:dynamic}. The code is similar as before,
but uses a pattern matching against multiple tags
\oCODE{`login} and \oCODE{`cancel} to specify an {\em external choice} on the client, i.e the client can receive messages of different types
and exhibit different behaviour according to received labels. %
The behaviour that a role can send messages of different types, which is often referred to as an {\em internal choice}, is represented as an object with multiple methods.

Our implementation also preserves the subtyping relation in session types \cite{DBLP:journals/jlp/GhilezanJPSY19}, i.e the safe
replacement of a channel of more capabilities in a context where a channel of less
capabilities is expected. Session subtyping is important in practice since it ensures
backward compatibility for protocols:
a new version of a protocol does not break existing implementations.
For example, the client function in
\cfig{fig:client:oAuth2:dynamic} is typable under both protocols
\oCODE{oAuth2} and \oCODE{oAuth3} since the type of the channel stipulating the behaviour for role \oCODE{c} in \oCODE{oAuth2} (receiving either message \oCODE{`login} or \oCODE{`cancel}) is a subtype of the channel for \oCODE{c} in \oCODE{oAuth3} (receiving \oCODE{`login},  \oCODE{`cancel}, or \oCODE{`retry}). 

\myparagraph{Static linearity and session delegation.}
The implementations presented in \cfig{fig:full:impl}, as well as \cfig{fig:client:oAuth2:dynamic}
detect linearity violations at runtime, as common in MPST implementations \cite{HY2016,scalas17linear} in a non-substructural type system.
We overcome this dynamic checking issue by an alternative approach, listed in \cfig{fig:client:oAuth2:static}.
We utilise an extension (\oCODEEsc{let\%lin}) %
for linear types in OCaml \cite{linocaml} that statically
enforces linear usage of resources by combining the usage of
parameterised monads \cite{indexedmonad2,atkey09parameterized,padovani16simple} and lenses \cite{foster07combinators}.
Our library is parameterised on the chosen approach, static or dynamic.
A few changes are made to avoid explicit handling of linear resources:
(1) \oCODE{ch} in \cfig{fig:client:oAuth2:static}
refers to a {\em linear} resource and has to be matched against a {\em linear pattern} prefixed by \oCODE{#}. 
(2) Roles and labels are now specified as a {\em selector} function of the form \oCODE{(fun x->x#role#label)}.

\begin{figure}[t]
\begin{subfigure}[t]{0.40\textwidth}
{\lstset{numbers=left}
\begin{OCAMLLISTING}
match recv ch#role_S with
|`login(pass, ch) ->
  let ch = send ch#role_A#pwd pass
  in close ch
|`cancel(_,ch) ->
  let ch = send ch#role_A#quit ()
  in close ch
\end{OCAMLLISTING}
}
\caption{Dynamic Linearity Checking\label{fig:client:oAuth2:dynamic}}
\end{subfigure}
\begin{subfigure}[t]{0.59\textwidth}
\begin{OCAMLLISTING}
match&%
|`login(pass, &#&ch) ->
  let&%
  in close ch
|`cancel(_, &#&ch) ->
  let&%
  in close ch
\end{OCAMLLISTING}
\caption{Static Linearity Checking\label{fig:client:oAuth2:static}}
\end{subfigure}
\caption{Two Modes on Linearity Checking}
\label{fig:client:oAuth2}
\end{figure}

Our implementation is also the first to support {\em static}
multiparty sessions delegation (the capability to pass a channel to
another endpoint): our encoding yields it for free, via existing
mechanisms for binary delegation (see \S~\ref{sec:implementation}).

\myparagraph{Errors in global protocol and \must endpoint
  programs.}
Our framework ensures that a well-typed \must  program  precisely implements the behaviour of its defined global protocol. %
Hence, if a program does not conform to its protocol, a compilation error is
reported.  \cfig{fig:typeerror} shows the error reported when swapping the order of send and receive actions
  (Lines~\ref{line:cli:send} and \ref{line:cli:recv}) in the client implementation in \cfig{fig:full:impl}.
Similarly,  errors will also be reported if we misspell any of the methods
	\oCODE{pwd}, \oCODE{role_A}, or \oCODE{role_C}. 

\begin{figure}[t]
  \begin{adjustbox}{width=\columnwidth,center}
    \begin{tabular}{ll}
      \toprule
      \multicolumn{2}{l}{\textbf{Global Combinators to Local Types} where $t_i$ is a local type at \lstinline!r$_i$! in $g$ ($1\le i \le n$)}
      \\
      \toprule
      Global Combinator & Synopsis\\
      \toprule
      \lstinline!(r$_i$ -->$\ $r$_j$)$\,$m$\,\,g$!
      & 
          Transmission from \lstinline!r$_i$! to \lstinline!r$_j$! of label \lstinline!m! (with a payload).
      \\
      \hline
      \TableStrut
      \lstinline!choice_at $r_a$ $\mrg$ ($r_a$, $g_1$) ($r_a$, $g_2$)!
      & \begin{minipage}{0.6\textwidth}
          Branch to $g_1$ or $g_2$ guided by $r_a$.
        \end{minipage}
      \\
      \hline
      \lstinline!finish!
      & Finished session.
      \\
      \hline
      \lstinline!fix (fun x ->$\,g$)!
      & 
          Recursion. Free occurrences of $x$ is equivalent to $g$ itself.
      \\
      \bottomrule
      \multicolumn{2}{l}{\textbf{Local Types and Communication Primitives}}
      \\
      \toprule
      Communication Primitive & Synopsis\\
      \toprule
      \lstinline^send s#role_--!r!--#m$_k$ $e$^
      & 
          Send to role \lstinline^--!r!--^
          label \lstinline!m$_k$! with payload \lstinline^$e$^, returning continuation.
      \\
      \hline
      \TableStrut
      \begin{lstlisting}
let `m(x, s) = receive s#role_--!r!--
in $e$
      \end{lstlisting}
      & \begin{minipage}{0.6\textwidth}
          Receive from \lstinline^--!r!--^ label \lstinline!m! with payload \lstinline^x : $v$^ and continue
          to $e$ with endpoint \lstinline!s : $t$!
        \end{minipage}
      \\
      \hdashline[1.5pt/1.5pt]
      \TableStrutDouble
      \begin{lstlisting}
match receive s#role_--!r!-- with
| `m$_1$(x$_1$, s) -> $e_1$ | $\cdots$
| `m$_n$(x$_n$, s) -> $e_n$
      \end{lstlisting}
      & \begin{minipage}{0.6\textwidth}
          Receive from \lstinline^--!r!--^ one of labels \lstinline!$\{$m$_i\}$! ($1 \le i \le n$)
          where payload is $v_i$ and continue with $t_i$ in $e_i$
        \end{minipage}
      \\
      \hline
      \lstinline^close s^
      & Closes a session
      \\
      \bottomrule
    \end{tabular}%
  \end{adjustbox}
  \caption{(a) Global Combinators (top) and (b) Communication APIs of \ourlibrary (bottom)}
  \label{fig:synopsis}
\end{figure}

\begin{figure}[b]
  \begin{adjustbox}{width=0.99\columnwidth,center}
    \includegraphics[scale=0.50]{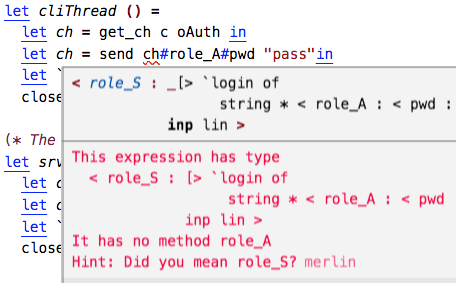}
    \hspace*{2em}
    \includegraphics[scale=0.50]{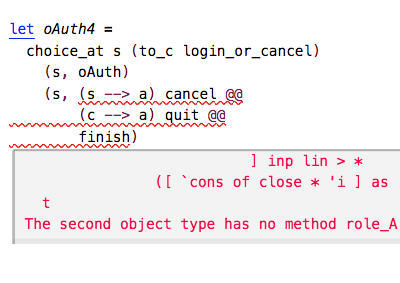}
  \end{adjustbox}
\caption{Type Errors Reported by Visual Studio Code (Powered by Merlin), in (a) Local Type (left) and (b) Global Combinator (right)}
\label{fig:typeerror}
\end{figure}

Similarly, an error is reported if 	the global protocol is {\em not safe} (which corresponds to an ill-formed MPST protocols \cite{ICALP}) since this may lead to {\em unsafe}
implementations. Consider \cfig{fig:typeerror} (b), where we modify \oCODE{oAuth2}
such that \oCODE{s} sends a \oCODE{cancel} message to
\oCODE{a}. This protocol (\oCODE{oAuth4}) exhibits a race condition: even if all parties adhere to the specified behaviour,
\Rptp{c} can send a \oCODE{quit} before \Rptp{s} sends \oCODE{login}, which will lead to a deadlock on \Rptp{s}.
Our definition of global combinators
prevents such ill-formed protocols, and the \OCaml compiler
will report an error.
The actual error message reported in \OCaml
detects the mismatch between \oCODE!a! and \oCODE!c!, 
indicating violation of the {\em active role} property in the MPST literature \cite{ICALP} --
the sender must send to the same role.
\section{Formalisms and Typing for Global Combinators}
\label{sec:formalism}

This section formalises global combinators and their typing system,
along a formal correspondence between a global combinator and 
channel vectors. The aim of this section is to provide 
a guidance towards 
descriptions of the implementations presented in 
\S~\ref{sec:implementation},\ref{sec:linear:channels}. 

We first give the syntax of global combinators and channel vectors 
in  \S~\ref{subsec:calculus-global}. We then propose 
a typing system of global combinators in \S~\ref{sec:typing:global},
illustrating that the rules check their well-formedness. 
We define derivation of channel vectors from global combinators
in \S~\ref{sec:evalglobal}. 
The main theorem (Theorem \ref{col:SubjectReductionForGC}) 
states that a well-typed global combinator always derives 
a channel vector which is typable by a corresponding set of 
local types, i.e.~any well-typed global combinator is soundly realisable 
by a tuple of well-typed channel vectors. 

\subsection{Global Combinators and Channel Vector Types \label{subsec:calculus-global}}
\myparagraph{{\bf Global combinators}} denote
a communication protocol
which describes the whole
conversation scenario of a multiparty session.

\begin{DEF}[Global combinators and channel vector types]\label{def:global}\thmstart
The syntax of {\em global combinators}, written $\gocaml,\gocaml',..$,
are given as:
\[
\gocaml  \ \ \grmeq \ %
\ogtComm{\roleP}{\roleQ}{}{\mpLab}{\otT}{\gocaml}
\grmor \ogtChoice{\roleP}{\gocaml[i]}{i \in I}
\grmor \ogtRec{\ogtRecVar}{\gocaml}
\grmor \ogtRecVar \grmor \ogtEnd
\]
where the syntax of {\em payload types} $\otS, \otT, \ldots$ (also called {\em channel vector types}) is
given below:
\label{def:ocaml-types}
    \begin{displaymath}
        \otT,\otS \quad \grmeq \quad
        \otOut{\otT} \grmor \otInp{\otT} \grmor \otChan{\otT}
        \grmor \otT[1]{\otTimes}{\elipc}{\otTimes}\otT[n] \grmor \otRecord{\ocLab[i]}{\otT[i]}{i\in I} \grmor \otVariant{\ocLab[i]}{\otT[i]}{i \in I}
        \grmor \otRec{\otRecVar}{\otT} \grmor \otRecVar \grmor \otUnit
    \end{displaymath}
\end{DEF}
The formal syntax of global combinators comes from Scribble
\cite{scribble} and corresponds to the standard global types in MPSTs \cite{FeatherweightScribble}.
We assume a set of
participants ($\roleUnivSet =\{\roleP, \roleQ, \roleR, \cdots\}$),
and that of alphabets ($\ASigma = \{\labOk, \labCancel, \cdots\}$). %
{\textbf{\emph{Communication combinator}}
$\ogtComm{\roleP}{\roleQ}{}{\mpLab}{\otT}{\gocaml}$
states that participant $\roleP$ can
send a message of type $\otT$ %
with label $\mpLab$ to participant $\roleQ$
and that the interaction
described in $\gocaml$ follows.
We require $\roleP\neq\roleQ$ to prevent self-sent messages.
We omit the payload type when {\em unit} type $\otUnit$,
and assume $\otT$ is {\em closed}, i.e. it does not contain free recursive
variables.
{\textbf{\emph{Choice combinator}}}
$\ogtChoice{\roleP}{\gocaml[i]}{i \in I}$
is a branching in a protocol where $\roleP$ makes a decision
(i.e. an output) on which branch
the participants will take.
{\textbf{\emph{Recursion}}
$\ogtRec{\ogtRecVar}{\gocaml}$ is for recursive protocols,
assuming that variables
($\ogtRecVar, \ogtRecVari, \dots$) are guarded in the
standard way, i.e. they only occur under the communication combinator.
{\textbf{\emph{Termination}}
$\ogtEnd$ represents session termination.
We write\; $\roleP \in \otRolesSet{\gocaml}$ %
\;(or simply $\roleP \!\in\! \gocaml$) \;iff, for some $\roleQ$, %
either $\gtFmt{\roleP {\ogtTo} \roleQ}$ %
or $\gtFmt{\roleQ {\ogtTo} \roleP}$ %
occurs in $\gocaml$.

\begin{EX}\label{ex:globalcombinatorauth}\thmstart
  The global combinator $\gAuth$ below specifies a variant of
an authentication protocol in \cfig{fig:auth2}
  where $T=\mathtt{string}$ and {$\roleC$}lient sends $\labAuth$ to {$\roleS$}erver, then
  {$\roleS$}erver replies with either $\labOk$ or $\labCancel$.
\[
\begin{array}{l}
\gAuth = 
\ogtComm{\roleC}{\roleS}{}{\labAuth}{\otT}{\bigl(\ogtCommBin{\roleS}{\roleC}{\labOk}{\otT}{\ogtEnd}{\labCancel}{\otT}{\ogtEnd}\bigr)}
\end{array}
\]
\end{EX}

\myparagraph{{\bf Channel vector types}} abstract behaviours of each
participant using standard data structure and channels.
We assume {\em labels} $\ocLab,\ocLab',\ldots$ range over
$\roleUnivSet \cup \ASigma$.
Types $\otOut{\otT}$ and $\otInp{\otT}$ denote
{\textbf{\emph{output}}} and {\textbf{\emph{input channel types}}},
with a value or channel of type 
$\otT$ (note that the syntax includes \emph{session delegation}).
$\otChan{\otT}$ is an {\textbf{\emph{io}}}-type which is a subtype of both input or output types \cite{SangiorgiD:picatomp}. 
$\otT[1]{\otTimes}{\elipc}{\otTimes}\otT[n]$ is an $n$-ary
{\textbf{\emph{tuple type}}.
$\otRecord{\ocLab[i]}{\otT[i]}{i\in I}$ is a {\textbf{\emph{record
      type}}
where each field ${\ocLab[i]}$ has type $\otT[i]$ for $i \in I$.
$\otVariant{\ocLab[i]}{\otT[i]}{i \in I}$
is a {\textbf{\emph{variant type}} {\color{modify}\cite{SangiorgiD:picatomp}} where each ${\ocLab[i]}$
is a possible {\em tag} (or {\em constructor}) of that type and
$\otT[i]$ is the argument type of the tag.
In both record and variant types, we assume the fields and tags are distinct
(i.e. in $\otRecord{\ocLab[i]}{\otT[i]}{i\in I}$ and $\otVariant{\ocLab[i]}{\otT[i]}{i \in I}$,
we assume $\ocLab[i]\neq\ocLab[j]$ for all $i \neq j$).
The symbol $\otUnit$ denotes a unit type.
Type $\otRecVar$ is a variable for recursion.
A \textbf{\emph{recursive type}} takes an equi-recursive viewpoint,
i.e.
$\otRec{\otRecVar}{\otT}$
is viewed as $\otT\subst{\otRecVar}{\otRec{\otRecVar}{\otT}}$.
Recursion variables are guarded and payload types are
closed.

\myparagraph{Channel vectors: Session types as record and variant types.}
The execution model of MPST assumes that processes communicate 
by exchanging messages over input/output (I/O) channels.  
Each channel has the capability to communicate with multiple other processes. 
\emph{A local session type} prescribes the local behaviour for a role in a
global protocol by assigning a type to the communication channel utilised by the role. More precisely, 
a local session type specifies
the exact order and payload types for the communication actions performed on each channel (see \cfig{fig:overview}(a)). 
In practice, processes communicate on a low-level \textit{bi-directional} I/O channels (\textit{bare} channels), which 
are used for synchronisation of two (but \emph{not} multiple) processes. 
Therefore, to implement local session types in practice, a process should utilise
multiple bare channels,  
preserving the order, in which such channels should be used. 
We encode local session types as channel vector types, which 
\textit{wrap} bare channels (represented in our setting by $\otInp{\otT}, \otOut{\otT},\otChan{\otT}$  types) in record and variant types. 
This is illustrated in the following table, 
with the corresponding local session types for reference.\\

\smallskip\noindent\centerline{
  \begin{tabular}{l|l|l|l}
Behaviour
& \ Channel vector type \ & \ Local session type \ \cite{scalas19less}\\
    \hline
Selection (Output choice) 
    & \ $\otIntSumSmall{\roleQ}{i \in I}{\otIntChoice{\mpLab[i]}{\otS[i]}{\otT[i]}}$
    & \ $\stIntSum{\roleQ}{i \in I}{\stChoice{\mpLab[i]}{\stS[i]}\stSeq \stT[i]}$ \\
Branching (Input choice)
    & \ $\otExtSumSmall{\roleQ}{i \in
      I}{\otExtChoice{\mpLab[i]}{\otS[i]}{\otT[i]}}$ \quad
    & \ $\stExtSum{\roleQ}{i \in I}{\stChoice{\mpLab[i]}{\stS[i]}\stSeq \stT[i]}$ \\
Recursion
& \ $\otRec{\otRecVar}{\otT}$,
    $\otRecVar$
& \ $\stRec{\stRecVar}{\stT}$, $\stRecVar$ \ \\
Closing
 & \ $\otUnit$ & \ $\stEnd$
  \end{tabular}}

\smallskip

\noindent
Intuitively, the behaviour of sending a message is represented as a record type, which stores inside its fields a bare output channel and a continuation; the input channel required when receiving a message is stored in a variant type. 
Type
$\otIntSumSmall{\roleQ}{i \in I}{\otIntChoice{\mpLab[i]}{\otS[i]}{\otT[i]}}$
is read as: to send label $\mpLab[i]$ to $\roleQ$,
(1)  the channel vector should be `peeled off' from the nested record by
extracting the field $\roleQ$ then $\mpLab[i]$;
then (2) it returns a pair $\otOut{\otS[i]}{\otTimes}\otT[i]$ of an output
channel and a continuation.
Type $\otExtSumSmall{\roleQ}{i \in I}{\otExtChoice{\mpLab[i]}{\otS[i]}{\otT[i]}}$
says that (1) the process extracts the value  stored in the field
$\roleQ$, then reads on the resulting input channel ($\otInp{}$)
to receive a variant of type
$\otVariant{\mpLab[i]}{\otPair{\otS[i]}{\otT[i]}}{i \in I}$;
then, (2) the tag (constructor) $\mpLab[i]$ of the received variant
indicates the label which $\roleQ$ has sent,
and the former's argument $\otS[i]$ is the payload, and
the latter $\otT[i]$ is the continuation.

The anti-symmetric structures between output types $\otIntSumSmall{\roleQ}{i \in I}{\otIntChoice{\mpLab[i]}{\otS[i]}{\otT[i]}}$
and input types $\otExtSumSmall{\roleQ}{i \in I}{\otExtChoice{\mpLab[i]}{\otS[i]}{\otT[i]}}$
(notice the placements of $\otOut{}$ and $\otInp{}$ symbol in these types)
come from the fact that an output is an {\em internal choice}
where output labels are proactively chosen via projection on a record field,
while an input is an {\em external choice}
where input labels are reactively chosen via
pattern-matching among variant constructors.

\subsection{Typing Global Combinators} \label{sec:typing:global}
A key finding of our work is that compatibility of local types can be checked using 
a type system with record and variant subtyping. 
Before explaining how each combinator ensures compatibility of types, 
we give an intuition of well-formed global protocols following \cite{ICALP}. 

\myparagraph{Well-formedness and choice combinator.} 
A well-formed global protocol ensures that a
protocol can be correctly and \textit{safely} realised by a system of endpoint processes. 
Moreover, a set of processes that follow the prescribed behaviour is \textit{deadlock-free}. 
{Well-formedness imposes several restrictions on the protocol structure, \precameraready{notably on {\em choices}}.
This is necessary because some protocols, such as \oCODE{oAuth4} in \cfig{fig:typeerror}(b) (\S~\ref{sec:overview}),
are unsafe or inconsistent.
More precisely, a protocol is well-formed if local types can be
generated for all of its roles, i.e the {\em endpoint projection}
function
\iftoggle{techreport}{%
  \cite[Def. 3.1]{ICALP}[Def.~\ref{def:projection} in Appendix (\S~\ref{sec:mpst})]
}{%
  \cite[Def. 3.1]{ICALP}\cite{IRYY2020TechReport}%
} is defined for all roles. 
\precameraready{Our encoding allows the well-formedness restrictions to be checked \textit{statically},  by the \OCaml typechecker.
Below, we explain the main syntactic restrictions of endpoint projection, which are imposed on {\em choices} and checked statically:}
\begin{description}
\item[R1] (\textbf{active role}) in each branch of a choice,
  the first interaction
  is from the same sender role ({\em active role}) to the same receiver role ({\em directed output}).
\item[R2] (\textbf{deterministic choice}) output labels from an active role are pairwise distinct (i.e., protocols are deterministic)
\item[R3] (\textbf{mergeable}) the behaviour of a role from all branches should be mergeable, which is ensured by the following restrictions: 
  \begin{description}
    \item[M1] two input choices are merged only if 
      (1) their sender roles are the same ({\em directed input}), and (2) their continuations are recursively mergeable if labels are the same.
   \item[M2]  two output choices can be merged only if they are the same.
  \end{description}
\end{description}
Intuitively, the conditions in \Rthree ensure that a process is able to determine unambiguously
which branch of the choice has been taken by the active role, otherwise 
the process should be \textit{choice-agnostic}, i.e it should preform the same actions in all branches.
Requirement \Rthree is known in the MPST literature as \emph{recursive full merging} \cite{ICALP}.%

\myparagraph{Typing system for global combinators.} Deriving channel vector types from a global combinator
corresponds to the {\em end point projection} in multiparty
session types \cite{DBLP:journals/jacm/HondaYC16}. 
Projection of global protocols relies on the notion of merging (\Rthree).
As a result of the encoding of local types as  channel vectors with record and variants, 
the  \textit{merging} relation coincides with the {\em least upper bound} (join) in the subtyping relation.
This key observation allows us to embed well-formed global protocols
in \OCaml, and check them using the \OCaml type system. 

Next we give the typing system of global combinators, explaining how each of the 
typing rules ensures the verification conditions \Rone-\Rthree. The
typing system uses the following subtyping rules. 

\begin{DEF}\label{def:subtyping}\thmstart
The subtyping relation \framebox{$\otSub$} is {\em co}inductively defined by the following rules.\\
    \scalebox{0.9}{\(
    \begin{array}{c}
    \cinference{%
     \inferrule{\iruleOSubUnit} \quad
    }{
      \otUnit \otSub \otUnit
    }
    \cinference{%
      \inferrule{\iruleOSubOutCh} \ %
    }{
      \otChan{\otT} \otSub \otOut{\otT}
    }
    \cinference{%
      \inferrule{\iruleOSubOut} \ %
      \otS \otSub \otT
    }{
      \otOut{\otT} \otSub \otOut{\otS}
    }
    \cinference{%
      \inferrule{\iruleOSubRcdDepth} \ \ \ %
      \otS[i] \otSub \otT[i]\ \ i \in I
    }{
      \otRecord{\ocLab[i]}{\otS[i]}{i \in I} \otSub \otRecord{\ocLab[i]}{\otT[i]}{i \in I}
    }
    \cinference{%
      \inferrule{\iruleOSubVar} \ \ \ %
      \otS[i] \otSub \otT[i]\ \ i \in I
    }{
      \otVariant{\ocLab[i]}{\otS[i]}{i \in I} \otSub \otVariant{\ocLab[i]}{\otT_i}{i \in I \cup J}
    }\\[2mm]
    \cinference{%
      \inferrule{\iruleOSubInpCh} \ %
    }{
      \otChan{\otT} \otSub \otInp{\otT}
    }
    \cinference{%
      \inferrule{\iruleOSubInp} \ %
      \otS \otSub \otT
    }{
      \otInp{\otS} \otSub \otInp{\otT}
    }
    \cinference{%
      \inferrule{\iruleOSubTuple}
      \ \ \otS[i] \otSub \otT[i] \ \ i \in I
    }{
      \otS[1]{\otTimes}{\elipc}{\otTimes}\otS[n] \otSub \otT[1]{\otTimes}{\elipc}{\otTimes}\otT[n]
    }
    \cinference{%
      \inferrule{\iruleOSubRecL} \ %
      \otS\subst{\otRecVar}{\otRec{\otRecVar}{\otS}} \otSub \otT%
    }{
      \otRec{\otRecVar}{\otS} \otSub \otT
    }
    \cinference{%
      \inferrule{\iruleOSubRecR} \ %
      \otS \otSub \otT\subst{\otRecVar}{\otRec{\otRecVar}{\otT}}%
    }{%
      \otS \otSub \otRec{\otRecVar}{\otT}
    }
    \end{array}
    \)}
\end{DEF}

Among those, the rules \inferrule{\iruleOSubRecL} and \inferrule{\iruleOSubRecR} realise equi-recursive view of types.
The only non-standard rule is \inferrule{\iruleOSubRcdDepth} which does not allow fields to be removed in the super type.
This simulates OCaml's lack of row polymorphism where positive occurrences of objects
are not allowed to drop fields.
Note that the negative occurrences of objects in \OCaml, which we use in process implementations, for example,  
do have row polymorphism, which correspond to standard record subtyping:
\scalebox{0.87}{\(
  \begin{array}{c}
    \cinference{
      \otS[i] \otSub \otT[i]\ \ i \in I
    }{
      \otRecord{\ocLab[i]}{\otS[i]}{i \in I \cup J} \otSub \otRecord{\ocLab[i]}{\otT[i]}{i \in I}
    }
  \end{array}\)}.  We use standard record subtyping, when typing processes.
Since it permits  removal of fields,  it precisely simulates session subtyping on outputs.
Typing rules for processes are left to
\iftoggle{techreport}{Appendix \S~\ref{def:typingexpr}}{\cite{IRYY2020TechReport}}.

The typing rules for global combinators {\color{modify}(\cfig{fig:typingforglobal})} are defined by the typing judgement of the form
$\ogtEnvEntailsEx{\roleSet}{\otEnv}{\ogtG}{\otT}$
where
$\otEnv$ is a type environment for recursion variables (definition follows),
$\roleSet = \roleP[1],\ldots,\roleP[n]$ is the sequence of roles which participate in $\ogtG$,
and $\otT = \otT[1]\otTimes\cdots\otTimes\otT[n]$
is a product of channel vector types
where each $\otT[i]$ indicates a protocol which the role $\roleP[i]$
must obey.
We use the product-based encoding to closely model
our our implementation and to avoid   
fixing the number of roles $n$ of $\ogtEnd$ combinator
by using {\em variable-length tuples}
(see
\iftoggle{techreport}{%
  Appendix \S~\ref{sec:appimpl}%
}{%
  \cite{IRYY2020TechReport}%
}).

\begin{DEF}[Global combinator typing rules]\label{def:typingforglobal}
\label{def:typingcontext}
\thmstart
A {\em typing context} $\otEnv$ is defined by the following grammar:
$\otEnv \grmeq \otEnvEmpty \grmor \otEnv \otEnvComp
\otEnvMap{\ocX}{\otT}$. 
The judgement 
    $\ogtEnvEntailsEx{\roleSet}{\otEnv}{\ogtG}{\otT}$
is defined by the rules in
\cfig{fig:typingforglobal}.
We say $\ogtG$ is {\em typable with} $\roleSet$ if
$\ogtEnvEntailsEx{\roleSet}{\otEnv}{\ogtG}{\otT}$
for some $\otEnv$ and $\otT$.
If $\otEnv$ is empty, we write $\ogtEnvEntails{}{\ogtG}{\otT}$.
\end{DEF}

The rule \inferrule{\iruleOTGComm} states that 
$\roleP[i]$ has 
an output type 
$\otIntSumSmall{\roleP[j]}{}{\otIntChoice{\mpLab}{\otS}{\otT[i]}}$ 
to $\roleP[j]$ with label $\mpLab$, a payload typed by $\otS$ and 
continuation typed by $\otT[i]$;
a dual input type 
$\otExtSumSmall{\roleP[i]}{}{\otExtChoice{\mpLab}{\otS}{\otT[j]}}$
from $\roleP[j]$ and continuation typed by $\otT[j]$; and 
the rest of the roles are unchanged.

Rule \inferrule{\iruleOTGSub} is the key to obtain full merging using the subtyping relation, and along 
with the rule \inferrule{\iruleOTGChoice}, is a key to ensure the protocol is realisable, and free of communication errors.
The rule \inferrule{\iruleOTGChoice} requires (1) role $\roleP[a]$ to have an
output type to the same destination role $\roleQ$, which satisfies \Rone.  
The output labels $\{\mpLab[k]\}_{k \in K_{i}}$ are mutually disjoint at 
each branch $\ogtG[i]$, and are merged into a single record, which ensures that the choice is deterministic (\Rtwo).  
All other types stay the same, up to subtyping. 
Following requirement \Mone of \Rthree, a non-directed external choices are prohibited.
This is ensured by encoding the sender role of an input type as a record field,
As the two different destination role labels would result in two record types with no join, following subtyping rule  \inferrule{\iruleOSubRcdDepth}, 
a non-directed external choices are safely reported as a type error. Non-directed internal choices are similarly prohibited (\Mtwo).
On the other hand, directed external choices are allowed, as stipulated by \Mone, and ensured by the subtyping relation on variant types \inferrule{\iruleOSubVar}. 
For example, the two input types $\otExtSumSmall{\roleQ}{}{\otExtChoice{\mpLab[1]}{\otS[1]}{\otT[1]}}$  and   $\otExtSumSmall{\roleQ}{}{\otExtChoice{\mpLab[2]}{\otS[2]}{\otT[2]}}$ can be {\color{modify}unified as}
$\otExtSumSmall{\roleQ}{i \in {1, 2}}{\otExtChoice{\mpLab[i]}{\otS[i]}{\otT[i]}}$. 

The rest of the rules are standard.
Rule \inferrule{\iruleOTGRec} is
for recursion; it assigns the recursion variable $\ogtRecVar$ a
sequence of distinct fresh type variables in the continuation which is
later looked up by \inferrule{\iruleOTGRecVar}. 
In $\otFix{\otRecVar}{\otT}$, 
we assign a unit type if the role does not
contribute to the recursion (i.e., $\otT=\otRecVari$ for any
$\otRecVari$), or forms a recursive type $\otRec{\otRecVar}{\otT}$
otherwise. 

\begin{figure}[t]
\begin{center}
  \scalebox{0.90}{\(
    \begin{array}{c}
      \inference{%
        \inferrule{\iruleOTGComm}\quad
        \ogtEnvEntails{\otEnv}{\ogtG}{
          \left(
          \otT[1]
            \otTimes\elipc\otTimes
            \otT[i]
            \otTimes\elipc\otTimes
            \otT[j]
            \otTimes\elipc\otTimes
            \otT[n]
          \right)
          }
        \quad
        \roleP[i],\roleP[j] \in \roleSet
      }{%
        \ogtEnvEntails{\otEnv}{
        \ogtComm{\roleP[i]}{\roleP[j]}{}{\mpLab}{\otS}{\ogtG}
        }{
          \left(
          \otT[1]
            \otTimes\elipc\otTimes
            \otIntSumSmall{\roleP[j]}{}{%
              \otIntChoice{\mpLab}{\otS}{\otT[i]}%
            }
            \otTimes\elipc\otTimes
            \otExtSumSmall{\roleP[i]}{}{%
              \otExtChoice{\mpLab}{\otS}{\otT[j]}%
            }
            \otTimes\elipc\otTimes
            \otT[n]
          \right)
        }%
      }\\[3mm]
      \inference{%
        \inferrule{\iruleOTGChoice}\quad
        \makecell[lb]{
        \ogtEnvEntails{\otEnv}{\ogtG[i]}{
          \otT[1]
            {\otTimes}{\elipc}{\otTimes}
            {\otT[{a-1}]}{\otTimes}
            \otIntSumSmall{\roleQ}{k \in K_i}{%
              \otIntChoice{\mpLab[k]}{\otS[k]}{\otTi[k]}%
            }
            {\otTimes}{\otT[{a+1}]}
            {\otTimes}{\elipc}{\otTimes}
            \otT[n]
          }\\
        K_{j} \cap K_{j'} = \emptyset \text{\ for all\ }{j} \neq {j'}
        \quad
        \forall i \in I
        \quad
        \roleP[a] \in \roleSet\\
        }
      }{%
        \ogtEnvEntails{\otEnv}{\ogtChoiceKwd\ \roleP[a]\ \{\ogtG[i]\}_{i \in I}}{
          \left(
          \makecell[l]{
            \otT[1]
            {\otTimes}{\elipc}{\otTimes}
            {\otT[{a-1}]}{\otTimes}
            \otIntSumSmall{\roleQ}{k \in {\bigcup_{i \in I}}{K_i}}{%
              \otIntChoice{\mpLab[k]}{\otS[k]}{\otTi[k]}%
            }
            {\otTimes}{\otT[{a+1}]}
            {\otTimes}{\elipc}{\otTimes}
            \otT[n]
          }\right)
        }%
      }
      \inference{%
        \inferrule{\iruleOTGRecVar} & \qquad
      }{
        \makecell[l]{
        \ogtEnvEntails{
          \otEnv \otEnvComp \otEnvMap{\ogtRecVar}{\otT}
        }{
          \ogtRecVar
        }{
          \otT
        }}
      }
      \\[3mm]
      \inference{%
        \inferrule{\iruleOTGEnd} & \qquad \qquad
      }{
        \ogtEnvEntails{
          \otEnv
        }{
          \ogtEnd
        }{
          \otUnit
          {\otTimes}{\cdots}{\otTimes}
          \otUnit
        }
      }
      \inference{%
        \inferrule{\iruleOTGSub}
        \ %
        \ogtEnvEntails{\otEnv}{\ogtG}{\otS}%
        \ %
        \otS \otSub \otT%
      }{%
        \ogtEnvEntails{\otEnv}{\ogtG}{\otT}%
      }%
      \inference{%
        \inferrule{\iruleOTGRec} &
        \ogtEnvEntails{
          \otEnv
          \otEnvComp
          \otEnvMap{\ogtRecVar}{
           \otRecVar[{\ogtRecVar}{1}]
           {\otTimes}{\cdots}{\otTimes}
           \otRecVar[{\ogtRecVar}{n}]
          }
        }{
          \ogtG
        }{
          \otT[1] {\otTimes}{\cdots}{\otTimes} \otT[n]
        }
      }{
        \ogtEnvEntails{\otEnv}{\ogtRec{\ogtRecVar}{\ogtG}}{
          \otFix{\otRecVar[{\ogtRecVar}{1}]}{\otT[1]}
           {\otTimes}{\cdots}{\otTimes}
          \otFix{\otRecVar[{\ogtRecVar}{n}]}{\otT[n]}}
      }%
    \end{array}
  \)}
\end{center}  
  where
  $\roleSet = \roleP[1], \ldots, \roleP[n]$ and,
  $\otFix{\otRecVar}{\otRecVari}{=}{\otUnit}$ and $\otFix{\otRecVar}{\otT}{=}\otRec{\otRecVar}{\otT}$ otherwise.
  \vspace*{-1mm}
  \caption{The typing rules for global combinators
    \framebox{$\ogtEnvEntails{\otEnv}{\ogtG}{\otT}$}\label{fig:typingforglobal}}
  \vspace*{-1em}
\end{figure}

\begin{EX}[Typing a global combinator]\label{ex:globalcombinatorfirst}\thmstart
  We show that the global combinator
  $\gAuth = \ogtComm{\roleC}{\roleS}{}{\labAuth}{}{\left(\ogtCommBin{\roleS}{\roleC}{\labOk}{}{\ogtEnd}{\labCancel}{}{\ogtEnd}\right)}$ has the following type under $\roleS,\roleC$:\smallskip\\
\centerline{\scalebox{1.0}{$
  \otExtSumSmall{\roleC}{}{
    \otExtChoice{\labAuth}{\otT}{
      \otIntSumSmall{\roleC}{}{\otIntChoice{\labOk}{\otT}{\otUnit},\otIntChoice{\labCancel}{\otT}{\otUnit}}
    }
  }
{\otTimes}
\otIntSumSmall{\roleC}{}{
  \otIntChoice{\labAuth}{\otT}{
    \otExtSumSmall{\roleS}{}{\otExtChoice{\labOk}{\otT}{\otUnit},\otExtChoice{\labCancel}{\otT}{\otUnit}}
  }
}$}}\smallskip\\
First, see that
$\ogtG[1]=\left(\ogtComm{\roleS}{\roleC}{}{\labOk}{}{\ogtEnd}\right)$
  has a typing derivation
  as follows (note that we omit the payload type $\otT$ in global combinators):\smallskip\\
\centerline{\scalebox{1.0}{\(
  \inference{
    \ogtEnvEntailsEx{\roleS,\roleC}{}{\ogtEnd}{\otUnit\otTimes\otUnit}
  }{
    \ogtEnvEntailsEx{\roleS,\roleC}{}{
      \ogtComm{\roleS}{\roleC}{}{\labOk}{}{\ogtEnd}
    }{
    \otIntSumSmall{\roleC}{}{
      \otIntChoice{\labOk}{\otT}{\otUnit}}
    \otTimes
    \otExtSumSmall{\roleS}{}{
      \otExtChoice{\labOk}{\otT}{\otUnit}}
    }
  }
  \)}}\smallskip\\
For $\ogtG[2]=\left(\ogtComm{\roleS}{\roleC}{}{\labCancel}{}{\ogtEnd}\right)$ we have similar derivation.
Then,
type of role $\roleC$ (the second of the tuple) is adjusted
by \inferrule{\iruleOTGSub},
{\footnotesize$\otExtSumSmall{\roleS}{}{
     \otExtChoice{\labOk}{\otT}{\otUnit}}\otSub
  \otExtSumSmall{\roleS}{}{
     \otExtChoice{\labOk}{\otT}{\otUnit},
     \otExtChoice{\labCancel}{\otT}{\otUnit}}$} and
{\footnotesize$\otExtSumSmall{\roleS}{}{
     \otExtChoice{\labCancel}{\otT}{\otUnit}}\otSub
  \otExtSumSmall{\roleS}{}{
     \otExtChoice{\labOk}{\otT}{\otUnit},
     \otExtChoice{\labCancel}{\otT}{\otUnit}}$},
thus we have:\smallskip\\
\centerline{\scalebox{1.0}{\(
\begin{array}{l}
\ogtEnvEntailsEx{\roleS,\roleC}{}{\ogtG[1]}{
  \otIntSumSmall{\roleC}{}{
    \otIntChoice{\labOk}{\otT}{\otUnit}}
  {\otTimes}
  \otExtSumSmall{\roleS}{}{
     \otExtChoice{\labOk}{\otT}{\otUnit},
     \otExtChoice{\labCancel}{\otT}{\otUnit}
}}\\
\ogtEnvEntailsEx{\roleS,\roleC}{}{\ogtG[2]}{
  \otIntSumSmall{\roleC}{}{
    \otIntChoice{\labCancel}{\otT}{\otUnit}}
  {\otTimes}
  \otExtSumSmall{\roleS}{}{
     \otExtChoice{\labOk}{\otT}{\otUnit},
     \otExtChoice{\labCancel}{\otT}{\otUnit}
}}
\end{array}\)}}
Then, by \inferrule{\iruleOTGChoice}, we have the following derivation:\smallskip\\
\centerline{\scalebox{0.85}{\(
\inference{
    \ogtEnvEntailsEx{\roleS,\roleC}{}{\ogtG[1]}{
      \otIntSumSmall{\roleC}{}{
        \otIntChoice{\labOk}{\otT}{\otUnit}}
      \otTimes
      \otExtSumBig{\roleS}{}{
        \makecell[l]{
          \otExtChoice{\labOk}{\otT}{\otUnit},\\
          \otExtChoice{\labCancel}{\otT}{\otUnit}
        }
      }
    }
    \quad
    \ogtEnvEntailsEx{\roleS,\roleC}{}{\ogtG[2]}{
      \otIntSumSmall{\roleC}{}{
        \otIntChoice{\labCancel}{\otT}{\otUnit}}
      \otTimes
      \otExtSumBig{\roleS}{}{
        \makecell[l]{
          \otExtChoice{\labOk}{\otT}{\otUnit},\\
          \otExtChoice{\labCancel}{\otT}{\otUnit}
        }
      }
    }
  }{
    \ogtEnvEntailsEx{\roleS,\roleC}{}{
      \ogtChoiceRaw{\roleS}{\ogtG[1], \ogtG[2]}
    }{
      \makecell[l]{
    \otIntSumSmall{\roleC}{}{
      \otIntChoice{\labOk}{\otT}{\otUnit},
      \otIntChoice{\labCancel}{\otT}{\otUnit}}
    \otTimes %
    \otExtSumSmall{\roleS}{}{
      \otExtChoice{\labOk}{\otT}{\otUnit},
      \otExtChoice{\labCancel}{\otT}{\otUnit}}}
    }
}
\)}}\smallskip\\
Note that, in the above premises, the first element of the tuple specifying the behaviour of choosing role $\roleS$,
namely $\otIntSumSmall{\roleC}{}{\otIntChoice{\labOk}{\otT}{\otUnit}}$
and
$\otIntSumSmall{\roleC}{}{\otIntChoice{\labCancel}{\otT}{\otUnit}}$,
are disjointly combined into
$\otIntSumSmall{\roleC}{}{\otIntChoice{\labOk}{\otT}{\otUnit},\otIntChoice{\labCancel}{\otT}{\otUnit}}$
in the conclusion.
Then, by applying \inferrule{\iruleOTGComm} again, we get the type for $\gAuth$ presented above.
\end{EX}

\subsection{Evaluating Global Combinators to Channel Vectors}
\label{sec:evalglobal}

Channel vectors are data structures which are created 
from a global combinator at initialisation, and used for sending/receiving values from/to
participants. Channel vectors implement multiparty
communications as nested binary io-typed channels.

\begin{DEF}[Channel vectors]\thmstart \emph{Channel vectors} ($\ocC,\ocC',...$) and
\emph{wrappers} ($\ocH,\ocH',...$)
are defined as:\smallskip\\
\centerline{$\begin{array}{c}%
\begin{array}{rlllll}
  \ocC,\ocCi \grmeq & \ocV,\elip \grmor \ocS, \ocSi, \elip
     \grmor \ocTuple{\ocC[1],\elip,\ocC[n]} \grmor
     \ocVariant{\ocLab}{\ocC}
\grmor
     \ocRecord{\ocLab[i]}{\ocC[i]}{i \in I} \grmor \ocRec{\ocX}{\ocC}
     \grmor \ocWrapper{\ocS[i]}{\ocH[i]}{i \in I}\\[1mm]
  \ocH,\ocHi \grmeq & \HOLE \grmor \ocVariant{\ocLab}{\ocH}
              \grmor \ocTuple{\ocC[1],\elip,
                \ocH[k],
                \elip,\ocC[n]} \grmor
      \ocRecordOpen
        \ocRecordEntry{\ocLab[1]}{\ocC[1]},\elip
        , \ocRecordEntry{\ocLab[k]}{\ocH},\elip
        , \ocRecordEntry{\ocLab[n]}{\ocC[n]}\ocRecordClose
\quad \ocLab \grmeq\ \roleP \grmor \mpLab
\\[1mm]
\end{array}
\end{array}$}

\end{DEF}
{\textbf{\emph{Channel vectors}}} $\ocC$ are either {\textbf{\emph{base values}}} $\ocV$ or
{\textbf{\emph{runtime values}}} generated from global combinators
which include {\textbf{\emph{names}}} (simply-typed binary channels) $\ocS,\ocSi,\elip$,
{\textbf{\emph{tuples}}} $\ocTuple{\ocC[1],\elip,\ocC[n]}$,
{\textbf{\emph{variants}}} $\ocVariant{\ocLab}{\ocC}$,
{\textbf{\emph{records}}} $\ocRecord{\ocLab[i]}{\ocC[i]}{i \in I}$, and
{\textbf{\emph{recursive values}}} $\ocRec{x}{\ocC}$ where $x$ is a
bound variable.

We introduce an extra runtime value,
{\textbf{\emph{wrapped names}}} $\ocWrapper{\ocS[i]}{\ocH[i]}{i\in I}$,
inspired by Concurrent ML's {\tt wrap} and {\tt choose} functions \cite{concurrentml},
which are a sequence $[\elip]_{i \in I}$ of pairs of input name $\ocS[i]$ and a {\textbf{\em wrapper}} $\ocH[i]$.
A wrapper $\ocH$ contains a single hole $\HOLE$.
An input on wrapped names $\ocWrapper{\ocS[i]}{\ocH[i]}{i\in I}$
is {\em multiplexed} over the set of names $\{\ocS[i]\}_{i \in I}$.
When a sender outputs value $\ocCi$ on name $\ocS[j]$ ($j \in I$),
the corresponding input waiting on $\ocWrapper{\ocS[i]}{\ocH[i]}{i\in I}$ yields a value $\ocWrapperApp{\ocH[j]}{\ocCi}$
where
the construct $\ocWrapperApp{\ocH}{\ocC}$ denotes a value obtained by replacing the hole $\HOLE$
in $\ocH$ with $\ocC$ (i.e. applying function $\ocH$ to $\ocC$).
We write
$\ocRecvWrapSmall{\ocLab[i]}{\ocS[i]}{\ocC[i]}{i \in I}$ for
$\ocWrapper{\ocS[i]}{\ocVariant{\ocLab[i]}{\ocPair{\HOLE}{\ocC[i]}}}{i \in I}$.

\begin{DEF}[Typing rules for channel vectors]\label{def:typing}\rm
  \cfig{fig:typingforchvec}
  gives the typing rules for channel vectors and wrappers. 
The typing judgement for 
(1)  channel vectors
  has the form
  $\otEnvEntails{\otEnv}{\ocC}{\otT}$;
(2) wrappers 
  has the form
  $\otEnvEntails{\otEnv}{\ocH}{\otH}$ where 
the type for wrappers is defined as $\otH \grmeq
\otWrapper{\otT}{\otS}$; 
We assume that all types in $\otEnv$ are closed. 
\end{DEF}

\begin{figure}[t]
\scalebox{0.91}{%
\begin{subfigure}[t]{0.99\textwidth}
\(\begin{array}{c}
    \inference{%
      \inferrule{\iruleOTCName} & \qquad
    }{%
      \otEnvEntails{\otEnv \otEnvComp \otEnvMap{\ocS}{\otChan{\otT}}}{\ocS}{\otChan{\otT}}%
    }%
    \ %
    \inference{%
      \inferrule{\iruleOTCVar} & \qquad
    }{%
      \otEnvEntails{\otEnv \otEnvComp \otEnvMap{\ocX}{\otT}}{\ocX}{\otT}%
    }%
    \ %
    \inference{%
      \inferrule{\iruleOTCUnit} & \quad
    }{%
      \otEnvEntails{\otEnv}{\ocUnit}{\otUnit}%
    }%
    \ %
    \inference{%
      \inferrule{\iruleOTCSub}\ %
      \otEnvEntails{\otEnv\!}{\!\ocC\!}{\!\otS}%
      \ %
      \otS{\otSub}\otT%
    }{%
      \otEnvEntails{\otEnv}{\ocC}{\otT}%
    }%
    \ %
    \inference{%
      \inferrule{\iruleOTCTuple}
      \ %
      \otEnvEntails{\otEnv}{\ocC[i]}{\otT[i]}%
      \ %
      \forall i, 1{\le}i{\le}n
    }{%
      \otEnvEntails{\otEnv}{\ocTuple{\ocC[1],{\elip},\ocC[n]}}{\otT[1]{\otTimes}{\elipc}{\otTimes}\otT[n]}
    }\\[1mm]%
    \inference{%
      \inferrule{\iruleOTCVariant}
      \ %
      \otEnvEntails{\otEnv}{\ocC}{\otT}%
    }{%
      \otEnvEntails{\otEnv}{\ocVariant{\ocL}{\ocC}}{\otVariant{\ocL}{\otT}{}}
    }%
    \ %
    \inference{%
      \inferrule{\iruleOTCRecord}
      \ %
      \otEnvEntails{\otEnv}{\ocC[i]}{\otT[i]}\ %
      \forall i \in I
    }{%
      \otEnvEntails{\otEnv}{\ocRecord{\ocL[i]}{\ocC[i]}{i \in I}}{\otRecord{\ocL[i]}{\otT[i]}{i \in I}}
    }%
    \ %
    \inference{%
      \inferrule{\iruleOTCRec}
      \ %
      \otEnvEntails{\otEnv \otEnvComp \otEnvMap{\ocX}{\otRec{\otRecVar}{\otT}}}{\ocC}{\otT\subst{\otRecVar}{\otRec{\otRecVar}{\otT}}}
    }{%
      \otEnvEntails{\otEnv}{\ocRec{\ocX}{\ocC}}{\otRec{\otRecVar}{\otT}}
    }\\[1mm]%
    \inference{%
      \inferrule{\iruleOTCWrapInp}
      \ %
      \otEnvEntails{\otEnv}{\ocS[i]}{\otInp{\otS[i]}}%
      \ %
      \otEnvEntails{\otEnv}{\ocH[i]}{\otWrapper{\otT}{\otS[i]}}
      \ %
      \forall i{\in}I
    }{%
      \otEnvEntails{\otEnv}{\ocWrapper{\ocS[i]}{\ocH[i]}{i \in I}}{\otInp{\otT}}
    }%
    \ %
    \inference{%
      \inferrule{\iruleOTCWrapper}
      \ %
      \otEnvEntails{\otEnv \otEnvComp \otEnvMap{\ocX}{\otTi}}{\ocC}{\otT}%
      \ %
      \ocC{=}\ocWrapperApp{\ocH}{\ocX}
      \ %
      x{\notin}\fv{h}
    }{%
      \otEnvEntails{\otEnv}{\ocH}{\otWrapper{\otT}{\otTi}}
    }%
  \end{array}\vspace*{-2mm}
\)
\end{subfigure}
}
\caption{
  The typing rules for channel vectors and wrappers
  \framebox{$\otEnvEntails{\otEnv}{\ocC}{\otT}$}
  \framebox{$\otEnvEntails{\otEnv}{\ocH}{\otH}$}
  \label{fig:typingforchvec}}
\end{figure}

The rules for channel vectors are standard 
where the subtyping relation in rule 
\inferrule{\iruleOTCSub} is defined at \Cref{def:subtyping}. %
For wrappers, 
rule \inferrule{\iruleOTCWrapInp} types wrapped names
where the payload type $\otSi$ of input channel $\ocS$ is the same as the hole's type,
and all wrappers have the same result type $\otT$.
Rule \inferrule{\iruleOTCWrapper} 
checks type of a channel vector $\ocC=\ocWrapperApp{\ocH}{\ocX}$ and replaces $\ocX$ with the hole $\HOLE$.

{\em Evaluation} of global combinators is the key to implement
a multiparty protocol to a series of binary, simply-typed communications based on channel vectors.
We define $\GCCV{\ogtG}$ where
$\roleSet$ is a sequence of roles in $\ogtG$ and
$\ocS$ is a {\em base name} freshly assigned to an initiation expression
at runtime.
The generated  channels are interconnected to each other and the
created channel vectors are distributed and shared among expressions running in parallel,
enabling them to interact via binary names. %

The followings are basic operations on records, tuples and recursive values
which are used to define evaluations of global combinators.
\begin{DEF}[Operations]\label{def:operations}\thmstart
{\bf (1)}
  The {\em unfolding} $\ocUnfold{\ocC}$ of a recursive value is defined by
  the smallest $n$ such that
  $\ocUnfoldN{n}{\ocC}=\ocUnfoldN{n+1}{\ocC}$,
  and
  $\ocUnfoldOne{\cdot}$ is defined as:\smallskip\\
  \noindent\centerline{\(
    \begin{array}{rcllrcll}
      \ocUnfoldOne{\ocRec{\ocX}{\ocC}}& = &\ocC\subst{\ocX}{\ocRec{\ocX}{\ocC}}  & \qquad\qquad\qquad &
      \ocUnfoldOne{\ocC} & = & \ocC & \text{otherwise}
    \end{array}
    \)}\smallskip\\
where $f^{n+1}(x)=f(f^{n}(x))$ for $n \geq 2$ and $f^{1}(x)=f(x)$.
{\bf (2)}
$\ocC\#\ocLab$ denotes
the {\textbf{\emph{record projection}}}, which
projects on field $\ocLab$ of record value $\ocC$, defined as:
$\ocRecord{\ocLab[i]}{\ocC[i]}{i\in I}\# \ocLab[k] = \ocUnfold{\ocC[k]}$,
where
$\#$ is left-associative, i.e.
$\ocC\# \ocLab[1]\#... \#\ocLab[n]  = ((...(\ocC\# \ocLab[1])\#...)
\#\ocLab[n])$.
{\bf (3)}
The $i$-th projection on a tuple, $\NTH{\ocC}{i}$ is defined as
  $\NTH{\ocTuple{\ocC[1],\elip,\ocC[n]}}{i}{=}\ocC[i]$ for $1 \le i \le n$.
{\bf (4)}
$\ocFix{\ocX}{\ocXi}{=}\ocUnit$; otherwise
$\ocFix{\ocX}{\ocC}{=}\ocRec{\ocX}{\ocC}$.
\end{DEF}

\begin{DEF}[Evaluation of a global combinator]\label{def:global:evaluation} \thmstart
  Given $\roleSet$ and
fresh $s$, the {\em evaluation}
$\GCCVR{\gocaml}{\roleSet}^s$ of global combinator $\gocaml$
is defined in \cfig{fig:gceval}.
We write $\GCCVR{\gocaml}{}^s$ if $\roleSet=\otRolesSet{\gocaml}$.
\end{DEF}

\begin{figure}[t]
  \noindent\centerline{\scalebox{0.91}{\(
    \begin{array}{rcl}
    \GCCV{\ogtComm{\roleP[j]}{\roleP[k]}{}{\mpLab}{\otS}{\ogtG}}
      & =
      & \\
    \multicolumn{3}{r}{
      \makecell{
        \oeFmt{
          \Bigl\ocTupOpen
        \NTH{\GCCV{\gocaml}}{1},
        \elip,
        \NTH{\GCCV{\gocaml}}{j{-}1},
        \ \ocIntSum{\roleP[k]}{}{
          \ocIntChoiceSmall{\mpLab}{\ocS_{\{\roleP[j],\roleP[k],\mpLab,\blueI\}}}{\NTH{\GCCV{\ogtG}}{j}}},
        \NTH{\GCCV{\gocaml}}{j{+}1},}\\
        \oeFmt{
        \qquad\qquad\qquad\qquad \elip, \NTH{\GCCV{\gocaml}}{k{-}1},
          \ocExtSum{\roleP[j]}{}{
            \ocExtChoiceSmall{\mpLab}{\ocS_{\{\roleP[j],\roleP[k],\mpLab,\blueI\}}}{\NTH{\GCCV{\ogtG}}{k}}},
        \ \NTH{\GCCV{\gocaml}}{k{+}1},
        \elip, \NTH{\GCCV{\gocaml}}{n}
        \Bigr\ocTupClose
        }}}
    \\
    & \multicolumn{2}{l}{\text{where\ }\blueI\text{ is fresh.}}\\
    \GCCV{\ogtChoice{\roleP[a]}{\ogtG[i]}{i \in I}} & = &\\
    \multicolumn{3}{r}{
    \oeFmt{\left\ocTupOpen
      \ocMerge{i \in I}{\left(\NTH{\GCCV{\ogtG[i]}}{1}\right)},\elip,
      \ocMerge{i \in I}{\left(\NTH{\GCCV{\ogtG[i]}}{a{-}1}\right)},
      \ocIntSumSmall{\roleQ}{k \in K}{
        \mpLab[k]{=}\ocC[k]
      },
      \ocMerge{i \in I}{\left(\NTH{\GCCV{\ogtG[i]}}{a{+}1}\right)},
      \elip,
      \ocMerge{i \in I}{\left(\NTH{\GCCV{\ogtG[i]}}{n}\right)}
      \right\ocTupClose}} \\
    & \multicolumn{2}{l}{\text{where\ }
      \ocUnfold{\NTH{\GCCV{\ogtG[i]}}{a}} = \ocIntSumSmall{\roleQ}{k \in K_i}{{\mpLab[k]}{=}\ocC[k]}
      \ \text{and}\ %
      K = \bigcup_{i \in I}{K_i}}\\
    \GCCV{\ogtRec{x}{\gocaml}} & = &
\oeFmt{\left\ocTupOpen
        \ocFixSmall{\ocX[1]}{\NTH{\GCCV{\gocaml}}{1}}, \elip,
        \ocFixSmall{\ocX[n]}{\NTH{\GCCV{\gocaml}}{n}}\right\ocTupClose}
\\[1mm]
\GCCV{\ogtRecVar}  & = & \oeFmt{\bigl\ocTupOpen \ocX[1], \elip,
  \ocX[n]\bigr\ocTupClose}
\quad\quad\quad\quad\quad
\GCCV{\ogtEnd} = \oeFmt{\bigl\ocTupOpen \ocUnit, \elip,
  \ocUnit\bigr\ocTupClose}
    \end{array}
  \)}}\smallskip\\
\vspace*{-3mm}
\caption{Evaluation of global combinators \framebox{$\GCCV{\ogtG}$}\label{fig:gceval}}
\end{figure}

The evaluation for communication
$\ogtComm{\roleP[j]}{\roleP[k]}{}{\mpLab}{\otS}{\ogtG}$
connects between
$\roleP[j]$ and $\roleP[k]$
by the name $\ocS_{\{\roleP[j],\roleP[k],\mpLab,\blueI\}}$
by wrapping $j$-th and $k$-th channel vector with an output and an input structure, respectively.
The name $\ocS_{\{\roleP[j],\roleP[k],\mpLab,\blueI\}}$
is indexed by two role names $\roleP[j]$, $\roleP[k]$, label $\mpLab$ and an index $\blueI$
so that
(1) it is only shared between two roles $\roleP[j]$ and $\roleP[k]$,
(2) communication only occurs when it tries to communicate
a specific label $\mpLab$, and
(3) both the sender and the receiver agree on the payload type.
Here, the index $\blueI$ is used to distinguish between
names generated from the same label $\mpLabi$ but
different payload type $\mpLab{:}\otT$ and $\mpLab{:}\otTi$, ensuring
consistent typing of generated channel vectors.
The choice combinator
$\ogtChoice{\roleP[a]}{\ogtG[i]}{i \in I}$
extracts the output channel vector (i.e. the nested records of the form
$\ocIntSumSmall{\roleQ}{k \in K_i}{\mpLab[k]{=}\ocC[k]}$)
at $\roleP[a]$ from each branch $\ogtG[i]$,
and merges them into a single output.
Channel vectors for the other roles are merged by $c_1 \ocBinMerge
c_2$
where merging for the outputs is
an intersection of branchings from $c_1$ and $c_2$, while
merging of the inputs is their union. We explain merging by example
(\Cref{ex:auth:gen}) and leave the full definition in
\iftoggle{techreport}{%
  \cfig{fig:merge} in \S~\ref{subsec:merge}%
}{%
  \cite{IRYY2020TechReport}%
}.

For the recursion combinator,
function $\ocFixSmall{\ocX[i]}{\ocC[i]}$ forms
a recursive value for repetitive session, or
voids it as $\ocUnit$ if it does not contain any names.

\begin{EX}[Global combinator evaluation]\thmstart \label{ex:auth:gen}
Let
$\ocS[1]=\ocS_{\{\roleC,\roleS,\labOk,0\}}$,
$\ocS[2]=\ocS_{\{\roleC,\roleS,\labCancel,0\}}$ and
$\ocS[3]=\ocS_{\{\roleS,\roleC,\labAuth,0\}}$. Then:\smallskip\\
\centerline{\scalebox{1.00}{\(
\begin{array}{l}
\GCCVR{\gAuth}{}^{\mpS} \qquad\\
= \GCCVR{\ogtComm{\roleC}{\roleS}{}{\labAuth}{}{\left(\ogtCommBin{\roleS}{\roleC}{\labOk}{}{\ogtEnd}{\labCancel}{}{\ogtEnd}\right)}}{}^{\mpS}\\
\qquad \left(
\makecell[l]{
  \text{Here, we have}\ \left\{
\makecell[l]{
  \ogtG[L]=\ogtComm{\roleC}{\roleS}{}{\labOk}{}{\ogtEnd}, \quad
  \ogtG[R]=\ogtComm{\roleC}{\roleS}{}{\labCancel}{}{\ogtEnd},\\
  \GCCVR{\ogtG[L]}{}^{\mpS}=
    \ocRecordOpen
    \ocExtSumSmall{\roleS}{}{
      \ocExtChoiceSmall{\labOk}{\ocS[1]}{\ocUnit}},
    \ocIntSumSmall{\roleC}{}{
      \ocIntChoiceSmall{\labOk}{\ocS[1]}{\ocUnit}}
    \ocRecordClose,\\
  \GCCVR{\ogtG[R]}{}^{\mpS}=
    \ocRecordOpen
    \ocExtSumSmall{\roleS}{}{
      \ocExtChoiceSmall{\labCancel}{\ocS[2]}{\ocUnit}},
    \ocIntSumSmall{\roleC}{}{
      \ocIntChoiceSmall{\labCancel}{\ocS[2]}{\ocUnit}}
    \ocRecordClose,\\
}\right\},\\
  \text{concatenating}\ \left\{
\makecell[l]{
  \ocUnfold{\NTH{\GCCVR{\ogtG[L]}{}^{\mpS}}{2}} =
  \NTH{\GCCVR{\ogtG[L]}{}^{\mpS}}{2} =
    \ocRecordSmall{\roleS}{\ocRecordSmall{\labOk}{\ocC[{L2}]}{}}{},
  \ocC[{L2}]{=}\ocPairSmall{\ocS[1]}{\ocUnit},\\
  \ocUnfold{\NTH{\GCCVR{\ogtG[R]}{}^{\mpS}}{2}} =
  \NTH{\GCCVR{\ogtG[R]}{}^{\mpS}}{2} =
    \ocRecordSmall{\roleS}{\ocRecordSmall{\labCancel}{\ocC[{R2}]}{}}{},
  \ocC[{R2}]{=}\ocPairSmall{\ocS[2]}{\ocUnit}
}\right\}
}\right)
\\
=
\ocTuple{\makecell[l]{
\ocIntSumSmall{\roleS}{}{
  \ocIntChoiceSmall{\labAuth}{\ocS[3]}{
    \NTH{\GCCVR{\ogtG[L]}{}^{\mpS}}{1}
    \ocBinMerge
    \NTH{\GCCVR{\ogtG[R]}{}^{\mpS}}{1}
    }},
\,\,
\ocExtSumSmall{\roleC}{}{
  \ocExtChoiceSmall{\labAuth}{\ocS[3]}{
    \ocRecordSmall{\roleC}{
      \ocRecordOpen
        \ocRecordEntry{\labOk}{\ocC[{L2}]} \ocRecordSep
        \ocRecordEntry{\labCancel}{\ocC[{R2}]}
      \ocRecordClose
    }{}}}
}}
\\
=
\left(%
\makecell[l]{
\ocIntSumSmall{\roleS}{}{
  \ocIntChoiceSmall{\labAuth}{\ocS[3]}{
    \ocExtSumSmall{\roleS}{}{
    \ocExtChoiceSmall{\labOk}{\ocS[1]}{\ocUnit},
    \ocExtChoiceSmall{\labCancel}{\ocS[2]}{\ocUnit}}}},\\
\quad\ocExtSumSmall{\roleC}{}{
  \ocExtChoiceSmall{\labAuth}{\ocS[3]}{
    \ocIntSumSmall{\roleC}{}{
    \ocIntChoiceSmall{\labOk}{\ocS[1]}{\ocUnit},
    \ocIntChoiceSmall{\labCancel}{\ocS[2]}{\ocUnit}}}}
}\right)%
\end{array}
\)}}
\end{EX}

The following main theorem states that if a global combinator is 
typable, the generated channel vectors are well-typed
under the corresponding local types. 

\begin{restatable}[Realisability of global combinators]{THM}{colSubjectReductionForGC}\label{col:SubjectReductionForGC}\thmstart
  If $\ogtEnvEntails{}{\ogtG}{\otT}$,
  then ${\GCCV{\gocaml}} = \ocC$ is defined and
  $\otEnvEntails{\{\otEnvMap{\ocS[i]}{\otS[i]}\}_{\ocS[i]\in \fn{\ocC}}}{\ocC}{\otT}$
  for some $\{\widetilde{\otS[i]}\}$.
\end{restatable}

This property offers the type soundness and communication safety for
\must endpoint programs:  
a statically well-typed \must program will satisfy subject reduction
theorem and never
performs a non-compliant I/O action w.r.t. the underlying binary channels. %
We leave the formal definition of 
\must endpoint programs, operational semantics, typing system, 
and the subject reduction theorem in
\iftoggle{techreport}{%
  \S~\ref{app:mio}%
}{%
  \cite{IRYY2020TechReport}%
}.

\section{Implementing Global Combinators} \label{sec:implementation}
We give a brief overview on the type
manipulation techniques that enable type checking of global
combinators in native \OCaml. 
\S~\ref{sec:proofsystem} gives a high-level intuition of our approach, 
\S~\ref{subsec:gcevalimpl} illustrates evaluation of global combinators to channel vectors in pseudo \OCaml code,
and \S~\ref{subsec:chanvecimpl} presents the typing of global combinators in \OCaml.
Furthermore,
in
\iftoggle{techreport}{%
  Appendix \S~\ref{sec:appimpl}%
}{%
  \cite{IRYY2020TechReport}%
},
we develop {\em variable-length tuples}
using state-of-art functional programming techniques, e.g., GADT and polymorphic variants,
to improve usability of \must.

\label{sec:gcimpl}
\subsection{Typing Global Combinators in OCaml: A Summary} %
\label{sec:proofsystem}

\begin{figure}
  \begin{adjustbox}{width=\columnwidth,center}
    \begin{tabular}{l;{7pt/10pt}l}
      \toprule
      Global Combinator & Type\\
      \toprule
      \TableStrutDouble
      \lstinline!finish!
      &
      \begin{tabular}{l}
        \lstinline!(close._*._,,._*._close)!
      \end{tabular}\\
      \hline
      \lstinline!($r_i$ -->$\ r_j$)._m._._$g$!
      &
      \begin{tabular}{l}
        Given $g$ : \lstinline!($t_{r_1}$._*._,,._*._$t_{r_n}$)!,\\
        Return \lstinline!($t_{r_1}$._*._,,._*._<$r_j$:._<m:._('v._*._$t_{r_i}$)._out>>._*._,,._*._<$r_i$:._[>._`m._of._'v._*._$t_{r_j}$]._inp>._*._,,._*._$t_{r_n}$)!
      \end{tabular}
      \\
      \hline
      \TableStrutTriple
      \begin{lstlisting}
choice_at $r_a$ $\mrg$
($r_a$, $g_1$)
($r_a$, $g_2$)
      \end{lstlisting}
      & 
      \begin{tabular}{l}
        Given $1 \le a \le n$,\\
        \, $g_1$ : \lstinline!($t_{r_1}$._*._,,._*._$t_{r_{a-1}}$*._<$r_b$:._<m$_i$:._($v_i$,._$s_i$)._out>$_{i \in I}$>._*._$t_{r_{a+1}}$*._,,._*._$t_{r_n}$)!,\\
        \, $g_2$ : \lstinline!($t_{r_1}$._*._,,._*._$t_{r_{a-1}}$*._<$r_b$:._<m$_j$:._($v_j$,._$s_j$)._out>$_{j \in J}$>._*._$t_{r_{a+1}}$*._,,._*._$t_{r_n}$)!, and\\
        \,  $\mrg$ : a concatenator ensuring the two label sets are mutually disjoint ($I \cap J = \emptyset$),\\
        Return \lstinline!($t_{r_1}$._*._,,._*._$t_{r_{a-1}}$*._<$r_b$:._<m$_k$:._($v_k$,._$s_k$)._out>$_{k \in I \cup J}$>._*._$t_{r_{a+1}}$*._,,._*._$t_{r_n}$)!\\
      \end{tabular}
      \\
      \hline
      \TableStrut
      \lstinline!fix._(fun $x$ ->$\,g$)!
      &
      \begin{tabular}{l}
        Given $g$ : \lstinline!($t_{r_1}$._*._,,._*._$t_{r_n}$! under assumption that $x$ : \lstinline!($t_{r_1}$._*._,,._*._$t_{r_n}$)!,\\
        \, $x$ is guarded in $g$\\
        Return \lstinline!($t_{r_1}$._*._,,._*._$t_{r_n}$)!
      \end{tabular}
      \\
      \hline
      \TableStrut
      \lstinline!closed_at $r_a$ $g$!
      &
      \begin{tabular}{l}
        Given $g$ : \lstinline!($t_{r_1}$._*._,,._*._$t_{r_{a-1}}$._*._close._*._$t_{r_{a+1}}$._*._,,._*._$t_{r_n}$)! and $1 \le a \le n$,\\
        Return \lstinline!($t_{r_1}$._*._,,._*._$t_{r_{a-1}}$._*._close._*._$t_{r_{a+1}}$._*._,,._*._$t_{r_n}$)!
      \end{tabular}
      \\
      \bottomrule
    \end{tabular}%
  \end{adjustbox}
  \caption{Type of Global Combinators in OCaml}
  \label{fig:gctypeupdate}
\end{figure}

In \cfig{fig:gctypeupdate} we illustrate
the type signature of each global combinator,
which is a transliteration of the typing rules (\cfig{fig:typingforglobal}) into OCaml.
In the figure, OCaml type \lstinline!($t_{r_1}$._*._,,._*._$t_{r_n}$)! corresponds to
a $n$-tuple of channel vector types $t_{r_1} \times\cdots\times t_{r_n}$. 
\precameraready{
  The implementation makes use of \textit{variable-length tuples}
  to represent tuples of channel vectors,
  and therefore the developer does not have to explicitly specify
  the number of roles  $n$
  (see
  \iftoggle{techreport}{%
    Appendix \S~\ref{sec:appimpl}%
  }{%
    \cite{IRYY2020TechReport}%
  }).}
A few type-manipulation techniques are expanded later in \S~\ref{subsec:chanvecimpl}.
Henceforth, we only make a few remarks, regarding some discrepancies with the implementation. 

\begin{wrapfigure}{r}{0.53\textwidth}
  \begin{tabular}{l|l}
 \OCaml types & Types in \S~\ref{sec:formalism} \\ %
    \hline
       \lstinline!<r:[>`m$_i$._of._$v_i$*$t_i$]$_{i \in I}$._inp>!
       & $\otExtSumSmall{\roleR}{i \in I}{\otExtChoice{\mpLab[i]}{\otS[i]}{\otT[i]}}$
       \\
   \lstinline!<r:<m$_i$:($v_i$,$t_i$)._out>$_{i \in I}$>!
       & $\otIntSumSmall{\roleR}{i \in I}{\otIntChoice{\mpLab[i]}{\otS[i]}{\otT[i]}}$
       \\
   \lstinline!close! ($=$\lstinline!unit!) &
   $\otUnit$
   \\
   \lstinline!$t$ as 'x! &
   $\otRec{x}{\otT}$
  \end{tabular}
  \vspace*{-1em}
\end{wrapfigure}
\myparagraph{Channel vector types in OCaml.}
The \OCaml syntax of channel vector types is given on the right. 
The difference with its formal counterparts are minimal. 
In particular, records are implemented using \OCaml object types, and record fields 
correspond to object methods, i.e. \oCODE{role_$\roleQ$} is a method. 
In type \lstinline![>`m$_i$._of._$t_i$]$_{i \in I}$!,
the symbol \lstinline!>! marks an {\em open} polymorphic variant type which can have more tags.
The types \oCODE{inp} and \oCODE{out} stand for an input and output types with a payload type $v_i$ and a continuation $t_i$.
Recursive channel vector types are implemented using OCaml equi-recursive types. 

\myparagraph{On branching and compatibility checking.}
As we explained in \S~\ref{sec:typing:global}, branching %
is the key to ensure the protocol is realisable, and free of communication errors.
To ensure that the choice is deterministic, %
it must be verified that the set of labels in each branch are disjoint.
Since \OCaml objects do not support {\em concatenation} (combining of multiple methods e.g., \cite{DBLP:journals/iandc/Wand91,DBLP:conf/popl/HarperP91}),
and cannot automatically verify that the set of labels (encoded as object methods) are disjoint,
the user has to manually write a disjoint merge function $\mrg$ that concatenates two objects with different methods into one
(see \iftoggle{techreport}{\S~\ref{sec:disj_merge}}{\cite{IRYY2020TechReport}} for examples).
This part can be completely automated by PPX syntactic extension in OCaml.
On compatibility checking of non-choosing roles,
external choice \lstinline!<$r$:._[>`m1._of,,]._inp>! and   \lstinline!<$r$:._[>`m2._of,,]._inp>!,  the types can be recursively
merged by \OCaml type inference to  \lstinline!<$r$:._[>`m1._of,,|`m2._of,,]._inp>!
thanks to the row polymorphism on polymorphic variant types (\oCODE{>}),
while non-directed external choices
and other incompatible combination of types (e.g., input and output, input and closing, and output and closing) are statically excluded.
\myparagraph{On unguarded recursion.}
The encoding of recursion \lstinline!fix._(fun $x$ ->$\,g$)! has two caveats w.r.t the typing system:
(1) \OCaml does not check if a recursion is guarded, thus for example \lstinline!fix (fun x ->._._x)! is allowed.
{\color{modify} We cannot use OCaml value recursion, because global combinators generate channels at run-time.}
(2) Even if a loop is guarded, Hindley-Milner type inference may introduce arbitrary local type at some roles.
For example, consider the global protocol \lstinline!fix._(fun x ->._._($r_a$._--> $r_b$) msg x)!
which specifies an infinite loop for roles $\notin \{r_a, r_b\}$, and does not specify any behaviour for any other roles. 
To prevent undefined behaviour, the typing rule marks the types of the roles that are not used as closed $\otFix{\otRecVar}{\otT}$. 
Unfortunately, in type inference, we do not have such control, and the above protocol will introduce
a polymorphic type \lstinline!'t$_{r_i}$! for role $r_i \notin \{r_a, r_b\}$,
which can be instantiated by {\em any} local type.  

\myparagraph{Fail-fast policy.}  We regard the above intricacies on recursion as a {\em fact of life} %
in any programming language, %
and provide a few workarounds.
For (1), we adopt a ``fail-fast'' policy: Our library throws an exception if there is an unguarded occurrence of a recursion variable.
This check is performed when evaluating a global combinator before any communication is started.
As for (2), we require the programmer to adhere to a coding convention when specifying an infinite protocol. They have to insert
additional combinator \lstinline!closed_at $r_a$ $g$!, 
which consistently instantiates type variable \lstinline!'t$_{r_a}$! with \lstinline!close!,
leaving other roles intact.
If the programmer forgets this insertion, fail-fast approach applies, and our library throws 
a runtime exception before the protocol has started. %
In addition, self-sent messages \oCODE{(r --> r) msg} for any \oCODE{r} are reported as an error at runtime.}

\subsection{Implementing Global Combinator Evaluation}
\label{subsec:gcevalimpl}

Following \S~\ref{sec:evalglobal}, in \cfig{fig:impl_gc},
we illustrate the implementation of the global combinators,
by assuming that method names and variant tags are {\em first class} in this pseudo-OCaml.
Communication combinator \lstinline!(-->)! is presented in \cfig{fig:impl_gc} (a)
where the communication combinator (\lstinline!($r_i$._-->._$r_j$)._$m$._$g$!)
yields two reciprocal channel vectors of type
\lstinline!<$r_j$:<$m$:._($v$,$t_{r_i}$)._out>>! and \lstinline!<$r_i$:[>`$m$._of._$v$*$t_{r_j}$]._inp>!.

The implementation starts by extracting the continuations (the channel vectors) at each role (Line~\ref{line:cc1}).
Line~\ref{line:cc2} creates a fresh new channel \lstinline!&s&! of a polymorphic type \lstinline!'v._channel! shared among two roles,
which is a source of type safety regarding {\em payload} types.
Line~\ref{line:cc3} creates an output channel vector.
We use a shorthand \lstinline!<$m$ = $e$>! to represent an OCaml object \lstinline!object method $m$ = $e$ end!.
Thus, it is bound to \lstinline!c$_{r_i}$!, by nesting the pair \lstinline!(&s&,c$_{r_i}$)! inside two objects,
one with a method role, and another with a method label,
forming type \lstinline!<$r_j$:<$m$:._('v,$t_{r_i}$)._out>>!.
Similarly, Line~\ref{line:cc4} creates an input channel vector \lstinline!c$_{r_j}$!, 
by wrapping channel \lstinline!&s&! in a polymorphic variant using \lstinline!Event.wrap! from Concurrent ML and
nesting it in an object type,
forming type \lstinline!<$r_i$:[>`$m$._of._'v*$t_{r_j}$]._inp>!.
This wrapping relates tag $m$ and continuation $t_j$ to the input side, enabling external choice when merged. 
Finally, the newly updated tuple of channel vectors is returned (Line~\ref{line:cc6}). 

\begin{figure}[t]
\begin{minipage}{0.5\textwidth}
{\lstset{numbers=left}
\begin{OCAMLLISTING}
let (-->) $r_i$ $r_j$ $m$ $g$ = 
 (* extract the continuations *)
  let (c$_{r_1}$, c$_{r_2}$, $\elip$ , c$_{r_n}$) = $g$ in  ^\label{line:cc1}^
  let &s& = Event.new_channel () in  ^\label{line:cc2}^
   (* create an output channel vector  *) 
  let c$_{r_i}$ = (<$r_j$ = <$m$ = (&s&,c$_{r_i}$)> >)$\ $in  ^\label{line:cc3}^
   (* create an input channel vector  *)
  let c$_{r_j}$ = (<$r_i$ =  ^\label{line:cc4}^
      Event.wrap &s& (fun x -> `$m$(x,c$_{r_j}$)) >) in
  (c$_{r_1}$, c$_{r_2}$, $\elip$ , c$_{r_n}$)  ^\label{line:cc6}^
\end{OCAMLLISTING}}
\end{minipage}
\vspace{10pt}
\begin{minipage}{0.4\textwidth}
\begin{OCAMLLISTING}
 let choice_at $r_a$ $\mrg$ $g_1$ $g_2$ = 
  let (c1$_{r_1}$, c1$_{r_2}$, $\elip$ , c1$_{r_n}$) = $g_1$ in  ^\label{line:bc1}^
  let (c2$_{r_1}$, c2$_{r_2}$, $\elip$ , c2$_{r_n}$) = $g_2$ in  ^\label{line:bc2}^
  let c$_{r_a}$ = 
    $\mbox{\rm (concatenate}$ c1$_{r_a}$ $\mbox{\rm and}$ c2$_{r_a}$ $\mbox{\rm using}$ $\mrg\mbox{\rm)}$ in
  let c$_{r_1}$ = merge c1$_{r_1}$ c2$_{r_1}$ in ^\label{line:bc3}^
  let c$_{r_2}$ = merge c1$_{r_2}$ c2$_{r_2}$ in 
  (* .. repeatedly merge each $r_i \neq r_a$ .. *)
  let c$_{r_n}$ = merge c1$_{r_n}$ c2$_{r_n}$ in ^\label{line:bc4}^
  (c$_{r_1}$, c$_{r_2}$, $\elip$ , c$_{r_n}$)
\end{OCAMLLISTING}
\end{minipage}
\caption{Implementation of  communication combinator and (a) branching combinator (b)}
\label{fig:impl_gc}
\end{figure}

\begin{figure}[b]
{\lstset{numbers=left}
\begin{OCAMLLISTING}
(* the  definition of the type method_*)
type ('obj, 'mt) method_ = {make_obj: 'mt -> 'obj; call_obj: 'obj -> 'mt} ^\label{line:method}^
(* example usage of method_: *)^\label{line:methodusagestart}^
val login_method : (<login : 'mt>, 'mt) method_ (* the type of login_method *)
let login_method = 
  {make_obj=(fun v -> object method login = v end); call_obj=(fun obj -> obj#login)} ^\label{line:methodusage}^

(* the  definition of the type label*)
type ('obj, 'ot, 'var, 'vt) label = {obj: ('obj, 'ot) method_; var: 'vt -> 'var} ^\label{line:lable}^
(* example usage of label *)
val login : (<login : 'mt>, 'mt, [> `login of 'vt], 'vt) label
let login = {obj=login_method; var=(fun v -> `login(v))} ^\label{line:lableusage}^

 (* example usage of role: *)
 let s = {index=Zero; ^\label{line:roleusage}^ 
 label={make_obj=(fun v -> object method role_S=v end); call_obj=(fun o -> o#role_S)}} 
\end{OCAMLLISTING}}
\caption{Implementation of first-class methods and labels}
\label{fig:labels}
\end{figure}

\cfig{fig:impl_gc} (b) illustrates the choice combinator \lstinline!choice_at!.
Line~\ref{line:bc3}--\ref{line:bc4} specifies that the channel vectors at non-choosing roles are {\em merged}, using a \oCODE{merge} function. 
Intuitively, \oCODE{merge} does a type-case analysis on the type of channel vectors, as follows: 
(1) for an input channel vector, it makes an {\em external choice} among (wrapped) input channels, using the 
\lstinline!Event.choose! function from Concurrent ML;
(2) for an output channel vector, the bare channel is {\em unified} label-wise,
in the sense that an output on the unified channel can be observed on both input sides,
which is achieved by having channel type around a reference cell;
and (3) handling of channel vector of type \oCODE{close} is trivial.

\myparagraph{First-class methods.} Method names $r_i$, $r_j$ and $m$ and the variant tag $m$ occurring in (\lstinline!($r_i$._-->._$r_j$)._$m$._$g$!) are assumed in \S~\ref{sec:proofsystem} to be  first-class values. Since such behaviour is not readily available in vanilla OCaml,
we simulate it by introducing the type \oCODE{method_} (Line~\ref{line:method} in \cfig{fig:labels}), which creates values that behave like method objects. The type is a record with a  {\em constructor function} \oCODE{make_obj} and
a {\em destructor function} \oCODE{call_obj} (see example in Lines~\ref{line:methodusagestart}--\ref{line:methodusage}).
We use that idea to implement labels and roles as object methods. The encoding of local types stipulates that labels are object methods (in case of internal choice) and
as variant tags (in case of external choice). Hence, the \oCODE{label} type (Line~\ref{line:lable} in \cfig{fig:labels}), 
is defined as a pair of a first-class method, i.e using \oCODE{method_}, and
a {\em variant constructor function}. 
While object and variant constructor functions are needed to compose a channel vector in \lstinline!(-->)!,
 object destructor functions are used in \lstinline!merge! in \lstinline!choice_at!,
 to extract bare channels inside an object. Variant destructors are not needed, as they are destructed via pattern-matching and merging is
 done by \lstinline!Event.choose! of Concurrent ML.
 Roles are defined similarly to labels.
See example in Line~\ref{line:roleusage} (the full definition of \lstinline!role! type is available in
\iftoggle{techreport}{\S~\ref{sec:vartup}}{\cite{IRYY2020TechReport}}).
\subsection{Typing Global Combinators via Polymorphic Lenses}
\label{subsec:chanvecimpl}

This section shows one of our main implementation techniques -- the use of \textit{polymorphic lenses} \cite{foster07combinators,pickering17profunctor} for \textit{index-based updates} on tuple types. 
This is essential to the implementation of the typing of \cfig{fig:gctypeupdate} in \OCaml. 
To demonstrate our technique, we sketch the type of the branching combinator, in a simplified form.
The types of all combinators, incorporating first-class methods and variable-length tuples, can be found in
\iftoggle{techreport}{\S~\ref{sec:gc_details}}{\cite{IRYY2020TechReport}}.
The branching combinator demonstrates our key observation that merging of local types 
can be implemented using row polymorphism in \OCaml, which simulates the least upper bound on channel vector types.  

Intuitively, a lens is a functional pointer,
often utilised to access and modify elements of a nested data structure.
In our implementation, lenses provide a way to {\em update} a channel vector in a tuple \lstinline!($t_{r_1}$._*._,,._*._$t_{r_n}$)!.
The type of the lens \lstinline@(--!'g0!--,._--?'t0?--,._--!'g1!--,._--?'t1?--)._idx@
itself points to an element in a specific position in a tuple, by denoting that ``an element \lstinline@--?'t0?--@ is
in a tuple \lstinline@--!'g0!--@'' in a type-parametric way.
Furthermore, this polymorphic lens is capable to express updating the {\em type} of an element,
from \lstinline@--?'t0?--@ in tuple \lstinline@--!'g0!--@ to \lstinline@--?'t1?--@,
which will update \lstinline@--!'g0!--@ itself to \lstinline@--!'g1!--@.
More precisely, the \lstinline!idx! type has two operations:\\
\hspace*{1em}\lstinline@get:._(--!'g0!--,--?'t0?--,--!_!--,--?_?--)._idx._->._--!'g0!--._->._--?'t0?--@
\hspace*{0.5em}and\hspace*{0.5em}
\lstinline@put:._(--!'g0!--,--?_?--,--!'g1!--,--?'t1?--)._idx._->._--!'g0!--._->._--?'t1?--._->._--!'g1!--@.\\
For example, a lens pointing to the first element of a 3-tuple has the type \lstinline@(--!('x*'a*'b)!--,._--?'x?--,._--!('y*'a*'b)!--,._--?'y?--)._idx@.

The branching combinator
\lstinline!choice_at._$r_a$._$\mrg$._($r_a$,$g_1$)._($r_a$,$g_2$)!
is declared in following way: 
{\lstset{numbers=left}
\begin{OCAMLLISTING}
val choice_at : (--!'g0!--, close, --!'g!--, --?'tlr?--) idx -> (* the index of the selecting role *)
    (--?'tlr?--, --?'tl?--, --?'tr?--) disj -> (* the type of disjoint merge function *)
    (--!'gl!--, --?'tl?--, --!'g0!--, close) idx * --!'gl!-- -> (* the type of the first tuple *)
    (--!'gr!--, --?'tr?--, --!'g0!--, close) idx * --!'gr!-- ->  (* the type of the second tuple *)
    --!'g!-- (* the type of the result tuple *)
\end{OCAMLLISTING}}

\noindent
The type variables in the above is resolved {\em a la} logic programs in Prolog,
where several type variables are unified to compose a tuple type of channel vectors.
It requires that both continuation tuples
\oCODE{--!'gl!--} and \oCODE{--!'gr!--} should be of the same type,
{\em except for} the position of active role $r_a$.
The two \lstinline!idx! types paired with continuations
force this unification, by putting \lstinline!close! at $r_a$ in \oCODE{--!'gl!--} and \oCODE{--!'gr!--}.
Thus, the result type \lstinline@--!'g0!--@ is shared among both lenses, so that it contains only types of non-choosing roles and \lstinline!close!.
Each element in \lstinline@--!'g0!--@ is then pairwise merged\footnote{
  We have implemented the type-case analysis for \lstinline!merge! mentioned in \S~\ref{subsec:gcevalimpl}
  via a wrapper called {\em mergeable} around each channel vector,
  which bundles a channel vector and its {\em merging strategy}.
}. The result type of the combinator \oCODE{--!'g!--} 
is obtained by modifying the merged tuple of channel vectors \oCODE{--!'g0!--}
by updating the type of the active role $r_a$ from \lstinline!close!
to \oCODE{--?'tlr?--}, which is the result type of the object concatenation function $\mrg$. 
Function $\mrg$ takes the channel vector types for the role $r_a$ in \oCODE{--!g1!--} and \oCODE{--!g2!--}, 
namely \oCODE{--?'tl?--} and  \oCODE{--?'tr?--},  and 
returns the result type \oCODE{--?'tlr?--}. 
The signature of the combinator also explains the extra occurrence roles paired with each branch. 
Since we need lens $r_a$ within three {\em different instantiations} for different element types
\oCODE{--?'tl?--}, \oCODE{--?'tr?--} and \oCODE{--?'tlr?--} at the position $r_a$,
we need three occurrences of the same lens.

\section{Dynamic and Static Linearity Checks in the Communication API}
\label{sec:linear:channels}

\begin{figure}[t]
\begin{center}
\begin{tabular}{c}
 \noindent
{\small\begin{tabular}{l|l}
Dynamic & Static\\
\hline
 \noindent
\begin{BIGOCAMLLISTING}
<role_$\roleQ$:._<m:._('v,'t)._out>>
\end{BIGOCAMLLISTING}
&
 \noindent
\begin{BIGOCAMLLISTING}
<role_$\roleP$:._<m:._('v._data,'t)._out>>._lin $\mbox{\rm (base value)}$
<role_$\roleP$:._<m:._('s._lin,'t)._out>>._lin $\mbox{\rm (delegation)}$
\end{BIGOCAMLLISTING}
\\[1mm]
\hline
 \noindent
\begin{BIGOCAMLLISTING}
<role_$\roleP$:[`m._of._'v._*._'t]._inp>
\end{BIGOCAMLLISTING}
  & 
 \noindent
\begin{BIGOCAMLLISTING}
<role_$\roleP$:[`m._of._'v._data._*'t._lin]._inp._lin>._lin $\mbox{\rm (base value)}$
<role_$\roleP$:[`m._of._'s._lin._*'t._lin]._inp._lin>._lin $\mbox{\rm (delegation)}$
\end{BIGOCAMLLISTING}
\\
\begin{BIGOCAMLLISTING} 
close
\end{BIGOCAMLLISTING}
  & 
\begin{BIGOCAMLLISTING} 
close lin 
\end{BIGOCAMLLISTING}
\end{tabular}}
\end{tabular}
\end{center}
\vspace{-2em}
\caption{Channel Vector Types with (a)  Dynamic and (b) Static Linearity Checks}
\label{fig:ch_correspondence}
\end{figure}

To ensure that an implementation faithfully implements a well-formed, safe global protocol, MPST theory 
requires that all communication channels are used linearly. 
Similarly, the safety of our library depends on the linear usage of channels.
Our library offers two mechanisms for checking that a channel is used linearly: 
{\em static} and {\em dynamic}. 
Here, we %
briefly explain each of these mechanisms,
by comparing their API usages in \cfig{fig:api} and types in \cfig{fig:ch_correspondence},
where the dynamic version stays on the left while the static one is on the right.

\myparagraph{Dynamic Linearity Checking.}
Dynamic checking, where linearity violations are detected at runtime, is proposed by \cite{DBLP:conf/esop/TovP10} and \cite{HY2016},
and later adopted by \cite{padovani17context, scalas17linear}.
In \must, dynamic linearity checking is implemented by wrapping the input and output channels,
with a boolean flag
that is set to true once the channel has been used.
If linearity is violated, i.e a channel is accessed after the linearity flag has been set to true, then  an exception \oCODE{InvalidEndpoint} will be raised.
Note that our library correctly handles output channels between several alternatives
being used {\em only once}; for example, 
from a channel vector $c$ of type
\oCODE!<$\roleR$:._<$\labOk$:._(string,close)._out;._$\labCancel$:._(string,close)._out>>!,
the user can extract two channels \oCODE{$c$#$\roleR$#$\labOk$} and \oCODE{$c$#$\roleR$#$\labCancel$}
where an output must take place on either of the two bare channels, but not both.
In addition, our library wraps each bare channel
with a {\em fresh} linearity flag on each method invocation,
since in recursive protocols, a bare channel is often {\em reused},
as the formalism (\S~\ref{sec:formalism}) implies.
\myparagraph{Static Linearity Checking with Monads and Lenses.}
The static checking is built on top of {\tt linocaml} \cite{linocaml}:
a library implementation of linear types in \OCaml
which combines the usage of {\em parameterised monads} \cite{atkey09parameterized}
and polymorphic lenses (see \S~\ref{subsec:chanvecimpl}),
to enable static type-checking on the linear usage of channels.
In particular, we reuse several techniques from \cite{linocaml, imai18sessionscp}.
A parameterised monad, which we model by the type \oCODE{(($\mathit{pre}$,$\mathit{post}$,$v$)$\,$monad)},
denotes a computation of type $v$ with a {\em pre}- and a {\em post}-condition, and they are utilised to
track the creation and consumption of resources at the type level. A well-known restriction of parameterised monads in the context of session types, is that they support communication on a single channel only, and hence are incapable of expressing session delegation and/or interleaving of multiple session channels. 
To overcome this limitation, the {\em slot monad} proposed in \cite{linocaml,imai18sessionscp}
extends the parameterised monad to denote
{\em multiple} linear resources in the {\em pre}- and {\em post}-conditions. The 
resources are represented as a sequence, and each element is modified 
using polymorphic lenses \cite{pickering17profunctor}. %

We incorporate the above-mentioned techniques of {\tt linocaml} so that,
instead of having a single channel vector in the {\em pre} and {\em post} conditions,
we can have a sequence of channel vectors, and we use lenses to {\em focus} on a channel vector at a particular \textit{slot}.  
If we do not require delegation or interleaving, then the length of the sequence is one and 
the monadic operations always update the first element of the sequence. 
In particular, as in \cite{imai18sessionscp}, if a channel is delegated
i.e sent through another channel,
that slot (index) of the sequence is updated to \oCODE{unit}, marking it as consumed.

The \must API, for static linearity checking, is given in \cfig{fig:api}(b), where
\oCODE{s$_i$}, and \oCODE{s$_j$} in delegation, denote {\em lenses} pointing at $i$-th and $j$-th
slot in the monad.
The binary channels in the channel vector, used within the monadic primitives \oCODE{send} and \oCODE{receive},
are of the types given in  \cfig{fig:ch_correspondence}(b).
Functions \oCODE{send} and \oCODE{receive} both take (1) a lens $s_i$ pointing to a channel vector;
and (2) a selector function which extracts, from the channel vector at index $s_i$,
a channel (\oCODE{('v._data,._'t$_1$)._out} for output and
\oCODE{'a._inp} for input.
Type \oCODE{data} denotes unrestricted (non-linear) payload types, whose values are matched against ordinary variables. The result of the monadic primitives is returned as a value of either type \oCODE{'t._lin} for output
or \oCODE{'a._lin} for input, which is matched by \oCODEEsc{match\%lin} or \oCODEEsc{let\%lin},
ensuring the channels (and payloads, in case of delegation) are used linearly. A \oCODE{lin} type must be matched against
{\em lens-pattern} prefixed by \oCODE{&#&}. Note that, {\tt linocaml} overrides
the \oCODE{let} syntax and \lstinline!&#&! pattern, in the way that \oCODEEsc{let\%lin #s$_i$ = $\mathit{exp}$}
updates the index \oCODE{s$_i$}, in the sequence of channel 
vectors, with the value returned from $\mathit{exp}$.  

\begin{figure}[t]
\begin{tabular}{c}
{\small\begin{tabular}{l|l}
Dynamic & Static (monadic) \\
\hline
\noindent\hspace{-0.5em}
\begin{LISTING}
let s = send s#role_q#m v in $e$
let s = send s#role_q#m s' in $e$
\end{LISTING}
&
\noindent\hspace{-0.5em}
\begin{LISTING}
let???lin #s$_i$ = send s$_i$ (fun x -> x#role_q#m) v in $e$
let???lin #s$_i$ = deleg_send s$_i$ (fun x -> x#role_q#m) s$_j$ in $e$
\end{LISTING}
\\[1mm]
\noindent\hspace{-0.5em}
\begin{LISTING}
match receive s#role_p with
|`m$_1$(x,s) -> $e_1$
|`m$_2$(s',s) -> $e_2$
\end{LISTING}
  & \noindent\hspace{-0.5em}
 \begin{LISTING}
match???lin receive s$_i$ (fun x->#role_$\roleP$) with
|`m$_1$(x,#s$_i$) -> $e_1$
|`m$_2$(#s$_j$,#s$_i$) -> $e_2$ $\mbox{\rm (delegation)}$
\end{LISTING}
\\[1mm]
\noindent\hspace{-0.5em}
\begin{LISTING} 
close s
  \end{LISTING}
  &
 \noindent\hspace{-0.5em}
\begin{LISTING} 
close $s_i$
  \end{LISTING}
\end{tabular}}
\end{tabular}
\caption{\OCaml{} API for MPST with Dynamic (a) and Static (b) linearity checks}
\label{fig:api}
\end{figure}

To realise session delegation, we have implemented a separate monadic primitive, \oCODE{deleg_send._$s_i$._(fun._x->x#$\roleP$#$\ocLab$)._$s_j$}, presented in \cfig{fig:api}(b).  %
The primitive extracts the channel vector at position $s_i$
and then updates the channel vector at position $s_j$.
As a result, the slot for $s_j$ is
returned and used in further communication, the slot  $s_i$ is updated to \oCODE{unit}. 
An example program that uses \must static API is given in \cfig{fig:client:oAuth2}(b). 

\section{Evaluation}
\label{sec:evaluation}
We evaluate our framework in terms of run-time performance
(\S~\ref{sec:performance}) and applications
(\S~\ref{sec:usecases}, \S~\ref{sec:oauth}).
We compare the performance of \ourlibrary with programs written in a continuation-passing-style (following the encoding presented in \cite{Scala}) and untyped implementations (Bare-OCaml) that utilise popular communication libraries.
In summary,  \ourlibrary
has negligible overhead in comparison with \textit{unsafe} implementations (Bare-OCaml),  and CPS-style implementations.
We demonstrate the applicability of \ourlibrary by implementing %
a lot of use cases. In \S~\ref{sec:oauth}, 
we show the implementation of the OAuth protocol, 
which is the first application of session types over \CODE{http}.

\subsection{Performance}
\label{sec:performance}
The runtime overhead of \ourlibrary stems from the implementation of channel vectors, more specifically:
(1) extracting a channel from an \OCaml object when performing a
communication action, and (2) either (2.1) dynamic linearity checks or (2.2) more closures introduced by the usage of a  slot monad for static checking.
Our library is parameterised on the underlying communication transport. We evaluate its performance in case of synchronous, asynchronous and distributed transports. 
Specifically, we use the following communication libraries:
\begin{enumerate}
\item[(1)] {\bfseries \ttfamily ev}:
\OCaml's standard {\tt Event} channels which implements channels
shared among POSIX-threads; 
\item[(2)] 
{\bfseries \ttfamily lwt}:
Streams between {\em lightweight-threads} \cite{lwt}, which are more
efficient for I/O-intensive application in general, and
broadly-accepted by the OCaml communities, and
\item[(3)] 
{\bfseries \ttfamily ipc}: UNIX pipes distributed over UNIX processes.
\end{enumerate}
Note that {\bfseries \ttfamily ev} is synchronous, while the other two are asynchronous.
Also, due to current \OCaml limitation, POSIX-threads in a process cannot run simultaneously in parallel,
which  particularly affects the overall performance of (1).
As OCaml garbage collector is not a concurrent GC, only a single OCaml thread is allowed to manipulate the heap, which in general limits the overall performance of multi-threaded programs written in OCaml.
For (3), we generate a single pipe for each pair of processes, and maintain a mapping between a local channel
and its respective dedicated UNIX pipe. In addition, we also 
implement an optimised variant of \must in the case of { \ttfamily lwt},  denoted as {\bfseries \ttfamily lwt-single} in \cfig{fig:runtime-performance};
it reuses a single stream among different payload types,
instead of using different channels for types.
In particular, we cast a payload to its required payload type utilising \oCODE{Obj.magic},
as proposed and examined by \cite{padovani16simple, DBLP:conf/coordination/ImaiYY17}.
Our benchmarks are generalisable because each microbenchmark exhibits the worst-case scenario for its potential source of overhead.

We compare implementations, written using (1) \must static API, 
 (2) \must dynamic API, (3) a Bare-\OCaml implementation using untyped channels as provided by the corresponding transport library, and (4) a CPS implementation, following 
the encoding in \cite{scalas17linear}. We have implemented the encoding manually such that a channel is created at each communication step, and passed as a continuation. 
\cfig{fig:runtime-performance} reports the results on three microbenchmarks.

\vspace{2pt}
\noindent\emph{\textbf{Setup.}}
We use the native \textit{ocamlopt} compiler of \OCaml 4.08.0 with Flambda optimiser\footnote{
\url{https://caml.inria.fr/pub/docs/manual-ocaml/flambda.html}}. %
Our machine configurations are
Intel Core i7-7700K CPU (4.20GHz, 4 cores), Ubuntu 17.10, Linux
4.13.0-46-generic, 16GB. 
We use \CODE{Core\_bench}\footnote{\url{https://blog.janestreet.com/core_bench-micro-benchmarking-for-ocaml/}},
a popular benchmark framework in \OCaml, which uses its built-in linear regression for estimating the reported costs.
We repeat each microbenchmark for 10 seconds of quota where \CODE{Core\_bench} takes
hundreds of samples, each consists of
up to 246705 runs of the targeted OCaml function,
we obtain the average of execution time with fairly narrow 95\% confidence interval.

\begin{figure}
  \begin{center}
    \includegraphics[scale=0.35]{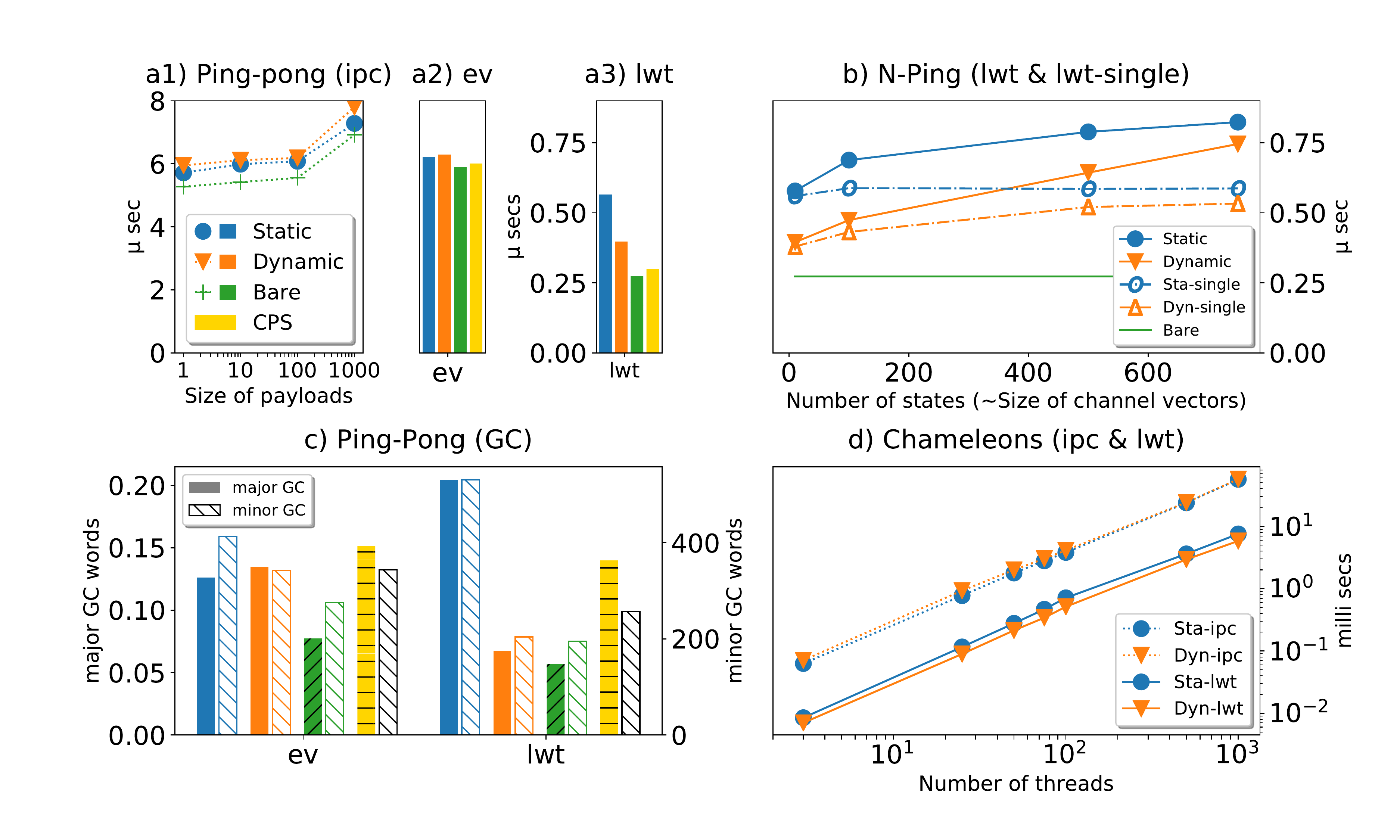}
  \vspace*{-1em}
\caption{Runtime performance vs GC time performance}
\label{fig:runtime-performance}
  \end{center}
\end{figure}

\myparagraph{Ping-pong} benchmark measures the execution time for completing a recursive protocol between two roles, which are repeatedly exchanging request-response messages
of increasing size (measured in 16 bit integers). %
The example is communication intensive and exhibits no other cost apart from the (de)serialisation of values that happens in the {\ttfamily ipc} case,
hence it demonstrates the pure overhead of channel extraction, dynamic checks and parameterised monads.
In the case of a shared memory transports ({\color{modify}ev and lwt}), we report the results of 
a payload of one integer since the size of the message does not affect the running time. 

The slowdown of \must is negligible (approx. 5\% for Dynamic vs Bare-\OCaml, and 13\% for Static vs Bare-\OCaml) when using either { \ttfamily ev}, \cfig{fig:runtime-performance} (a1), or { \ttfamily ipc}, \cfig{fig:runtime-performance}(a2), as a transport, since the overhead cost is overshadowed by latency. The shared memory case using  { \ttfamily lwt}, \cfig{fig:runtime-performance}(a3), represents the worse case scenario for \ourlibrary
since it measures the pure overhead of the implementation of many
interactions purely done on memory with minimal latency. The slowdown
in the static version is expected \cite{imai18sessionscp} and reflects
the cost of monadic closures, as the current implementation does not optimise them away.
The linearity monad is implemented via a state monad \cite{linocaml}, which incurs considerable overhead.
  The OCaml Flambda optimiser could remove more closures if we annotate the program with inlining specifications.
The slowdown (although negligible) in
comparison with CPS is surprising since we pre-generate all channels
up-front, while the CPS-style implementation creates a channel at each
interaction step. Our observation is that 
the compiler is optimised for handling large amounts of immutable values, while  
\OCaml objects %
(utilised by the channel vector abstraction) are less efficient than normal records and variants.

\cfig{fig:runtime-performance} (c) reports on the memory consumption 
(in terms of words in the major and minor heap) for executing the protocol. 
Channel vectors with dynamic checking have approximately the same memory footprint
as Bare-\OCaml, and significantly less footprint when compared with a CPS implementation.

\myparagraph{n-Ping} is a protocol of increasing size, 
\oCODE{nping} global combinator forming repeated composition
of the communication combinators defined by
\oCODE{g$_i$ = (a-->b) ping @@ (b-->a) pong @@ g$_{i-1}$}, \oCODE{g$_0$ = t}
and \oCODE{nping = fix (fun t -> g$_n$)}, where $n$
corresponds to the number of \oCODE{ping} and \oCODE{pong} states.
In contrast to Ping-Pong, this example generates a large number of channels and large channel vector objects, evaluating how  well \must scales w.r.t the size of the channel vector structure. We show the results for 
transports  { \ttfamily lwt} and { \ttfamily lwt-single} in  \cfig{fig:runtime-performance} (b). The static version  of {\ttfamily lwt-single} has a constant overhead from Bare-\OCaml. 
Although the static checking implementation is in general slower, the relative overhead, in comparison with dynamic checking, decreases as the protocol length increases.

\myparagraph{Chameleons} protocol specifies that \CODE{n} roles ("chameleons") connect to a central
broker, who picks pairs and sends them their respective reference, so they can interact peer-to-peer. The example tests delegation (central broker sends a reference) and creation of many concurrent sessions (peer-to-peer interaction of chameleons). The results reported in \cfig{fig:runtime-performance} (d) show that 
the implementation of delegation with static linearity checking scales as well as its dynamic counterpart. The cost of linearity (monadic closures) is less than the cost of dynamic checks for many concurrent sessions over {\ttfamily lwt} transport. 

%
%
%
%
%
%

\subsection{Use Cases}
\label{sec:usecases}

We demonstrate the expressiveness and applicability of \must by 
specifying and implementing protocols for a range of applications, listed in \cfig{fig:usecase}.
We draw the examples from three categories of benchmarks: (1)
\emph{session benchmarks} (examples 1-9), which are gathered from the session types literature; (2) \emph{concurrent algorithms} from the Savina benchmark suit \cite{Savina} (examples 10-13); and  (3) \emph{application protocols} (examples 14-16), which focus on well-established protocols that demonstrate interoperability between
\ourlibrary implemented programs and existing client/servers. For each use case we report on Lines of Code (LoC) of global combinators and the compilation time (\CODE{CT} reported in milliseconds).
We also report if the example requires 
full-merge \cite{DY12} (FM) -- a well-formedness condition on global
protocols that is not supported in \cite{scalas17linear}, but
supported in \must. 
\newcommand{\xmark}{\text{\sffamily x}}
\begin{figure}[t]
\begin{center}
\hspace{3mm}
\begin{minipage}{0.40\textwidth}
\footnotesize
\begin{tabular}{|l@{\hspace{0.5mm}}|l@{\hspace{0.5mm}}|l@{\hspace{0.5mm}}|l@{\hspace{0.5mm}}|l|l|}
 \hline
 Example (role)  & LoC &\CODE{CT}\tiny{(ms)} & FM \\ %
 \hline \hline
 1. 2-Buyer \cite{HY2016}   &  15 & 45 & $\checkmark$ \\
 2. 3-Buyer \cite{HY2016}    &  21 & 47  & $\checkmark$ \\
 3. Fibonacci 	\cite{HY2016}        &  8 & 38  & \xmark\\
 4. SAP-Negotiation  \cite{HY2016}    & 17  & 46  & \xmark \\
 5. Supplier Info  \cite{HY2016}  &  50 & 85  & $\checkmark$ \\
 6. SH \cite{pucella08session,HY2016}
                            &  27 & 58 & $\checkmark$ \\
 7. Distributed Calc    \cite{HY2016}             &  12 & 41 & \xmark  \\
 8. Travel Agency  \cite{HY2016}     &  16 & 66 & $\checkmark$ \\
 
  \hline
\end{tabular}
\end{minipage}
\hspace{6mm}
\begin{minipage}{0.5\textwidth}
\footnotesize
\begin{tabular}{|l@{\hspace{0.5mm}}|l@{\hspace{0.5mm}}|l@{\hspace{0.5mm}}|l@{\hspace{0.5mm}}|l| l|l|}
 \hline
  Example (role)  &  LoC & \CODE{CT}\tiny{(ms)} & FM \\%

 \hline \hline
9. Game  \cite{scalas17linear}    &  17 & 49 & \xmark \\
10. MapReduce \cite{Savina}                        & 5 & 33& \xmark \\
11. Nqueen \cite{Savina}                           &  12 & 55 & \xmark \\
12. Santa \cite{benton03santa,linocaml}                             &  14 & 42 & 
\xmark \\
13. Sleeping Barber     \cite{HY2016}                            &  15 & 43  & $\checkmark$ \\
 \hline
 14. SMTP \cite{HY2016}    &  54 & 124  & \xmark \\
 15. OAuth        &  26 & 60   & $\checkmark$\\
 16. DNS            &  11 & 57 &  \xmark  \\
\hline
\end{tabular}
\end{minipage}
\end{center}
\hspace*{-2em}
\caption{Implemented Use cases (LoC: Lines of code, CT: Compiling Time, FM: Full merge.)}
\label{fig:usecase}
\end{figure}

Examples 1-9 are gathered from the official Scribble test
suite\footnote{https://github.com/scribble/scribble-java}
\cite{BETTYTOOLBOOK}, and we have  
converted Scribble protocols to global protocol combinators. 
Examples 10-13 are concurrent algorithms and are parametric on the number of roles ($n$).
To realise the scatter-gather pattern required in the examples, we
have added
{\color{modify}
  two new constructs, \oCODE{scatter} and \oCODE{gather},
which correspond to a subset of the parameterised role extension
for MPST protocols \cite{CHJNY2019}.}

To test the applicability of \must to real-world protocols 
we have specified, using global combinators,  a core subset of three Internet protocols (examples 14-16), namely the Simple Mail Transfer Protocol (SMTP), 
the Domain Network System (DNS) protocol and the OAuth  protocol. Using the \must APIs, it was straightforward to implement 
compliant clients in \OCaml that interoperate with popular servers.  
In particular,  we have implemented an SMTP client that interoperates
with the Microsoft exchange server and sends an e-mail,
an OAuth authorisation service that connects to a Facebook server and authenticates
a client, and a DNS client and a server, which are implemented on top of a popular DNS library in \OCaml (\CODE{ocaml-dns}). 
Note that DNS has sessions, as the DNS protocol has an ID field to discriminate sessions; and a request forwarding in the DNS protocol 
involves more than two participants (i.e. servers).

%
%
%
%

%
%
%
%

%
%
%

%
%
%
%

%
%
%
%
%
%
%
%
%
%
%
%
%

\subsection{Session Types over HTTP: Implementing OAuth}\label{sec:oauth}
In this section, we discuss more details about \must implementation of
OAuth\footnote{https://oauth.net/2/}, which is an Internet standard for authentication. 
OAuth is commonly used as a way for Internet users to grant websites or applications access to their information on other websites but without giving them the passwords 
by providing a specific authorisation flow. 
\cfig{fig:oAuth} shows the
specification of the global combinator, along with an implementation for the 
authorisation server.  
We have specified a subset
of the protocol, which includes establishing a secure connection
and conducting the main authentication transaction. 
Using \CODE{OAuth} as an example, we also discuss practically motivated extensions, 
{\em explicit connection handling} akin to the one in \cite{HY2017},
to the core global combinators.
We present that a common pattern when HTTP is used as an underlying transport.

\begin{figure}[t!]
\begin{minipage}{0.45\textwidth}
{\lstset{numbers=left,firstnumber=1}
\begin{OCAMLLISTINGTINY}
let fb_oauth =^\label{line:fb:gstart}^
 (c -!-> s) (get "/start_oauth") @@^\label{line:fb:gentry}^
 (s -?-> c) _302  @@^\ ^(* 302: HTTP redirect *)^\label{line:fb:gfirstreply}^
 (c -!-> a) (get "/login_form") @@^\label{line:fb:gloginformreq}^
 (a -?-> c) _200 @@^\label{line:fb:gloginformresp}^
 (c -!-> a) (post "/auth") @@^\label{line:fb:glogin}^
 choice_at a (to_c success_or_fail)^\label{line:fb:gchoice}^
 (a,(a -?-> c) (_200_success ...) @@^\label{line:fb:g200success}^
    (c -!-> s) (success is_ok "/callback") @@^\label{line:fb:gcallbacksuccess}^
    (s -!-> a) (get "/access_token") @@^\label{line:fb:gaccesstoken}^
    (a -?-> s) _200 @@^\label{line:fb:gaccesstoken200}^
    (s -?-> c) _200 @@^\label{line:fb:gcallback200}^
    finish)^\label{line:fb:gsuccessend}^
 (a,(a -?-> c) (_200_fail ...) @@^\label{line:fb:g200fail}^
    (c -!-> s) (fail is_fail "/callback") @@^\label{line:fb:gfail}^
    (s -?-> c) _200 @@^\label{line:fb:gfail200}^
    finish)^\label{line:fb:gfailend}^
\end{OCAMLLISTINGTINY}
}
\end{minipage}
\begin{minipage}{0.45\textwidth}
{\lstset{numbers=left,firstnumber=18}
\begin{OCAMLLISTINGTINY}
let fb_acceptor = H.start_server 8080 "/mpst-oauth"^\label{line:fb:sstart}^
let rec facebook_oauth_consumer () =^\label{line:fb:sfunstart}^
  let ch = get_ch s &fb_oauth& in^\label{line:fb:sepp}^
  let sid = string_of_int (Random.int ()) in^\label{line:fb:sgenid}^
  let &conn& = fb_acceptor sid in^\label{line:fb:saccept1}^
  let `get(_, ch) = receive (ch &conn&)#role_C in^\label{line:fb:s}^
  let redir_url = fb_redirect_url sid "/callback" in^\label{line:fb:sredirecturl}^
  let ch = send ch#role_C#_302 redir_url in^\label{line:fb:sredirect}^
  let &conn& = fb_acceptor sid in^\label{line:fb:saccept2}^
  let ch = match receive (ch &conn&)#role_C with^\label{line:fb:schoice}^
    | `success(_,ch) ->^\label{line:fb:ssuccess}^
       let &conn_p& = H.http_connector^\label{line:fb:sconn1}^
            "https://graph.facebook.com/v2.11/oauth" in^\label{line:fb:sconn2}^
       let ch = send (ch &conn_p&)#role_A#get [] in^\label{line:fb:sgetprop}^
       let `_200(auinfo,ch) = receive ch#role_A in^\label{line:fb:saccesstoken}^
       send ch#role_C#_200 "auth succeeded"^\label{line:fb:successback}^
    | `fail(_,ch) -> send ch#role_C#_200 "auth failed"^\label{line:fb:sfail}^
  in close ch; facebook_oauth_consumer ()^\label{line:fb:send}^
\end{OCAMLLISTINGTINY}
}\end{minipage}
\hspace*{-2em}
\caption{Global Combinators and Local Implementations for OAuth (excerpt)}
\label{fig:oAuth}
\end{figure}

\myparagraph{Extension for handling stateless protocols.} The protocol has a very similar structure to the \oCODE{oAuth} protocol, presented in 
\S~\ref{sec:overview}. However, 
the original OAuth protocol is realised over a RESTful API, which means that 
every session interaction is either an HTTP request or an HTTP response. 
To handle HTTP connections, we have implemented a thin wrapper around an HTTP library, {\tt Cohttp}\footnote{\url{https://github.com/mirage/ocaml-cohttp}}, 
and we make HTTP actions explicit in the protocol by 
 proposing two new global combinators, 
{\em connection establishing} combinator \oCODE{(-!->)} and
{\em disconnection} combinator \oCODE{(-?->)}.
Session types represent the types of the communication channel after a session (a TCP connection in the general
case) has been established. Since RESTful protocols, realised over HTTP transport, are stateless, a connection is
``established'' at every HTTP Request. We explicitly encode this behaviour by replacing the --> combinator that 
denotes that one role is sending to another, with two new combinators. The combinator -!-> means establishing a connection and
piggybacking a message, while  -?-> denotes piggybacking a message and disconnect. This simple extension allows us to faithfully 
encode HTTP Request and HTTP Response. For example, \lstinline@a-!->b@ requires that role \lstinline!a! connects on an HTTP port to \lstinline!b! and then \lstinline!a! sends a
message to \lstinline!b!, hence implementing HTTP Response; on the other hand \lstinline@a-?->b@ specifies an HTTP Response.

\myparagraph{Implementation.} The global combinator \oCODE{fb_oauth} is given in \cfig{fig:oAuth} (a). 
As before, the protocol consists of three parties, a service
\Rptp{s}, a client \Rptp{c}, and an authorisation server \Rptp{a}.
First, \Rptp{c}  connects to \Rptp{s}
via a relative path \oCODE{"/start_oauth"} (Line~\ref{line:fb:gentry}).
Then \Rptp{s}  redirects \Rptp{c} to \Rptp{a}
using HTTP redirect code  \oCODE{_302} (Line~\ref{line:fb:gfirstreply}).
As a result the client sees a login form at \oCODE{"/login_form"}
(Lines~\ref{line:fb:gloginformreq}-\ref{line:fb:gloginformresp}),
where they enter their credentials (Line~\ref{line:fb:glogin}).
Based on the validity of the credentials received by \Rptp{c},
\Rptp{a} sends \oCODE{_200_success} (Line~\ref{line:fb:g200success}) or \oCODE{_200_fail}.
If the credentials are valid,  \Rptp{c} proceeds and connects to \Rptp{s} on path \oCODE{"/callback"} (Line~\ref{line:fb:gcallbacksuccess}),
requesting to get access to a secure page.
The service \Rptp{s} then retrieves an {\em access token} from  \Rptp{a} on 
URL \oCODE{"/access_token"} (Lines~\ref{line:fb:gaccesstoken}-\ref{line:fb:gaccesstoken200}),
and navigates the client to an authorised page,
finishing the session (Lines~\ref{line:fb:gcallback200}-\ref{line:fb:gsuccessend}). 
If the credentials are not valid, the client reports the failure to \Rptp{s} (Lines~\ref{line:fb:gfail}-\ref{line:fb:gfail200}),
and the session ends (Line~\ref{line:fb:gfailend}).

The \oCODE{s}erver role of \oCODE{fb_oauth} is faithfully implemented in
Lines~\ref{line:fb:sstart}-\ref{line:fb:send}
which provides an OAuth application utilising Facebook's authentication service.
Line~\ref{line:fb:sstart}
starts a thread which listens on a port \oCODE{8080} for connections.
Essentially it
starts a web service at an absolute URL \oCODE{"/mpst-oauth"}
(i.e. relative URLs like \oCODE{"/callback"} are mapped to \oCODE{"https://.../mpst-oauth/callback"}). %
The recursive function \oCODE{facebook_oauth_consumer} starting from Line~\ref{line:fb:sfunstart} is the main event loop for \Rptp{s}.
Line~\ref{line:fb:sepp} extracts a channel vector
from the global combinator \oCODE{fb_oauth}, of which type is propagated to the rest of the code.
Then it generates a session id via a random number generator \oCODE{(Random.int ())} (Line~\ref{line:fb:sgenid}),
and waits for an HTTP request from a client on  \oCODE{fb_acceptor} (Line~\ref{line:fb:saccept1}).
When a client connects, the connection is bound to the variable \oCODE{&conn&} associated with the pre-generated session id. 
Note that the channel vector expects a connection
since no connection has been set for the client yet.
Here, the connection is supplied to the channel vector via function application \oCODE{(ch &conn&)}.
On Line~\ref{line:fb:sredirecturl}, expression \oCODE{(fb_redirect_url sid "/callback")} prepares a
redirect URL to an authentication page of a Facebook Provider (\verb!https://www.facebook.com/dialog/oauth!)
After sending back (HTTP Response) the redirect url to the client with \oCODE{_302} label (Line~\ref{line:fb:sredirect}),
the connection is implicitly closed by the library. Note that we do not need to supply 
a connection to the channel vector on Line~\ref{line:fb:sredirect};
because 
a connection already exists, we have already received an HTTP request from the
user and Line~\ref{line:fb:sredirect} simply performs HTTP response. 
The next lines proceed as expected following the protocol, with 
the only subtlety  that we thread the connection object in subsequent send/receive calls.

The full source code of the benchmark
protocols and applications and the raw data are available from the project repository.
\section{Related Work}
\label{sec:related}

\newcommand{\fuse}{FuSe }

We summarise the most closely related works on session-based languages 
or multiparty protocol implementations. See \cite{BETTYTOOLBOOK} for recent
surveys on theory and implementations.

The work most closely related to ours is \cite{scalas17linear}, 
which implements multiparty session interactions 
over binary channels in Scala built on 
an encoding of a multiparty session calculus 
to the $\pi$-calculus. 
The encoding relies
on {\em linear decomposition} of channels, which is defined in terms
of {\em partial projection}. 
Partial projection is restrictive, and rules out many
protocols presented in this paper. For example, 
it gives an undefined behaviour for role \Rptp{c} and \Rptp{s} for protocols \oCODE{oAuth2} and \oCODE{oAuth3} in Fig.~\ref{fig:auth2}.
Programs in \cite{scalas17linear}  
have to be written in a continuation passing style 
where a fresh channel is created at each 
communication step. %
In addition, the ordering of communications across separate channels
is not preserved in the implementation, e.g. sending a \oCODE{login}
and receiving a \oCODE{password} in the protocol \oCODE{oAuth} is
decomposed to two separate elements which are not causally related.
This problem is mitigated by providing an external protocol 
description language, Scribble \cite{scribble}, and its API generation tool,
that links each protocol state using a call-chaining API
~\cite{HY2016}. 
The linear
usage of channels is checked at runtime.

An alternative way to realise multiparty session communications over
binary channels is using an orchestrator -- an intermediary process
that forwards the communication between interacting parties.  The work
\cite{DBLP:conf/forte/CairesP16} suggests addition of a medium process
to relay the communication and recover the ordering of communication
actions, while the work \cite{DBLP:conf/concur/CarboneLMSW16} adds
annotations that permit processes to communicate directly without
centralised control, resembling a proxy process on each
side.  Both of the above works are purely theoretical.

Among multiparty session types implementations, 
several works exploit the equivalence
between local session types and communicating automata to generate
session types APIs for mainstream programming languages (e.g., Java
\cite{HY2016, DBLP:conf/ppdp/KouzapasDPG16}, Go \cite{CHJNY2019}, F\#
\cite{scalas17linear}). Each state from  state automata is implemented as
a class, or in the case of \cite{DBLP:conf/ppdp/KouzapasDPG16}, as a
type state.  To ensure safety, state automata have to be derived from
the same global specification.
All of the works in this category 
 use the Scribble toolchain to generate the state classes  from a
 global specification. 
Unlike our framework, a local type is not 
inferred automatically and the subtyping relation is limited 
since typing is nominal and is 
constrained by the fixed subclassing relation 
between the classes that represent the states.
All of these implementations also 
detect linearity violations at runtime, and 
offer no static alternative. 

In the setting of binary session types, \cite{imai18sessionscp}
propose an \OCaml library, which uses a slot monad to manipulate
binary session channels. Our encoding of global combinators to
simply-typed binary channels enable the reuse of the techniques
presented in \cite{imai18sessionscp}, e.g.~for delegations and
enforcement of linearity of channels.  

\fuse  \cite{padovani17context} is another library for session
programming in \OCaml.  It supports a runtime mechanism for linearity
violations, as well as a monadic API for a single session without
delegation. The implementation of \fuse is based on the encoding of
binary session-typed process into the linear $\pi$-calculus, proposed
by \cite{dardha12session}. The work \cite{scalas16lightweight} also 
implements this encoding in Scala, 
and the work \cite{scalas17linear} extends the encoding and
implementations to the multiparty session types (as discussed 
in the first paragraph).  

Several Haskell-based works \cite{pucella08session,
  orchard16effects, lindley16embedding} exploit its richer typing
system to statically enforce linearity with various
expressiveness/usability trade-offs based on their session types
embedding strategy. These works depend on type-level features in
Haskell, and are not directly applicable to \OCaml. A detailed
overview of the different trade-off between these implementations in
functional languages is
given in Orchard and Yoshida's chapter in \cite{BETTYTOOLBOOK}.
Based on logically-inspired representation of session types, 
embedding higher-order binary session processes using contextual monads
is studied in 
\cite{DBLP:conf/esop/ToninhoCP13}.  
This work is purely
theoretical.

Outside the area of session-based programming languages, 
various works study protocol-aware
verification. Brady et al. \cite{IdrisProtocolDevelopment} describe a
discipline of protocol-aware programming in Idris, in which adherence
of an implementation to a protocol is ensured by the host language
dependent type system.  Similarly,
\cite{DBLP:journals/pacmpl/SergeyWT18} proposes a programming logic,
implemented in the theorem prover Coq, for reasoning on protocol
states.  A more  lightweight verification approach is developed in
\cite{DBLP:conf/padl/AndersenS19} for a set of
protocol combinators, capturing patterns for distributed
communication. However, the verification is done only at runtime.
The work \cite{DBLP:conf/popl/CarboneM13} presents a global language for
describing choreographies and a global execution model where the
program is written in a global language, and then automatically
projected using code generation to executable processes (in the style
of BPMN).  All of the above works either develop a new language or are
built upon powerful dependently-typed host languages (Coq, Idris).
Our aim is to utilise the MPST framework for specification and verification of
distributed protocols, proposing 
a type-level treatment of protocols 
which relies solely on existing
language features.

\section{Conclusion and Future Work}
\label{sec:conclusion}
In this work, we present a library for programming multiparty protocols
in \OCaml, which ensures \textit{safe} multiparty communication over
binary I/O channels.  The key ingredient of our work is the notion of
global combinators -- a term-level representation of global types, that
automatically derive channel vectors -- a data structure of nested
binary channels. We present two APIs for programming with channel
vectors, a monadic API that enables static verification of linearity
of channel usage,
and one that checks channel usage at runtime. 
OCaml is intensively used for system programming among several groups and companies in both
industry and academia ~\cite{minsky11ocaml,barham03xen,
madhavapeddy08xen,
madhavapeddy14unikernels,merkel14docker,mldonkey,
chanezon16docker,radanne16eliom}.
We plan to apply \ourlibrary to such real-world applications.

We formalise a type-checking algorithm for global protocols, 
and a sound derivation of channel vectors, which, we believe, are applicable beyond \OCaml. 
In particular, TypeScript is a
promising candidate as it  
is equipped with a structural type system akin to the one presented in our paper. 

To our best knowledge, this is the first work  to  
enable MPST protocols to be written, verified, and implemented in a 
single (general-purpose) programming language and the first
implementation framework of statically verified 
MPST programs. By combining protocol-based specifications, static linearity checks and structural typing, we allow one to implement communication programs that are extensible and type safe by design. 

%
%
%

\bibliography{session}

\iftoggle{techreport}{%
  \newpage
  \appendix
  \let\cleardoublepage\clearpage
  \appendixpage
\section{Auxiliary Definitions}\label{sec:omitted}

\subsection{Merging of Channel Vectors}\label{subsec:merge}
On merging $\ocBinMerge[\ocChi]$,
an extra bookkeeping $\ocChi$
is introduced to ensure termination.
Note that, according to our typing rule for $\ogtChoiceKwd$,
both hand sides must have the same type.
Merging for output
$\ocIntSumSmall{\roleP}{i \in I}{\ocIntChoiceSmall{\mpLab[i]}{\ocS[i]}{\ocC[i]}}     $
requires both branches to have an intersection.
It generates a record only with the overlapping fields
(which means we have the same set of output labels at any branches),
and for each field
it puts the name from left hand side
(left/right does not matter since both names are identical if the global combinator is well-typed)
and the continuation is obtained by merging the ones from both hand sides.
Merging for input
$\ocExtSumSmall{\roleP}{i \in I}{\ocExtChoiceSmall{\mpLab[i]}{\ocS[i]}{\ocC[i]}}$
is more permissive,
as it keeps variant tags which do not exist in the other hand side as-is.
For the overlapping tags,
names from the left hand side is taken as well and
the continuations are merged.
For recursions, it generates
a fresh recursion variable and bind it on top of the channel vector being generated,
then it continues merging by expanding the recursion binder on each side.
It adds a mapping between the recursion variable and the pair of given
channel vectors to $\ocChi$,
so that the corresponding recursive variable
is returned if the merging encounters the pair again,
forming an appropriate loop and ensuring termination.

}

\begin{figure}[t]
  \centerline{
    \scalebox{0.91}{\(
    \begin{array}{rclcl}
     \ocIntSumSmall{\roleP}{i \in I}{\ocIntChoiceSmall{\mpLab[i]}{\ocS[i]}{\ \ocC[1i]}}
     & \ocBinMerge[\ocChi] &
     \ocIntSumSmall{\roleP}{j \in J}{\ocIntChoiceSmall{\mpLab[j]}{\ocS[j]}{\ \ocC[2j]}}
     & = &
     \ocIntSumSmall{\roleP}{k \in I \cap J}{\ocIntChoiceSmall{\mpLab[k]}{\ocS[k]}{\ \ocC[1k]{\ocBinMerge[\ocChi]}\ocC[2k]}}\\[2mm]
    &  &  &
    \multicolumn{2}{l}{
      \text{where }\ocS[k] = \ocS[1k] = \ocS[2k]\text{ for all }k\in I \cap J
    }\smallskip\\
    \ocExtSumSmall{\roleP}{i \in I}{
      \ocExtChoiceSmall{\mpLab[i]}{\ocS[1i]}{\ocC[1i]}}
     & \ocBinMerge[\ocChi] &
    \ocExtSumSmall{\roleP}{j \in J}{
      \ocExtChoiceSmall{\mpLab[j]}{\ocS[2j]}{\ocC[2j]}}
    & = &
    \ocRecord{\roleP}{\left(
    \makecell[c]{
    \ocRecvWrapSmall{\mpLab[i]}{\ocS[1i]}{\ocC[1i]}{i\in I \setminus J}\ \cup\\
    \ocRecvWrapSmall{\mpLab[j]}{\ocS[2j]}{\ocC[2j]}{j\in J \setminus I}\ \cup\\
    \ocRecvWrapSmall{\mpLab[k]}{\ocS[k]}{\ocC[1k] \ocBinMerge[\ocChi] \ocC[2k]}{k\in I \cap J}
    }\right)}{}\\
    &  &  &
    \multicolumn{2}{l}{
      \text{where }\ocS[k] = \ocS[1k] = \ocS[2k]\text{ for all }k\in I \cap J
    }\smallskip\\
    \end{array}\)}}
  \centerline{\scalebox{0.91}{\(
    \begin{array}{r|l}
      \begin{array}{ccl}
    \ocRec{\ocX}{\ocC[1]} \ocBinMerge[\ocChi] \ocC[2]
    & = &
    \left\{
    \begin{array}{l}
      \ocZ
        \quad
        \text{if}\ \mapsubst{\ocZ}{(\ocRec{\ocX}{\ocC[1]},\ocC[2])} \in \ocChi\\
      \ocRec{\ocZ}{\left(\ocC[1]\subst{\ocX}{\ocRec{\ocX}{\ocC[1]}}
        \ \ocBinMerge[{\ocChi \cdot \mapsubst{\ocZ}{(\ocRec{\ocX}{\ocC[1]},\ocC[2])}}]\ %
        \ocC[2]\right)}
      \quad \ocZ \text{\ fresh}\\
      \qquad \text{otherwise}
    \end{array}\right.
      \\[4mm]
    \ocC[1] \ocBinMerge[\ocChi] \ocRec{\ocX}{\ocC[2]}
    & = &
    \left\{
    \begin{array}{l}
      \ocZ
        \quad
        \text{if}\ \mapsubst{\ocZ}{(\ocC[1],\ocRec{\ocX}{\ocC[2]})} \in \ocChi\\
      \ocRec{\ocZ}{\left(\ocC[1]
        \ \ocBinMerge[{\ocChi \cdot \mapsubst{\ocZ}{(\ocC[1],\ocRec{\ocX}{\ocC[2]})}}]\ %
        \ocC[2]\subst{\ocX}{\ocRec{\ocX}{\ocC[2]}}\right)}
        \quad \ocZ \text{\ fresh}\\
      \qquad \text{otherwise}
    \end{array}\right.
      \end{array}
    \quad & \quad%
      \begin{array}{rclcl}
        \ocX & \ocBinMerge[\ocChi] & \ocX & = & \ocX\\
        \ocUnit & \ocBinMerge[\ocChi] & \ocUnit & = & \ocUnit
      \end{array}
    \end{array}
    \)}}
  \caption{Merging of channel vectors \framebox{$\ocMerge{i \in 1 .. n}{\ocC[i]} = ((\ocC[1]\ocBinMerge[\ocEmpty]\ocC[2])\ocBinMerge[\ocEmpty]\cdots\ocBinMerge[\ocEmpty]\ocC[n])$}\label{fig:merge}}
\end{figure}

\section{More Examples on Global Combinators}
In this section, we review more examples of global combinators.

\begin{EX}[Global combinator evaluation]\thmstart
\label{ex:auth:gen:full}
Let
$\ocS[1]=\ocS_{\{\roleC,\roleS,\labOk,0\}}$,
$\ocS[1]=\ocS_{\{\roleC,\roleS,\labCancel,0\}}$,
$\ocS[3]=\ocS_{\{\roleS,\roleC,\labAuth,0\}}$,
$\ocS[4]= \ocS_{\{\roleS,\roleA,\labOk,1\}}$ and
$\ocS[5] = \ocS_{\{\roleS,\roleA,\labCancel,2\}}$. Then:\smallskip\\
\centerline{\scalebox{0.91}{\(
\begin{array}{l}
\GCCVR{\ogtCommBin{\roleS}{\roleC}{\labOk}{}{\ogtEnd}{\labCancel}{}{\ogtEnd}}{\roleC,\roleS}^{\mpS}\\
=\left\ocTupOpen
  \ocExtSumSmall{\roleS}{}{
    \ocExtChoiceSmall{\labOk}{\ocS[1]}{\ocUnit},
    \ocExtChoiceSmall{\labCancel}{\ocS[2]}{\ocUnit}},
  \,\,\ocIntSumSmall{\roleC}{}{
    \ocIntChoiceSmall{\labOk}{\ocS[1]}{\ocUnit},
    \ocIntChoiceSmall{\labCancel}{\ocS[2]}{\ocUnit}}
 \right\ocTupClose
\\
\GCCVR{\ogtG[Auth]}{\roleC,\roleS}^{\mpS} \qquad \text{(from \Cref{ex:globalcombinatorauth})}\\
= \GCCVR{\ogtComm{\roleC}{\roleS}{}{\labAuth}{}{\left(\ogtCommBin{\roleS}{\roleC}{\labOk}{}{\ogtEnd}{\labCancel}{}{\ogtEnd}\right)}}{\roleC,\roleS}^{\mpS}\\
=\ocTuple{\makecell[l]{
\ocIntSum{\roleS}{}{
  \ocIntChoice{\labAuth}{\ocS[3]}{
    \ocExtSumSmall{\roleS}{}{
    \ocExtChoiceSmall{\labOk}{\ocS[1]}{\ocUnit},
    \ocExtChoiceSmall{\labCancel}{\ocS[2]}{\ocUnit}}}},\\
\,\,
\ocExtSum{\roleC}{}{
  \ocExtChoice{\labAuth}{\ocS[3]}{
    \ocIntSumSmall{\roleC}{}{
    \ocIntChoiceSmall{\labOk}{\ocS[1]}{\ocUnit},
    \ocIntChoiceSmall{\labCancel}{\ocS[2]}{\ocUnit}}}}}}
\\
\GCCVR{\ogtRec{\ogtRecVar}{\ogtComm{\roleC}{\roleS}{}{\labOk}{}{\ogtRecVar}}}{\roleC,\roleS}^{\mpS}
=\ocPair
{\ocRec{\ocX[\roleC]}{\ocIntSumSmall{\roleS}{}{\ocIntChoiceSmall{\labOk}{\ocS[1]}{\ocX[\roleC]}}}}
{\,\,\ocRec{\ocX[\roleS]}{\ocExtSumSmall{\roleC}{}{\ocExtChoiceSmall{\labOk}{\ocS[2]}{\ocX[\roleS]}}}}
\\[2mm]
\text{This example shows how channel vectors for roles not participating in a choice (here $\roleA$) are merged:}\\[2mm]
\GCCVR{\ogtCommBin{\roleS}{\roleC}
  {\labOk}{}{\left(\ogtComm{\roleS}{\roleA}{}{\labOk}{}{\ogtEnd}\right)}
  {\labCancel}{}{\left(\ogtComm{\roleS}{\roleA}{}{\labCancel}{}{\ogtEnd}\right)}}{\roleS,\roleC,\roleA}^{\mpS}\\
=\ocTuple{\makecell[l]{
  \ocIntSum{\roleC}{}{
    \ocIntChoice{\labOk}{\ocS[1]}{
      \ocIntSum{\roleA}{}{\ocIntChoice{\labOk}{\ocS[4]}{\ocUnit}}},
    \ocIntChoice{\labCancel}{\ocS[2]}{
      \ocIntSum{\roleA}{}{\ocIntChoice{\labCancel}{\ocS[4]}{\ocUnit}}}},\\
\,\,\ocExtSum{\roleS}{}{
    \ocExtChoice{\labOk}{\ocS[1]}{\ocUnit},
    \ocExtChoice{\labCancel}{\ocS[2]}{\ocUnit}},\,\,
  \ocExtSum{\roleS}{}{
    \ocExtChoice{\labOk}{\ocS[4]}{\ocUnit},
    \ocExtChoice{\labCancel}{\ocS[5]}{\ocUnit}}
}}
\end{array}
\)}}
\end{EX}

The following example illustrates channel vectors and the usage of $\ocUnfold{}$. 
for the syntax of processes we use \mustL, defined in \S~\ref{app:mio}

\begin{EX}\thmstart 
  Let\smallskip\\
\centerline{\scalebox{0.81}{\(
\begin{array}{rcl}
\ogtG[Cal] &=&
\ogtRec{\ogtRecVar}{\ogtCommBin{\roleC}{\roleS}
  {\labLoop}{}{\ogtRecVar}
  {\labStop}{}{\left(\ogtComm{\roleS}{\roleC}{}{\labStop}{}{\ogtEnd}\right)}}
\\
\oeE[\roleC] &=&
  \oeSend{\ocX[1]}{\ocX[0]\#\roleS\#\labLoop}{\ocUnitPayload}{
  \oeSend{\ocX[2]}{\ocX[1]\#\roleS\#\labLoop}{\ocUnitPayload}{
  \oeSend{\ocX[3]}{\ocX[2]\#\roleS\#\labStop}{\ocUnitPayload}{}}}\\
& &\oeBranchSingle{\ocVariantPairPat{\labStop}{\UNL}{\ocX[4]}}{\ocX[3]}{\roleS}{\oeNil}\\
\oeE[\roleS] &=& \oeLetrecAbbrev{\oeCall{\oeRecVar}{\ocX}{=}\oeE[\roleS0]}{\oeCall{\oeRecVar}{\ocXi[0]}}\\
\oeE[\roleS0] &=&
  \oeBranchRaw{\ocX}{\roleC}{
    \oeChoice{\labLoop}{\UNL}{\ocX[1]}{\oeCall{\oeRecVar}{\ocX[1]}};\,
    \oeChoice{\labStop}{\UNL}{\ocX[2]}{
      \oeSend{\ocX[3]}{\ocX[2]\#\roleC\#\labStop}{\ocUnitPayload}{\oeNil}}}\\
\oeE[{\tt Cal}] &=& \oeInit{\ocX[0],\ocXi[0]}{\ogtG[Cal]}{\left(\oeE[\roleC] \oePar \oeE[\roleS]\right)}
\end{array}
\)}}
Then, $\oeE[{\tt Cal}] \rightarrow\rightarrow \oeEi[{\tt Cal}] = \oeRes{\ocS[1],\ocS[2],\ocSi}{\bigl(
  \oeE[\roleC]\subst{\ocX[0]}{\ocC[\roleC]}
  \ \big|\ %
  \oeLetrecAbbrev{\oeCall{\oeRecVar}{\ocX}{=}\oeE[\roleS0]}{(\oeE[\roleS0]\subst{\ocXi[0]}{\ocC[\roleS]})}
  \bigr)}$
where\smallskip\\
\centerline{\scalebox{0.91}{\(
\begin{array}{rcl}
  \ocC[\roleC] &=&
  \ocRec{\ocX[\roleC]}{\ocIntSum{\roleS}{}{
    \ocIntChoiceSmall{\labLoop}{\ocS[1]}{\ocX[\roleC]},
    \ocIntChoiceSmall{\labStop}{\ocS[2]}{
      \ocExtSum{\roleS}{}{\ocExtChoice{\labStop}{\ocSi}{\ocUnit}}}}}\\
\ocC[\roleS] &=&
  \ocRec{\ocX[\roleS]}{\ocExtSum{\roleS}{}{
    \ocExtChoiceSmall{\labLoop}{\ocS[1]}{\ocX[\roleS]},
    \ocExtChoiceSmall{\labStop}{\ocS[2]}{
      \ocIntSum{\roleS}{}{\ocIntChoice{\labStop}{\ocSi}{\ocUnit}}}}}\\
& & \left(= \ocRec{\ocX[\roleS]}{\ocRecord{\roleS}{
    \left[\ocWrapperElem{\ocS[1]}{\ocVariant{\labLoop}{\ocPair{\HOLE}{\ocX[\roleS]}}},
    \ocWrapperElem{\ocS[2]}{\ocVariant{\labStop}{\ocPair{\HOLE}{
      \ocIntSum{\roleS}{}{\ocIntChoice{\labStop}{\ocSi}{\ocUnit}}}}}\right]}{}}\right)\\
\end{array}\)}}\\
See that\smallskip\\
\centerline{\scalebox{0.91}{\(\begin{array}{rcl}
\ocUnfold{\ocC[\roleC]\#\roleS\#\labLoop} &=& \ocPair{\ocS[1]}{\ocC[\roleC]}\\
\ocUnfold{\ocC[\roleC]\#\roleS\#\labStop} &=& \ocPair{\ocS[2]}{\ocExtSum{\roleS}{}{\ocExtChoice{\labStop}{\ocSi}{\ocUnit}}}
\\
\ocUnfold{\ocC[\roleS]\#\roleC}&=&
  \left[\ocRecvWrapElem{\labLoop}{\ocS[1]}{\ocC[\roleC]}, \ocRecvWrapElem{\labStop}{\ocS[2]}{
    \ocIntSum{\roleC}{}{\ocIntChoice{\labStop}{\ocSi}{\ocUnit}} }\right]\\
  & & \left(= \left[\ocWrapperElem{\ocS[1]}{\ocVariant{\labLoop}{\ocPair{\HOLE}{\ocC[\roleC]}}}, \ocWrapperElem{\ocS[2]}{\ocVariant{\labStop}{\ocPair{\HOLE}{
        \ocIntSum{\roleC}{}{\ocIntChoice{\labStop}{\ocSi}{\ocUnit}}
    }}}\right]\right)
\end{array}\)}}\\
and each time $\roleC$lient sends a label, $\roleS$erver
makes an external choice between $\ocS[1]$ and $\ocS[2]$,
and they reduce as follows:
$\oeEi[{\tt Cal}] \rightarrow^{6}
  \oeRes{\ocS[1],\ocS[2],\ocSi}{\left(\oeNil\ \big|\ \oeLetrecAbbrev{\oeCall{\oeRecVar}{\ocX}{=}\oeE[\roleS0]}{\oeNil}\right)}
  \equiv \oeNil$.
\end{EX}

The types are further elaborated by 
{\em subtyping} with I/O types \cite{PiSa96b}
which is defined in \Cref{def:subtyping}. %

\begin{EX}[Merging via subtyping]\thmstart
  The following typing involves {\em merging}
  where the behaviour of two or more channel vector types in the branches
  are mixed into one, as\smallskip\\
  \scalebox{0.81}{\(
    \ogtEnvEntailsEx{\roleS,\roleC,\roleA}{}{
      \ogtCommBin{\roleS}{\roleC}
                 {\labOk}{}{\left(\ogtComm{\roleS}{\roleA}{}{\labOk}{}{\ogtEnd}\right)}
                 {\labCancel}{}{\left(\ogtComm{\roleS}{\roleA}{}{\labCancel}{}{\ogtEnd}\right)}
    }{}\)}\linebreak
  \scalebox{0.91}{\(\otT[\roleS]\otTimes\otT[\roleC]\otTimes\otT[\roleA]\)}
  where
    $\otT[\roleS] = \otIntSumBig{\roleC}{}{
      \otIntChoice{\labOk}{\otO}{
        \otIntSumSmall{\roleA}{}{\otIntChoice{\labOk}{\otO}{\otUnit}},
      \otIntChoice{\labCancel}{\otO}{
        \otIntSumSmall{\roleA}{}{\otIntChoice{\labCancel}{\otO}{\otUnit}}}}}$,
    and
      $\otT[\roleC] = \otT[\roleA] = \otExtSumSmall{\roleS}{}{
      \otExtChoice{\labOk}{\otO}{\otUnit},
      \otExtChoice{\labCancel}{\otO}{\otUnit}}
    $.
  See that the continuations\linebreak
  $(\ogtComm{\roleS}{\roleA}{}{\labOk}{}{\ogtEnd})$ and
  $(\ogtComm{\roleS}{\roleA}{}{\labCancel}{}{\ogtEnd})$
  have the channel vector types
  $\otExtSumSmall{\roleS}{}{\otExtChoice{\labOk}{\otO}{\otUnit}}$ and
  $\otExtSumSmall{\roleS}{}{\otExtChoice{\labCancel}{\otO}{\otUnit}}$
  at role $\roleA$ respectively,
  where each of them receives label $\labOk$ and $\labCancel$ from $\roleS$.
  By subtyping, they are amalgamated into a {\em common super type}
  $\otExtSumSmall{\roleS}{}{\otExtChoice{\labOk}{\otO}{\otUnit},\otExtChoice{\labCancel}{\otO}{\otUnit}}$
  which can now receive both labels.
  This is underpinned by the subtyping relation as
  $\otExtSumSmall{\roleS}{}{\otExtChoice{\labOk}{\otO}{\otUnit}}
    \otSub \otExtSumSmall{\roleS}{}{\otExtChoice{\labOk}{\otO}{\otUnit},\otExtChoice{\labCancel}{\otO}{\otUnit}}$
  (similar for $\labCancel$)
  which is justified by \inferrule{\iruleOTGSub}.
\end{EX}

\begin{EX}[Recursion]\thmstart
  The following typing derivation is valid under {$\roleSet=\roleS,\roleC,\roleA$}:
\vspace{0.5em}

\noindent\scalebox{0.75}{\(
  \inference{
    \inference{
      \inference{
        \ogtEnvEntails{
          \otEnvMap{\ogtRecVar}{\otRecVar[\roleS]\otTimes\otRecVar[\roleC]\otTimes\otRecVar[\roleA]}
        }{
          \ogtRecVar
        }{
          \otRecVar[\roleS]\otTimes\otRecVar[\roleC]\otTimes\otRecVar[\roleA]
        }
      }{
        \ogtEnvEntails{
          \otEnvMap{\ogtRecVar}{\otRecVar[\roleS]\otTimes\otRecVar[\roleC]\otTimes\otRecVar[\roleA]}
        }{
          \ogtComm{\roleS}{\roleA}{}{\labOk}{}{\ogtRecVar}
        }{\otIntSum{\roleA}{}{\otIntChoice{\labOk}{\otO}{\otRecVar[\roleS]}}
          \otTimes
          \otRecVar[\roleC]
          \otTimes
          \otExtSum{\roleS}{}{\otExtChoice{\labOk}{\otO}{\otRecVar[\roleA]}}
        }
      }
    }{
    \ogtEnvEntails{
        \otEnvMap{\ogtRecVar}{\otRecVar[\roleS]\otTimes\otRecVar[\roleC]\otTimes\otRecVar[\roleA]}}{
      \ogtComm{\roleS}{\roleC}{}{\labOk}{}{
        \left(\ogtComm{\roleS}{\roleA}{}{\labOk}{}{\ogtRecVar}\right)}
    }{\otIntSum{\roleC}{}{\otIntChoice{\labOk}{\otO}{
            \otIntSum{\roleA}{}{\otIntChoice{\labOk}{\otO}{
              \otRecVar[\roleS]}}}}
        \otTimes
        \otExtSum{\roleS}{}{\otExtChoice{\labOk}{\otO}{\otRecVar[\roleC]}}
        \otTimes
        \otExtSum{\roleS}{}{\otExtChoice{\labOk}{\otO}{\otRecVar[\roleA]}}}
    }
  }{
    \ogtEnvEntails{}{
      \ogtRec{\ogtRecVar}{\left(
        \ogtComm{\roleS}{\roleC}{}{\labOk}{}{
          \left(\ogtComm{\roleS}{\roleA}{}{\labOk}{}{\ogtRecVar}\right)}\right)}
    }{\left(\makecell[l]{
        \otRec{\otRecVar[\roleS]}{
          \otIntSum{\roleC}{}{\otIntChoice{\labOk}{\otO}{
            \otIntSum{\roleA}{}{\otIntChoice{\labOk}{\otO}{
              \otRecVar[\roleS]}}}}}
        \otTimes\\
        \ \ \otRec{\otRecVar[\roleC]}{\otExtSum{\roleS}{}{\otExtChoice{\labOk}{\otO}{\otRecVar[\roleC]}}}
        \otTimes
        \otRec{\otRecVar[\roleA]}{\otExtSum{\roleS}{}{\otExtChoice{\labOk}{\otO}{\otRecVar[\roleA]}}}
      }\right)}
  }
  \)}
\end{EX}

\begin{EX}[Loops and the finished session]\thmstart
The following example shows the usage of the
function $\otFix{\cdot}{\cdot}$ in the rule \inferrule{\iruleOTGRec}
comes from the corresponding case in the End Point Projection in MPST \cite{scalas19less}.%
It declares the termination of the session for a role
in a loop in which the role in question
never participate in.
\[
\ogtEnvEntailsEx{\roleP,\roleQ,\roleR}{}{\left(\ogtRec{\ogtRecVar}{\ogtComm{\roleP}{\roleR}{}{\labOk}{}{\ogtRecVar}}\right)}
{\otRec{\otRecVar[\roleP]}{\otIntSumSmall{\roleR}{}{\otIntChoice{\labOk}{\otO}{\otRecVar[\roleP]}}}
\otTimes\otUnit\otTimes
\otRec{\otRecVar[\roleR]}{\otExtSumSmall{\roleP}{}{\otExtChoice{\labOk}{\otO}{\otRecVar[\roleR]}}}
}
\]
where the channel vector type for $\roleQ$ is the finished session $\otUnit$
because $\otFix{\otRecVar[\roleQ]}{\otRecVar[\roleQ]}=\otUnit$.
\end{EX}

\section{\mustL{}: A minimal \must calculus}
\label{app:mio}
We introduces a minimal functional calculus, \mustL{} and its typing systems. 
The calculus distils the main features required for embedding session types in \OCaml, 
notably equi-recursive types, record and variant types, structural subtyping, and simply-typed I/O channels. 
We prove the type soundness for \mustL{} (Theorem~\ref{lem:SubjectReduction}).

\subsection{\mustL{}: Syntax and Dynamic Semantics}
\label{sec:language}

This section introduces the syntax and
operational semantics of \OCAMLMPST{}.
\subsubsection{\mustL{} Program}\label{subsec:mio_syntax}
We introduce the syntax of \mustL{} program, which is written 
by the programmer.   

\begin{DEF}[\mustL{} program]\thmstart \label{def:program}
The \emph{program} (or expression) of \mustL{} is defined as:\smallskip\\
\centerline{$\begin{array}{c|c}
 \begin{array}{rcll}
   & & \hspace{-3mm}\oeE  \quad \grmeq\quad\\
    & & \oeInit{\ocX[1],\elip,\ocX[n]}{\gocaml}{\oeE}
       & \text{\footnotesize(initiation)}\\
    & & \oeSend{\ocX}{\ocY\boldsymbol{\#}\roleQ\boldsymbol{\#}\mpLab}{\ocV}{\oeE}
       & \text{\footnotesize(send)}\\
    & & \oeRecv{\ocX}{\ocY\boldsymbol{\#}\roleQ}{\oeE}
       & \text{\footnotesize(receive)}\\
   & & \hspace{-3mm}\ocV  \quad \grmeq\quad x,y,z,...  \grmor \ocUnit
   & \text{\footnotesize(values)}\\
 \end{array}
&
 \begin{array}{rcll}
\\
    & & \oeMatch{\ocX}{\mpLab[i]}{\ocX[i]}{\ocY[i]}{\oeE[i]}{i \in I}
       & \text{\footnotesize(pattern match)}\\
    & & \oeNil \grmor \oeE \oePar \oeEi %
       & \text{\footnotesize(unit, par)}\\
    & & \oeLetrecAbbrev{\oeD}{\oeE}
       & \text{\footnotesize(recursion)}\\
& & \hspace{-1mm}\oeD  \ \ \grmeq  \ \ \oeCallDec{\oeRecVar}{\tilde{\ocX}}{=}\oeE
& \text{\footnotesize(declaration)}
 \end{array}
\end{array}$}
\\[1mm]
We assume mutually disjoint sets of
{\em variables} ($\ocX,\ocY,\ldots$),
and {\em function variables} ($\oeRecVar, \oeRecVari, \ldots$).
    In
    $\left(\oeFmt{\oeLetKwd\SP\ocX\SP=\elip\oeInKwd\SP\oeE}\right)$,
variable $\ocX$ in $\oeE$ is bound. Similarly,
    $\left(\oeMatch{\elip}{\mpLab[i]}{\ocX[i]}{\ocY[i]}{\oeE[i]}{i \in
        I}\right)$ and
  $\left(\oeFmt{\oeLetrecKwd\SP\oeRecVar(\ocX[1],\elip,\ocX[n])\SP=\SP\oeE\SP\oeInKwd\SP\elip}\right)$,
    variables $\ocX[i]$ and $\ocY[i]$ in $\oeE[i]$
    ($i \in I$) and
    $\ocX[1],\elip,\ocX[n]$ in $\oeE$
are bound, respectively. 
    $\oeLetrec{\oeRecVar}{\elip}{\oeE[1]}{\oeE[2]}$ binds $\oeRecVar$ in both
    $\oeE[1]$ and $\oeE[2]$.
    $\fv{\oeE}/\fn{\oeE}$ denote the set of free variables/names (introduced later) in
    $\oeE$.
$\ffv{\oeE}$ is the set of free function variables in $\oeE$,
    and $\dfv{\oeD}$ is the set of declared function variables in $\oeD$
    (i.e. $\dfv{\oeCallDec{\oeRecVar}{\tilde{\ocX}}{=}{\oeE}}=\{\oeRecVar\}$).
$\ocUnit$ denotes unit value and  
$\UNL$ stands for unused binding variables.
\end{DEF}

{\textbf{\emph{Program}}}
includes
{\textbf{\emph{initiation}}} which generates
a series of interconnected channels
from a global combinator $\ogtG$,
each of which corresponds to a role occurring in $\ogtG$.
This expression corresponds to
Line 1 and Line 8 in
\Cref{fig:full:impl}.
\textbf{\emph{Output}} expression
$\oeSend{\ocX}{\ocY\#\roleQ\#\mpLab}{\ocV}{\oeE}$
sends label $\mpLab$ with payload $\ocV$ via channel $\ocY$ to role $\roleQ$,
then binds the continuation to $\ocX$,
and proceeds to $\oeE$.
\textbf{\emph{Input}} expression
$\oeRecv{\ocX}{\ocY\#\roleQ}{\oeE}$
receives on $\ocY$ from $\roleQ$
then binds the received value
to $\ocX$,
and proceeds to $\oeE$.
The received value will have the form
$\ocVariant{\mpLab}{\ocPair{\ocV[1]}{\ocV[2]}}$
where $\mpLab$ and $\ocV[1]$ are the label and payload sent from $\roleQ$,
and $\ocV[2]$ is a continuation.
The received value is decomposed by
{\textbf{\emph{pattern matching}}} expression 
$\oeMatch{\ocX}{\mpLab[i]}{\ocX[i]}{\ocY[i]}{\oeE[i]}{i\in I}$ which
matches
against patterns $\ocVariant{\mpLab[i]}{\ocPair{\ocX[i]}{\ocY[i]}}$ ($i\in I$),
and if $\mpLab=\mpLab[k]$,
it continues to $\oeE[k]$ after simultaneously substituting $\ocX[k]$ and $\ocY[k]$ with $\ocV[1]$ and $\ocV[2]$, respectively.
{\textbf{\emph{Recursive function definition}}}
$\oeLetrec{\oeRecVar}{\ocX[1],\elip,\ocX[n]}{\oeE[1]}{\oeE[2]}$ defines a recursive function $\oeRecVar$ with
parameters $\ocX[1],\elip,\ocX[n]$ and body $\oeE[1]$ which is local to $\oeE[2]$.
A unit value $\oeNil$ represents an inactive thread.
{\textbf{\emph{Parallel}}} $\oeE[1] \oePar \oeE[2]$ represents two
  threads running concurrently.
$\oeCall{\oeRecVar}{\tilde{\ocV}}$ is the {\textbf{\emph{function application}}}.

We use the following shorthand for expressions with $\ocZ$ fresh:\\
\centerline{\scalebox{0.91}{\(
  \begin{array}{rcl}
  \oeBranchMult{\ocX_0}{\roleQ}{\mpLab[i]}{\ocX[i]}{\ocY[i]}{\oeE[i]}{i \in I}%
  & \eqdeff
  & \oeBranchSingle{\ocZ}{\ocX_0}{\roleQ}{
    \oeMatch{\ocZ}{\mpLab[i]}{\ocX[i]}{\ocY[i]}{\oeE[i]}{i \in I}}\\
  \oeBranchSingle{\ocVariantPairPat{\mpLab}{\ocX}{\ocY}}{\ocX[0]}{\roleQ}{\oeE}
  & \eqdeff
  & \oeBranchSingle{\ocZ}{\ocX[0]}{\roleQ}{
    \oeMatch{\ocZ}{\mpLab}{\ocX}{\ocY}{\oeE}{}}
\end{array}\)}}

\begin{EX}\thmstart\label{ex:auth:mio}
The following expressions implement the protocol
in Example \ref{ex:globalcombinatorauth}:
\[
\eAuth =
\oeInit{\ocX,\ocXi}{\gAuth}{\left(\oeE[\roleC] \oePar
    \oeE[\roleS]\right)}
\]
where, with $\fv{\oeE[\roleC]}=\{\ocX\}$, $\fv{\oeE[\roleS]}=\{\ocXi\}$,
and $\fv{\eAuth}=\{\}$,
and $\oeE[\roleC]$ and $\oeE[\roleS]$ are following:\smallskip\linebreak
{\scalebox{0.91}{\(
\begin{array}{l}
\oeE[\roleC] = \oeSend{\ocX[1]}{\ocX\#\roleS\#\labAuth}{\text{\oCODE{"passwd"}}}{
  \left(\oeBranchRaw{\ocX[1]}{\roleS}{
    \oeChoice{\labOk}{\UNL}{\ocX[2]}{\oeNil}; \oeChoice{\labCancel}{\UNL}{\ocX[3]}{\oeNil}}\right)}\\
\oeE[\roleS] =
\oeBranchSingle{\ocVariantPairPat{\labAuth}{\UNL}{\ocX[1]}}{\ocXi}{\roleC}{
  \oeSend{\ocX[2]}{\ocX[1]\#\roleC\#\labOk}{\text{\oCODE{"ok"}}}{\oeNil}}.
\end{array}
\)}}
\end{EX}

\begin{EX}\thmstart
  Let\smallskip\\
\centerline{\scalebox{0.81}{\(
\begin{array}{rcl}
\ogtG[Cal] &=&
\ogtRec{\ogtRecVar}{\ogtCommBin{\roleC}{\roleS}
  {\labLoop}{}{\ogtRecVar}
  {\labStop}{}{\left(\ogtComm{\roleS}{\roleC}{}{\labStop}{}{\ogtEnd}\right)}}
\\
\oeE[\roleC] &=&
  \oeSend{\ocX[1]}{\ocX[0]\#\roleS\#\labLoop}{\ocUnitPayload}{
  \oeSend{\ocX[2]}{\ocX[1]\#\roleS\#\labLoop}{\ocUnitPayload}{
  \oeSend{\ocX[3]}{\ocX[2]\#\roleS\#\labStop}{\ocUnitPayload}{}}}\\
& &\oeBranchSingle{\ocVariantPairPat{\labStop}{\UNL}{\ocX[4]}}{\ocX[3]}{\roleS}{\oeNil}\\
\oeE[\roleS] &=& \oeLetrecAbbrev{\oeCall{\oeRecVar}{\ocX}{=}\oeE[\roleS0]}{\oeCall{\oeRecVar}{\ocXi[0]}}\\
\oeE[\roleS0] &=&
  \oeBranchRaw{\ocX}{\roleC}{
    \oeChoice{\labLoop}{\UNL}{\ocX[1]}{\oeCall{\oeRecVar}{\ocX[1]}};\,
    \oeChoice{\labStop}{\UNL}{\ocX[2]}{
      \oeSend{\ocX[3]}{\ocX[2]\#\roleC\#\labStop}{\ocUnitPayload}{\oeNil}}}\\
\oeE[{\tt Cal}] &=& \oeInit{\ocX[0],\ocXi[0]}{\ogtG[Cal]}{\left(\oeE[\roleC] \oePar \oeE[\roleS]\right)}
\end{array}
\)}}
Then, $\oeE[{\tt Cal}] \rightarrow\rightarrow \oeEi[{\tt Cal}] = \oeRes{\ocS[1],\ocS[2],\ocSi}{\bigl(
  \oeE[\roleC]\subst{\ocX[0]}{\ocC[\roleC]}
  \ \big|\ %
  \oeLetrecAbbrev{\oeCall{\oeRecVar}{\ocX}{=}\oeE[\roleS0]}{(\oeE[\roleS0]\subst{\ocXi[0]}{\ocC[\roleS]})}
  \bigr)}$
where\smallskip\\
\centerline{\scalebox{0.91}{\(
\begin{array}{rcl}
  \ocC[\roleC] &=&
  \ocRec{\ocX[\roleC]}{\ocIntSum{\roleS}{}{
    \ocIntChoiceSmall{\labLoop}{\ocS[1]}{\ocX[\roleC]},
    \ocIntChoiceSmall{\labStop}{\ocS[2]}{
      \ocExtSum{\roleS}{}{\ocExtChoice{\labStop}{\ocSi}{\ocUnit}}}}}\\
\ocC[\roleS] &=&
  \ocRec{\ocX[\roleS]}{\ocExtSum{\roleS}{}{
    \ocExtChoiceSmall{\labLoop}{\ocS[1]}{\ocX[\roleS]},
    \ocExtChoiceSmall{\labStop}{\ocS[2]}{
      \ocIntSum{\roleS}{}{\ocIntChoice{\labStop}{\ocSi}{\ocUnit}}}}}\\
& & \left(= \ocRec{\ocX[\roleS]}{\ocRecord{\roleS}{
    \left[\ocWrapperElem{\ocS[1]}{\ocVariant{\labLoop}{\ocPair{\HOLE}{\ocX[\roleS]}}},
    \ocWrapperElem{\ocS[2]}{\ocVariant{\labStop}{\ocPair{\HOLE}{
      \ocIntSum{\roleS}{}{\ocIntChoice{\labStop}{\ocSi}{\ocUnit}}}}}\right]}{}}\right)\\
\end{array}\)}}\\
See that\smallskip\\
\centerline{\scalebox{0.91}{\(\begin{array}{rcl}
\ocUnfold{\ocC[\roleC]\#\roleS\#\labLoop} &=& \ocPair{\ocS[1]}{\ocC[\roleC]}\\
\ocUnfold{\ocC[\roleC]\#\roleS\#\labStop} &=& \ocPair{\ocS[2]}{\ocExtSum{\roleS}{}{\ocExtChoice{\labStop}{\ocSi}{\ocUnit}}}
\\
\ocUnfold{\ocC[\roleS]\#\roleC}&=&
  \left[\ocRecvWrapElem{\labLoop}{\ocS[1]}{\ocC[\roleC]}, \ocRecvWrapElem{\labStop}{\ocS[2]}{
    \ocIntSum{\roleC}{}{\ocIntChoice{\labStop}{\ocSi}{\ocUnit}} }\right]\\
  & & \left(= \left[\ocWrapperElem{\ocS[1]}{\ocVariant{\labLoop}{\ocPair{\HOLE}{\ocC[\roleC]}}}, \ocWrapperElem{\ocS[2]}{\ocVariant{\labStop}{\ocPair{\HOLE}{
        \ocIntSum{\roleC}{}{\ocIntChoice{\labStop}{\ocSi}{\ocUnit}}
    }}}\right]\right)
\end{array}\)}}\\
and each time $\roleC$lient sends a label, $\roleS$erver
makes an external choice between $\ocS[1]$ and $\ocS[2]$,
and they reduce as follows:
$\oeEi[{\tt Cal}] \rightarrow^{6}
  \oeRes{\ocS[1],\ocS[2],\ocSi}{\left(\oeNil\ \big|\ \oeLetrecAbbrev{\oeCall{\oeRecVar}{\ocX}{=}\oeE[\roleS0]}{\oeNil}\right)}
  \equiv \oeNil$.
\end{EX}

\subsubsection{Dynamic Semantics of \mustL{}}
\label{sec:reduction}
We introduce a reduction semantics of expressions,
which is a standard MPST $\pi$-calculus, with
extra handling on channel vectors.

\begin{DEF}\label{def:sem}\thmstart
The reduction relation $\rightarrow$ of the expressions
is defined by the rules
in Fig.~\ref{fig:reduction}.
The syntax of \mustL{}
in \Cref{def:program} is extended to the \emph{runtime syntax}
as follows:
\[
\begin{array}{rcl}
e  & ::= &
\oeSend{\ocX}{\ocC\boldsymbol{\#}\roleQ\boldsymbol{\#}\mpLab}{\ocCi}{\oeE}
\grmor
\oeRecv{\ocX}{\ocC\boldsymbol{\#}\roleQ}{\oeE}
\grmor
\oeMatch{\ocC}{\mpLab[i]}{\ocX[i]}{\ocY[i]}{\oeE[i]}{i \in I}\\
& \grmor  &
\oeCall{\oeRecVar}{\tilde{\ocC}} \grmor \oeRes{\ocS}{e}
\end{array}
\]
A {\em reduction context}\ $\OC$ is defined by the following grammar:\\[1mm]
\centerline{\(\OC ::= \OC \oePar \oeE \grmor \oeRes{\ocS}{\OC}
\grmor \oeLetrec{\oeRecVar}{\tilde{\ocX}}{\oeE}{\OC} \grmor \HOLE\)}
\end{DEF}
\textbf{\emph{Restriction}} $\oeRes{\ocS}{e}$
denotes session channel $s$ binding all free channels in the form of
$\ocS_{\{\roleP[j],\roleP[k],\mpLab,\blueI\}}$
which are generated by $\GCCVR{\gocaml}{}^s$.
 The structural congruence $\equiv$ (adapted from
 \cite{scalas19less}) is inductively defined by the rules in \Cref{fig:strcong}.%

\begin{figure}[t]
  \centerline{\scalebox{0.99}{\(
  \begin{array}{c}
    \oeE \oePar \oeEi \equiv \oeEi \oePar \oeE
    \qquad (\oeE \oePar \oeEi) \oePar \oeEii \equiv \oeE \oePar (\oeEi \oePar \oeEii)
    \qquad \oeE \oePar \oeNil \equiv \oeE
    \qquad \oeRes{\ocS}{\oeNil} \equiv \oeNil\\[1mm]
    \oeRes{\ocS}{\oeRes{\ocSi}{\oeE}} \equiv \oeRes{\ocSi}{\oeRes{\ocS}{\oeE}}
    \qquad \oeRes{\ocS}{\bigl(\oeE \oePar \oeEi\bigr)} \equiv \oeE \oePar \oeRes{\ocS}{\oeEi} \quad \text{if }\ocS\notin\fn{\oeE}\\[1mm]
    \oeLetrecAbbrev{\oeD}{\oeNil} \equiv \oeNil
    \qquad \oeLetrecAbbrev{\oeD}{\oeRes{\ocS}{\oeE}} \equiv \oeRes{\ocS}{\bigl(\oeLetrecAbbrev{\oeD}{\oeE}\bigr)} \quad\text{if }\ocS\notin\fn{\oeD}\\[1mm]
    \oeLetrecAbbrev{\oeD}{\bigl(\oeE \oePar \oeEi\bigr)} \equiv \oeFmt{\bigl(\oeLetrecAbbrev{\oeD}{\oeE}\bigr)} \oePar \oeEi \quad\text{if }\dfv{\oeD}\cap\ffv{\oeEi}=\emptyset\\[1mm]
      \oeLetrecAbbrev{\oeD}{\bigl(\oeLetrecAbbrev{\oeDi}{\oeE}\bigr)} \equiv \oeLetrecAbbrev{\oeDi}{\bigl(\oeLetrecAbbrev{\oeD}{\oeE}\bigr)}\\[1mm]
      \text{if }\bigl(\dfv{\oeD}\cup\ffv{\oeD}\bigr)\cap\dfv{\oeDi} = \bigl(\dfv{\oeDi}\cup\ffv{\oeDi}\bigr)\cap\dfv{\oeD} = \emptyset\\[1mm]
  \end{array}\)}}
  \caption{Structural Congruence rules \framebox{$\oeE\equiv\oeEi$}\label{fig:strcong}}
\end{figure}

\begin{figure}[t]
  \scalebox{0.90}{\(
  \begin{array}{c}
      \inference{%
        \inferrule{\iruleORedInit} \ &
\GCCVR{\gocaml}{}^{\mpS}
=\ocTuple{\ocC[1],\elip,\ocC[n]}
        \quad \mpS\ \text{fresh}
      }{
        \oeInit{\ocX[1],\ldots,\ocX[n]}{\gocaml}{(\oeE[1] \oePar \cdots \oePar \oeE[n])}
        \ \longrightarrow\ %
        \oeRes{\ocS}{(\oeE[1]\subst{\ocX[1]}{\ocC[1]} \oePar \cdots \oePar \oeE[n]\subst{\ocX[n]}{\ocC[n]})}
      }\\[2mm]
      \inference{%
          \inferrule{\iruleORedComm} \ &
          {\ocC[\roleP]\#\roleQ\#\mpLab}={\ocPair{\ocS_k}{\ocC[1]}} \quad
          {\ocC[\roleQ]\#\roleP}={\ocWrapper{\ocS[i]}{\ocH[i]}{i\in I}} \quad
                    \ocC[2]=\ocWrapperApp{\ocH[k]}{\ocCi}
          \ \  (\exists k\in I)
        }{
          \oeSend{\ocX}{\ocC[\roleP]\#\roleQ\#\mpLab}{\ocCi}{\oeE[1]}
            \mid \oeRecv{\ocY}{\ocC[\roleQ]\#\roleP}{\oeE[2]}
          \ \longrightarrow\ %
          \oeE[1]\subst{\ocX}{\ocC[1]} \oePar \oeE[2]\subst{\ocY}{\ocC[2]}
        }\\[2mm]
      \inference{%
        \inferrule{\iruleORedMatch} \ &
        \ocC = \ocVariant{\mpLab[k]}{\ocPair{\ocC[1]}{\ocC[2]}} \quad (\exists k \in I)
      }{
        \oeMatch{c}{\mpLab[i]}{\ocX[i]}{\ocY[i]}{\oeE[i]}{i\in I}
        \ \longrightarrow\ %
        \oeE[k]\subst{\ocX[k]}{\ocC[1]}\subst{\ocY[k]}{\ocC[2]}
      }\ %
      \inference{%
        \inferrule{\iruleORedCong} \ &
        \oeE\equiv \oeE[1] \quad \oeE[1] \longrightarrow \oeE[2] \quad \oeE[2] \equiv \oeEi
      }{
        \oeE \longrightarrow \oeEi
      }
      \\[2mm]
      \inference{%
        \inferrule{\iruleORedRec} \hspace{23em}
      }{
        \oeLetrec{\oeRecVar}{\tilde{\ocX}}{\oeE[1]}{\left(\oeCall{\oeRecVar}{\tilde{\ocC}} \oePar \oeE[2]\right)}%
        \longrightarrow
        \oeLetrec{\oeRecVar}{\tilde{\ocX}}{\oeE[1]}{\left(\oeE[1]\subst{\tilde{\ocX}}{\tilde{\ocC}} \oePar \oeE[2]\right)}%
      }
      \quad
      \inference{%
        \inferrule{\iruleORedCtx} \ &
        \oeE \longrightarrow \oeEi
      }{
        \OC[\oeE] \longrightarrow \OC[\oeEi]
      }
  \end{array}\)}
\caption{Reduction rules \framebox{$\oeE \longrightarrow \oeEi$}\label{fig:reduction}}
\end{figure}

The reduction rules of \mustL{}
are defined in \Cref{fig:reduction}.
Rule \inferrule{\iruleORedInit} generates
a tuple of channel vectors
$\ocTuple{\ocC[1],\elip,\ocC[n]}$
with fresh name $\ocS$
from a global combinator ($\GCCVR{\ogtG}{}$)
and then substitutes them to variables $\ocX[i]$ and continue to $\oeE$.
We assume that $\ocX[i]$ freely occurs in $\oeE[i]$ only, but not in $\oeE[j]$ where $i \neq j$.
The names
introduced by channel vectors are bound by restriction by $s$.
In rule \inferrule{\iruleORedComm},
the sender and receiver interact via two interconnected channel vectors
$\ocC[\roleP]$ and $\ocC[\roleQ]$ at role $\roleP$ and $\roleQ$,
respectively.
They have the form
$(\oeSendBare{\ocC[\roleP]\#\roleQ\#\mpLab[k]}{\ocCi})$
and
$(\oeRecvBare{\ocC[\roleQ]\#\roleP})$
which
communicates label $\mpLab[k]$ and payload $\ocCi$ from $\roleP$ to $\roleQ$.
On sender's side,
record projection $\ocC[\roleP]\#\roleQ\#\mpLab[k]$
yields $\ocPair{\ocS[k]}{\ocC[1]}$
where $\ocS[k]$ takes a form of
$\ocS_{\{\roleP,\roleQ,\mpLab[k],\blueIi\}}$.
On the receiver's side,
evaluation of $\ocC[\roleQ]\#\roleP$ yields
wrapped names $\ocRecvWrapSmall{\mpLab[i]}{\ocSi[i]}{\ocCi[i]}{i\in I}$
where each $\ocSi[i]$ takes a form of
$\ocSi_{\{\roleP,\roleQ,\mpLab[i],\blueJi\}}$.
The communication happens
if they both are generated
from the same global combinator and
interconnected via the same name $\ocS=\ocSi$ and the same index $\blueIi=\blueJi$.

After communication, the sender binds
$\ocC[1]$ to $\ocX$ and continues to $\oeE[1]$.
The receiver receives the variant value
\(\ocC[2]
= \ocWrapperApp{\ocH[k]}{\ocCi} = \ocWrapperApp{\ocVariantSmall{\mpLab[k]}{\ocPairSmall{\HOLE}{\ocCi[k]}}}{\ocCi}
= \ocVariantSmall{\mpLab[k]}{\ocPairSmall{\ocCi}{\ocCi[k]}}\)
which contains both received payload $\ocCi$ and continuation $\ocCi[k]$,
and binds it to $\ocY$ and continues to $\oeE[2]$,
and the variant value is matched in the subsequent reductions.

Rule \inferrule{\iruleORedMatch} matches
the variant values of the form
$\ocVariantSmall{\ocLab[k]}{\ocPairSmall{\ocC[1]}{\ocC[2]}}$
yielded by \textbf{recv}
against patterns $\ocVariant{\mpLab[i]}{\ocPair{\ocX[i]}{\ocY[i]}}_{i\in I}$,
and if $k \in I$, it
binds $\ocC[1]$ and $\ocC[2]$ to $\ocX[k]$ and $\ocY[k]$ respectively, and reduces to $\oeE[k]$.

The rest of the rules are standard from \cite{scalas19less}.
Rule \inferrule{\iruleORedRec} instantiates a recursive call
to its body $\oeE$;
Rule \inferrule{\iruleORedCong} defines a reduction up to
the structural congruence defined in
\Cref{fig:strcong}. %
Rule \inferrule{\iruleORedCtx} is a contextual rule.

\begin{EX}[Reduction]\label{ex:auth:reduction} \thmstart Recall Examples~\ref{ex:globalcombinatorauth},
  \ref{ex:auth:mio}
and \ref{ex:auth:gen}.
We have\smallskip:\\
\centerline{\scalebox{0.85}{\(
\begin{array}{l}
\oeInit{\ocX,\ocXi}{\gAuth}{\left(\oeE[\roleC] \oePar \oeE[\roleS]\right)}
\rightarrow
\oeRes{\ocS}{\bigl(
  \oeE[\roleC]\subst{\ocX[0]}{\ocC[\roleC]}
  \ \oePar\ %
  \oeE[\roleS]\subst{\ocXi[0]}{\ocC[\roleS]}
  \bigr)} \\[1mm]
=
\oeRes{\ocS}{\bigl(
  \oeSend{\ocC[\roleC]}{\ocX\#\roleS\#\labAuth}{}{\cdots}
  \ \oePar\ %
  \oeBranchSingle{\ocVariantPairPat{\labAuth}{\UNL}{\ocC[\roleS]}}{\ocXi}{\roleC}{\cdots}
  \bigr)}\\[1mm]
\left(
\makecell[l]{
  \text{They interact on $\ocS[3]$, since\ }
  \ocC[\roleC]=\ocIntSumSmall{\roleS}{}{\ocIntChoiceSmall{\labAuth}{\ocS[3]}{\ocCi[\roleC]}}
  \text{\ and\ }
  \ocC[\roleS]=\ocExtSumSmall{\roleC}{}{\ocExtChoiceSmall{\labAuth}{\ocS[3]}{\ocCi[\roleS]}}
  \\[1mm]
  \text{where\ }
  \ocCi[\roleC] =
    \ocExtSumSmall{\roleS}{}{
    \ocExtChoiceSmall{\labOk}{\ocS[1]}{\ocUnit},
    \ocExtChoiceSmall{\labCancel}{\ocS[2]}{\ocUnit}}
  \text{\ and\ }
  \ocCi[\roleS] =
    \ocIntSumSmall{\roleC}{}{
    \ocIntChoiceSmall{\labOk}{\ocS[1]}{\ocUnit},
    \ocIntChoiceSmall{\labCancel}{\ocS[2]}{\ocUnit}}
}\right)\\[3mm]
\rightarrow
\oeRes{\ocS}{\Bigl(
    \oeBranchRaw{\ocCi[\roleC]}{\roleS}{
      \oeChoice{\labOk}{\UNL}{\ocX[2]}{\oeNil}; \oeChoice{\labCancel}{\UNL}{\ocX[3]}{\oeNil}}
  \oePar
    \oeSend{\ocX[2]}{\ocCi[\roleS]\#\roleC\#\labOk}{\text{\oCODE{"ok"}}}{\oeNil}
    \Bigr)}\\[2mm]
\text{(Here, the sender selects $\labOk$, interacting on $\ocS[1]$ and evolving to:)}\\[1mm]
\rightarrow
\oeRes{\ocS}{\Bigl(
    \oeMatchRaw{\ocVariant{\labOk}{\ocPair{\text{\oCODE{"ok"}}}{\ocUnit}}}{
      \oeChoice{\labOk}{\UNL}{\ocX[2]}{\oeNil}; \oeChoice{\labCancel}{\UNL}{\ocX[3]}{\oeNil}}
  \oePar
    \oeNil
    \Bigr)}\rightarrow\oeRes{\ocS}{(\oeNil \oePar \oeNil)} \equiv \oeNil.
\end{array}
\)}}
\end{EX}

\subsection{Static Semantics and Properties of \mustL{}}
\label{sec:typing}
This section summarises the typing systems of 
\mustL{}; then proves type soundness of 
\mustL{}. Typing \mustL{} is divided into three judgements  
(channel vectors, wrappers and expressions) 

\begin{DEF}[Typing rules]\label{def:typingexpr}\rm
  \Cref{fig:typingforchvec} and \Cref{fig:typingforexpression}
  give the typing rules. 
We extend the syntax of typing contexts $\otEnv$ 
from \Cref{def:typingcontext}
as $\otEnv  \grmeq  \ldots \grmor \otEnv \otEnvComp
\otEnvMap{\ocS}{\otT}$  
and introduce
context for recursive functions $\oeEnv$ as:
$\oeEnv  \grmeq  \oeEnvEmpty \grmor \oeEnv \oeEnvComp
\oeEnvMap{\oeRecVar}{\otT[1],\ldots,\otT[n]}$.
Here, $\oeEnvMap{\oeRecVar}{\otT[1],\ldots,\otT[n]}$ states that 
the parameter type of an $n$-ary function $\oeRecVar$.
The typing judgement for 
(1)  channel vectors
  has the form
  $\otEnvEntails{\otEnv}{\ocC}{\otT}$;
(2) wrappers 
  has the form
  $\otEnvEntails{\otEnv}{\ocH}{\otH}$ where 
the type for wrappers is defined as $\otH \grmeq
\otWrapper{\otT}{\otS}$;  and (3) expressions
has a form $\otJudge{\oeEnv}{\otEnv}{\oeE}$. 
We assume that all types in $\otEnv$ and $\oeEnv$ are closed. 
\end{DEF}

\begin{figure}[t]
\scalebox{0.91}{%
\(\begin{array}{c}
   \inference{%
      \inferrule{\iruleOTInit}\ %
    \otRolesSet{\gocaml} = \{\roleP[1],\ldots,\roleP[n]\}
    \quad
    \ogtEnvEntailsEx{\roleP[1],\ldots,\roleP[n]}{}{\gocaml}{\otT[1]{\otTimes}{\cdots}{\otTimes}\otT[n]}
    \quad
    \otJudge{\oeEnv}{\otEnv \otEnvComp \otEnvMap{\ocX[i]}{\otT[i]}}{%
        \oeE[i]%
    }%
    \quad
    \forall i \in \{1{,\elip,}n\}
    }{%
      \otJudge{\oeEnv}{\otEnv}{%
       \oeInit{\ocX[1],\ldots,\ocX[n]}{\gocaml}{(\oeE[1] \oePar \cdots \oePar \oeE[n])}
      }%
  }%
\\[1mm]
    \inference{%
      \inferrule{\iruleOTSel}
      \ %
      \otEnvEntails{\otEnv}{\ocC}{%
        \otIntSumSmall{\roleQ}{}{\otIntChoice{\mpLab}{\otT}{\otTi}}%
      }%
      &%
      \otEnvEntails{\otEnv}{\ocCi}{\otT}%
      &%
      \otJudge{\oeEnv}{%
        \otEnv \otEnvComp \otEnvMap{\ocX}{\otTi}%
      }{%
        \oeE%
      }%
    }{%
      \otJudge{\oeEnv}{%
        \otEnv
      }{%
        \oeSel{\ocX}{\ocC}{\roleQ}{\mpLab}{\ocCi}{\oeE}%
      }%
    }%
    \ %
    \inference{%
      \inferrule{\iruleOTPar}
      \ %
      \otJudge{\oeEnv}{%
        \otEnv%
      }{%
        \oeE[1]%
      }%
      \qquad%
      \otJudge{\oeEnv}{%
        \otEnv%
      }{%
        \oeE[2]%
      }%
    }{%
      \otJudge{\oeEnv}{%
        \otEnv
      }{%
        \oeE[1] \oePar \oeE[2]%
      }%
    }\ %
    \\[1mm]%
    \inference{%
      \inferrule{\iruleOTRecv}\ %
        \otEnvEntails{\otEnv}{\ocC}{%
          \otExtSumSmall{\roleQ}{i \in I}{\otExtChoice{\mpLab[i]}{\otT[i]}{\otTi[i]}}%
        }%
        &%
        \otJudge{\oeEnv}{%
          \otEnv \otEnvComp%
          \otEnvMap{\ocX}{\otVar{\otExtChoice{\mpLab[i]}{\otT[i]}{\otTi[i]}}{i \in I}}%
        }{%
          \oeE%
        }%
    }{%
      \otJudge{\oeEnv}{%
        \otEnv
      }{%
        \oeBranchSingle{\ocX}{\ocC}{\roleQ}{\oeE}%
      }%
    }%
    \ %
    \inference{%
      \inferrule{\iruleOTNil}\qquad
    }{%
      \otJudge{\oeEnv}{\otEnv}{\oeNil}%
    }%
    \\[1mm]%
    \inference{%
      \inferrule{\iruleOTMatch}\ %
        \otEnvEntails{\otEnv}{\ocC}{%
          \otVar{\otExtChoice{\mpLab[i]}{\otT[i]}{\otTi[i]}}{i \in I}%
        }%
        &%
        \otJudge{\oeEnv}{%
          \otEnv \otEnvComp%
          \otEnvMap{\ocY[i]}{\otT[i]} \otEnvComp%
          \otEnvMap{\ocX[i]}{\otTi[i]}%
        }{%
          \oeE[i]%
        }%
        &
        \forall i \!\in\! I%
    }{%
      \otJudge{\oeEnv}{%
        \otEnv
      }{%
        \oeMatch{\ocC}{\mpLab[i]}{\ocY[i]}{\ocX[i]}{\oeE[i]}{i \in I}%
      }%
    }%
    \\[1mm]%
    \inference{%
      \inferrule{\iruleOTLetrec}\ %
        \otJudge{%
          \oeEnv \oeEnvComp%
          \oeEnvMap{\oeRecVar}{\otT[1],\elip,\otT[n]}%
        }{%
          \otEnv \otEnvComp
          \otEnvMap{x_1}{\otT[1]}%
          \otEnvComp\elip\otEnvComp%
          \otEnvMap{x_n}{\otT[n]}%
        }{%
          \oeE[1]%
        }%
        \quad%
        \otJudge{%
          \oeEnv \oeEnvComp%
          \oeEnvMap{\oeRecVar}{\otT[1],\elip,\otT[n]}%
        }{%
          \otEnv%
        }{%
          \oeE[2]%
        }%
    }{%
      \otJudge{\oeEnv}{%
        \otEnv%
      }{%
        \oeLetrec{\oeRecVar}{%
          \otEnvMap{x_1}{\otT[1]},%
          \ldots,%
          \otEnvMap{x_n}{\otT[n]}%
        }{\oeE[1]}{\oeE[2]}%
      }%
    }%
    \\[1mm]%
    \inference{%
      \inferrule{\iruleOTCall}\ %
        \otEnvMap{X}{\otT[1],\elip,\otT[n]} \in \oeEnv%
        &%
        \otEnvEntails{\otEnv}{\ocC[i]}{\otT[i]}%
        &%
        \forall i \in \{1..n\}%
    }{%
      \otJudge{\oeEnv}{%
        \otEnv
      }{%
        \oeCall{\oeRecVar}{\ocC[1],\elip,\ocC[n]}%
      }%
    }%
\quad 
    \ %
    \inference{%
    \inferrule{\iruleOTRes}\quad 
      \otJudge{\oeEnv}{%
        \otEnv \cdot 
          \otEnvMap{\ocS_1}{\otChan{\otT}_1},...,\otEnvMap{\ocS_n}{\otChan{\otT}_n}%
      }{%
        \oeE%
      }%
    }{%
      \otJudge{\oeEnv}{%
        \otEnv%
      }{%
       \oeRes{\ocS}\oeE%
      }%
    }%
  \end{array}\)\vspace*{-2mm}
}
\caption{The Typing Rules for Expressions \framebox{$\protect\otJudge{\oeEnv}{\otEnv}{\oeE}$}\label{fig:typingforexpression}}
\end{figure}

The rules for channel vectors are standard 
where the subtyping relation in rule 
\inferrule{\iruleOTCSub} is defined at \Cref{def:subtyping} in \Cref{sec:typing:global}.

For wrappers, 
rule \inferrule{\iruleOTCWrapInp} types wrapped names
where the payload type $\otSi$ of input channel $\ocS$ is the same as the hole's type,
and all wrappers have the same result type $\otT$.
Rule \inferrule{\iruleOTCWrapper} 
checks type of a channel vector $\ocC=\ocWrapperApp{\ocH}{\ocX}$ and replaces $\ocX$ with the hole $\HOLE$.

For expressions, 
rule \inferrule{\iruleOTInit} types 
the initialisation with a typed global combinator. 
Rule \inferrule{\iruleOTSel}  types the output expression 
which sends a label $\mpLab$ and a payload $\ocCi$ with 
as a nested record at $\ocC$.
Rule \inferrule{\iruleOTRecv} is the dual rule for the input
expression.  Rule \inferrule{\iruleOTRes} hides 
all indexed $s$ by $s$. 
Other rules are standard from \cite{scalas19less}. 

\begin{EX}[Typing expression]\thmstart
Recall that
  {\footnotesize $\eAuth =\oeInit{\ocX,\ocXi}{\gAuth}{\left(\oeE[\roleC] \oePar
      \oeE[\roleS]\right)}$} from \Cref{ex:auth:mio}.
Typing of $\oeE[\roleC]$ has the following derivation:
  \smallskip\\
\centerline{\scalebox{0.75}{\(\begin{array}{c}
    \inference{
      \inference{
        \inference{
          \otJudge{}{
            \otEnvi[\roleC] \otEnvComp
            \otEnvMap{\ocZ}{
              \otVariantMany{
                \otVariantElem{\labOk}{\otPair{\otT}{\otUnit}},
                \otVariantElem{\labCancel}{\otPair{\otT}{\otUnit}}}
            }
            \otEnvComp
            \otEnvMap{\ocX[2]}{\otUnit} \otEnvComp
            \otEnvMap{\UNL}{\otT} 
          }{\oeNil}
          \qquad
          \otJudge{}{
            \otEnvi[\roleC] \otEnvComp
            \otEnvMap{\ocZ}{
              \otVariantMany{
                \otVariantElem{\labOk}{\otPair{\otT}{\otUnit}},
                \otVariantElem{\labCancel}{\otPair{\otT}{\otUnit}}}
            }
            \otEnvComp
            \otEnvMap{\ocX[3]}{\otUnit}  \otEnvComp
            \otEnvMap{\UNL}{\otT}
          }{\oeNil}
        }{
          \otJudge{}{
            \otEnvi[\roleC] \otEnvComp
            \otEnvMap{\ocZ}{
              \otVariantMany{
                \otVariantElem{\labOk}{\otPair{\otT}{\otUnit}},
                \otVariantElem{\labCancel}{\otPair{\otT}{\otUnit}}}
            }
          }{
            \oeMatchRaw{\ocZ}{
              \oeChoice{\labOk}{\UNL}{\ocX[2]}{\oeNil}; \oeChoice{\labCancel}{\UNL}{\ocX[3]}{\oeNil}
            }
          }
        }
      }{
        \otJudge{}{
          \otEnv[\roleC] \otEnvComp
          \otEnvMap{\ocX[1]}{
            \otExtSumSmall{\roleS}{}{
              \otExtChoice{\labOk}{\otT}{\otUnit},
              \otExtChoice{\labCancel}{\otT}{\otUnit}}}
        }{
          \oeBranchSingle{\ocZ}{\ocX[1]}{\roleS}{
            \oeMatchRaw{\ocZ}{
              \oeChoice{\labOk}{\UNL}{\ocX[2]}{\oeNil}; \oeChoice{\labCancel}{\UNL}{\ocX[3]}{\oeNil}
            }
          }
        }
      }
    }{
      \otJudge{}{
        \otEnvMap{\ocX}{
          \otIntSumBig{\roleS}{}{
            \otIntChoice{\labAuth}{\otT}{
              \otExtSumSmall{\roleS}{}{
                \otExtChoice{\labOk}{\otT}{\otUnit},
                \otExtChoice{\labCancel}{\otT}{\otUnit}}}}}
      }{
        \oeSend{\ocX[1]}{\ocX\#\roleS\#\labAuth}{\text{\oCODE{"passwd"}}}{
          \oeEi[\roleC]
        }
      }
    }
    \end{array}\)}}
where
{\scalebox{0.76}{\(
\otEnv[\roleC] = \otEnvMap{\ocX}{
          \otIntSumBig{\roleS}{}{
            \otIntChoice{\labAuth}{\otT}{
              \otExtSumSmall{\roleS}{}{
                \otExtChoice{\labOk}{\otT}{\otUnit},
                \otExtChoice{\labCancel}{\otT}{\otUnit}}}}}
\)}},
{\scalebox{0.76}{\(
\otEnvi[\roleC] = \otEnv[\roleC]\otEnvComp\otEnvMap{\ocX[1]}{
              \otExtSumSmall{\roleS}{}{
                \otExtChoice{\labOk}{\otT}{\otUnit},
                \otExtChoice{\labCancel}{\otT}{\otUnit}}},
\)}}
and\linebreak
{\scalebox{0.76}{\(
\oeEi[\roleC] =
          \oeBranchRaw{\ocX[1]}{\roleS}{
            \oeChoice{\labOk}{\UNL}{\ocX[2]}{\oeNil}; \oeChoice{\labCancel}{\UNL}{\ocX[3]}{\oeNil}}
          \)}}
which is expanded to
\scalebox{0.75}{\(\mathbf{let}\ocZ=\mathbf{recv}\)}
construct.
Similarly, $\oeE[\roleS]$ can be typed as follows:\smallskip\\
\centerline{\scalebox{0.75}{\(
\begin{array}{c}
  \inference{
    \inference{
      \inference{
        \otIntSumSmall{\roleC}{}{
          \otIntChoice{\labOk}{\otT}{\otUnit},
          \otIntChoice{\labCancel}{\otT}{\otUnit}}
        \otSub
        \otIntSumSmall{\roleC}{}{\otIntChoice{\labOk}{\otT}{\otUnit}}
      }{
        \otEnvEntails{
         \otEnvii[\roleS]
        }{\ocXi}{\otIntSumSmall{\roleC}{}{\otIntChoice{\labOk}{\otT}{\otUnit}}}
      }
      \ %
      \inference{
        \otJudge{}{
          \otEnvi[\roleS] \otEnvComp
          \otEnvMap{\ocX[1]}{
            \otIntSumSmall{\roleC}{}{
              \otIntChoice{\labOk}{\otT}{\otUnit},
              \otIntChoice{\labCancel}{\otT}{\otUnit}}}
          \otEnvComp
          \otEnvMap{\ocX[2]}{\otUnit}}{\oeNil}
      }{
        \otJudge{}{
          \otEnvi[\roleS] \otEnvComp
          \otEnvMap{\ocX[1]}{
            \otIntSumSmall{\roleC}{}{
              \otIntChoice{\labOk}{\otT}{\otUnit},
              \otIntChoice{\labCancel}{\otT}{\otUnit}}
          }
        }{
          \oeSel{\ocX[2]}{\ocX[1]}{\roleC}{\labOk}{\text{\oCODE{"ok"}}}{\oeNil}
        }
      }
    }{
      \otJudge{}{\otEnv[\roleS]\otEnvComp
        \otEnvMap{\ocZ}{\otVariantMany{\otVariantElem{\labAuth}{\otPair{\otT}{
            \otIntSumSmall{\roleC}{}{
              \otIntChoice{\labOk}{\otUnit}{\otT},
              \otIntChoice{\labCancel}{\otUnit}{\otT}}
        }}}}
      }{
        \oeMatchRaw{\ocZ}{
          \oeChoice{\labAuth}{\UNL}{\ocX[1]}{\oeSel{\ocX[2]}{\ocX[1]}{\roleC}{\labOk}{\text{\oCODE{"ok"}}}{\oeNil}}
        }
      }
    }
   }{
    \otJudge{}{
      \otEnvMap{\ocXi}{
        \otExtSumBig{\roleC}{}{
          \otExtChoice{\labAuth}{\otT}{
            \otIntSumSmall{\roleC}{}{
              \otIntChoice{\labOk}{\otT}{\otUnit},
              \otIntChoice{\labCancel}{\otT}{\otUnit}}}}}
    }{
      \oeBranchSingle{\ocVariantPairPat{\labAuth}{\UNL}{\ocX[1]}}{\ocXi}{\roleC}{
        \oeSend{\ocX[2]}{\ocX[1]\#\roleC\#\labOk}{\text{\oCODE{"ok"}}}{\oeNil}
      }
    }
  }
\end{array}\)}}
  \smallskip\\
where
{\scalebox{0.76}{\(
    \otEnv[\roleS] =
    \otEnvMap{\ocXi}{
        \otExtSumBig{\roleC}{}{
          \otExtChoice{\labAuth}{\otT}{
            \otIntSumSmall{\roleC}{}{
              \otIntChoice{\labOk}{\otT}{\otUnit},
              \otIntChoice{\labCancel}{\otT}{\otUnit}}}}},
    \otEnvi[\roleS] = \otEnv[\roleS] \otEnvComp
        \otEnvMap{\ocZ}{\otVariantMany{\otVariantElem{\labAuth}{\otPair{\otT}{
            \otIntSumSmall{\roleC}{}{
              \otIntChoice{\labOk}{\otT}{\otUnit},
              \otIntChoice{\labCancel}{\otT}{\otUnit}}}}}}
\)}}

and

{\scalebox{0.76}{\(
    \otEnvii[\roleS] = \otEnvi[\roleS] \otEnvComp
        \otEnvMap{\ocX[1]}{
          \otIntSumSmall{\roleC}{}{
            \otIntChoice{\labOk}{\otT}{\otUnit},
            \otIntChoice{\labCancel}{\otT}{\otUnit}}}
\)}}.

See that the output is typed via subtyping.
Then, we have:\smallskip\\  
\centerline{\scalebox{0.76}{\(\begin{array}{c}
  \inference{
   \ogtEnvEntails{}{\gAuth}{\otT[\roleC]{\otTimes}\otT[\roleS]}
   \quad%
   \otJudge{}{\otEnvMap{\ocX}{\otT[\roleC]}}{
      \oeSel{\ocX}{\ocX}{\roleS}{\labAuth}{\text{\oCODE{"passwd"}}}{\oeEi[\roleC]}
   }
   \quad%
   \otJudge{}{\otEnvMap{\ocXi}{\otT[\roleS]}}{
     \oeBranchSingle{\ocVariantPairPat{\labAuth}{\UNL}{\ocX[1]}}{\ocXi}{\roleC}{\oeEi[\roleS]}
   }
  }{%
    \otJudge{}{}{
      \oeInit{\ocX,\ocXi}{\gAuth}{(\oeE[\roleC] \oePar \oeE[\roleS])}
    }
  }
\end{array}\)}}\smallskip\\
\end{EX}

\begin{restatable}[Subject reduction]{THM}{lemSubjectReduction}\label{lem:SubjectReduction}\thmstart
  If $\otJudge{\oeEnv}{\otEnv}{\oeE}$ and $\oeE \longrightarrow \oeEi$,
  then $\otJudge{\oeEnv}{\otEnv}{\oeEi}$.
\end{restatable}

\section{Proofs for Basic Properties of \OCAMLMPST}
\subsection{Substitution Lemma and other lemmas}

\begin{LEM}[Substitution lemma]\rm
  Followings hold:
  \begin{enumerate}
  \item
    (a) If $\otEnvEntails{\otEnv \otEnvComp \otEnvMap{\ocX}{\otTi}}{\ocC}{\otT}$
    and $\otEnvEntails{\otEnv}{\ocCi}{\otTi}$, then
    $\otEnvEntails{\otEnv}{\ocC \subst{\ocX}{\ocCi}}{\otT}$.
    (b) Moreover, if $\otEnvEntails{\otEnv \otEnvComp \otEnvMap{\ocX}{\otTi}}{\ocH}{\otWrapper{\otT}{\otT[0]}}$
    and $\otEnvEntails{\otEnv}{\ocCi}{\otTi}$, then
    $\otEnvEntails{\otEnv}{\ocH \subst{\ocX}{\ocCi}}{\otWrapper{\otT}{\otT[0]}}$.
  \item
    If $\otEnvEntails{\otEnv}{\ocH}{\otWrapper{\otT}{\otTi}}$
    and $\otEnvEntails{\otEnv}{\ocC}{\otTi}$, then
    $\otEnvEntails{\otEnv}{\ocWrapperApp{\ocH}{\ocC}}{\otT}$.
  \item If $\otJudge{\oeEnv}{\otEnv \otEnvComp \otEnvMap{\ocX}{\otT}}{\oeE}$
    and $\otEnvEntails{\otEnv}{\ocC}{\otT}$, then
    $\otJudge{\oeEnv}{\otEnv}{\oeE\subst{\ocX}{\ocC}}$.
  \end{enumerate}
\end{LEM}
\begin{proof}
Follows.
\begin{enumerate}
\item We proceed by mutual induction on the derivation trees of
  $\otEnvEntails{\otEnv}{\ocC}{\otT}$ and $\otEnvEntails{\otEnv}{\ocH}{\otWrapper{\otT}{\otTi}}$.
  We start from (a).\\
\CASE \inferrule{\iruleOTCUnit}. $\ocC = \ocUnit$. Trivial.\\
\CASE \inferrule{\iruleOTCVar}. $\ocC = \ocY$.
  If $\ocX = \ocY$,
  we have $\ocY \subst{\ocX}{\ocCi} = \ocCi$,
  and by rule \inferrule{\iruleOTCVar},
  we have $\otT = \otTi$.
  By assumption, we get $\otEnvEntails{\otEnv}{\ocCi}{\otT}$.
  If $\ocX \neq \ocY$, since $\ocY \subst{\ocX}{\ocCi} = \ocY$, it trivially holds.\\
\CASE \inferrule{\iruleOTCName}. $\ocC = \ocS$.
  We have $\ocS \subst{\ocX}{\ocCi} = \ocS$ and it trivially holds.\\
\CASE \inferrule{\iruleOTCTuple}. $\ocC = (\ocC[1],\elip,\ocC[n])$.
  For all $i \in \{1,\elip,n\}$, exists $\otT[i]$ such that $T = \otT[i]{\otTimes}\elip{\otTimes}\otT[n]$ and
  $\otEnvEntails{\otEnv \otEnvComp \otEnvMap{\ocX}{\otTi}}{(\ocC[1],\elip,\ocC[n])}{\otT[1]{\otTimes}{\elip}{\otTimes}\otT[n]}$,
  and we have  $\otEnvEntails{\otEnv \otEnvComp \otEnvMap{\ocX}{\otTi}}{\ocC[i]}{\otT[i]}$ for $i \in \{1,\elip,n\}$.
  By induction hypothesis, we have
    $\otEnvEntails{\otEnv}{\ocC[i]\subst{\ocX}{\ocCi}}{\otT[i]} \quad (i \in \{1,\elip,n\})$.
  Then, by applying \inferrule{\iruleOTCTuple}, we get
    $\otEnvEntails{\otEnv}{(\ocC[1],\elip,\ocC[n])\subst{\ocX}{\ocCi}}{\otT[1]{\otTimes}{\elip}{\otTimes}\otT[2]}$.\\
\CASE \inferrule{\iruleOTCRecord} and \inferrule{\iruleOTCVariant}.
  $\ocC = \ocRecord{\ocL[i]}{\ocC[i]}{i \in I}$ and $\ocC = \ocVariant{\ocL}{\ocCi}{}$. Similar.\\
\CASE \inferrule{\iruleOTCWrapInp}. $\ocC = \ocWrapper{\ocS[i]}{\ocH[i]}{i \in I}$.
  From rule \inferrule{\iruleOTCWrapInp}, $T = \otInp{\otTii}$ for some $\otTii$,
    and for each $i \in I$, there exists $\otT[i]$ such that
    $\otEnvEntails{\otEnv \otEnvComp \otEnvMap{\ocX}{\otTi}}{\ocS[i]}{\otT[i]}$ and
    $\otEnvEntails{\otEnv \otEnvComp \otEnvMap{\ocX}{\otTi}}{\ocH[i]}{\otWrapper{\otTii}{\otT[i]}}$.
  By induction hypothesis, we have
    $\otEnvEntails{\otEnv}{\ocS[i]\subst{\ocX}{\ocCi}}{\otT[i]}$ and
    $\otEnvEntails{\otEnv}{\ocH[i]\subst{\ocX}{\ocCi}}{\otWrapper{\otTii}{\otT[i]}}$ for each $i \in I$, and
  by applying \inferrule{\iruleOTCWrapInp}, it follows $\otEnvEntails{\otEnv}{\ocWrapper{\ocS[i]}{\ocH[i]}{i \in I}\subst{\ocX}{\ocCi}}{\otInp{\otTii}}$.\\
\CASE \inferrule{\iruleOTCSub}.
  We have $\otS$ such that $\otS \otSub \otT$ and $\otEnvEntails{\otEnv \otEnvComp \otEnvMap{\ocX}{\otTi}}{\ocC}{\otS}$.
  By induction hypothesis, $\otEnvEntails{\otEnv}{\ocC\subst{\ocX}{\ocCi}}{\otS}$.
  Again, by applying \inferrule{\iruleOTCSub}, we get $\otEnvEntails{\otEnv}{\ocC\subst{\ocX}{\ocCi}}{\otT}$.\\
  {\bfseries For} (b), we have $\otEnvEntails{\otEnv}{\ocH}{\otWrapper{\otT}{\otT[0]}}$ and the only rule is \inferrule{\iruleOTCWrapper}.
  By the rule, we have $\ocC$, $\ocY$ such that
    $\ocY \notin \fn{\ocH}$,
    $\ocC=\ocWrapperApp{\ocH}{\ocY}$ and
    $\otEnvEntails{\otEnv \otEnvComp \otEnvMap{\ocY}{\otT[0]}}{\ocC}{\otT}$.
  Note that, by Barendregt convention, we can assume $\ocX \neq \ocY$ and $\ocY \notin \fn{\ocCi}$.
  By induction hypothesis,
    $\otEnvEntails{\otEnv \otEnvComp \otEnvMap{\ocY}{\otT[0]}}{\ocC\subst{\ocX}{\ocCi}}{\otT}$,
  Furthermore, we see $\ocC\subst{\ocX}{\ocCi}=\ocWrapperApp{\ocH\subst{\ocX}{\ocCi}}{\ocY}$
  and, $\ocY \notin \fn{\ocH\subst{\ocX}{\ocCi}}$ (since $\ocY \notin \fn{\ocH}$).
  By \inferrule{\iruleOTCWrapper},
  $\otEnvEntails{\otEnv}{\ocH\subst{\ocX}{\ocCi}}{\otWrapper{\otT}{\otT[0]}}$.

\item From the derivation of $\otEnvEntails{\otEnv}{\ocH}{\otWrapper{\otT}{\otTi}}$, for some $\ocX$ and $\ocCi$ we have
  $\ocCi=\ocWrapperApp{\ocH}{\ocX}$ such that $\otEnvEntails{\otEnv \otEnvComp \otEnvMap{\ocX}{\otTi}}{\ocCi}{\otT}$.
  By (1), we have $\otEnvEntails{\otEnv}{\ocCi\subst{\ocX}{\ocC}}{\otT}$
  and since $\ocWrapperApp{\ocH}{\ocC}=\ocCi\subst{\ocX}{\ocC}$,
  we get $\otEnvEntails{\otEnv}{\ocWrapperApp{\ocH}{\ocC}}{\otT}$.

\item By induction on the derivation of $\otJudge{\oeEnv}{\otEnv}{\oeE}$.\\
\CASE \inferrule{\iruleOTNil} Trivial.\\
\CASE \inferrule{\iruleOTLetrec}
  We have
  $\oeE = \oeLetrec{\oeRecVar}{\otEnvMap{\ocX[1]}{\otT[1]},\ldots,\otEnvMap{\ocX[n]}{\otT[n]}}{\oeE[1]}{\oeE[2]}$
  and we assume $\ocX \notin \setenum{\ocX[i]}_{i \in 1 \elip n}$.
  By induction hypothesis,
  $\otJudge
    {\oeEnv \oeEnvComp \oeEnvMap{\oeRecVar}{\otT[1],\ldots,\otT[n]}}
    {\otEnv \otEnvComp \otEnvMap{\ocX[1]}{\otT[1]} \otEnvComp \ldots \otEnvMap{\ocX[n]}{\otT[n]}}
    {\oeE[1]\subst{\ocX}{\ocC}}$ and
  $\otJudge
    {\oeEnv \oeEnvComp \oeEnvMap{\oeRecVar}{\otT[1],\ldots,\otT[n]}}
    {\otEnv}
    {\oeE[2]\subst{\ocX}{\ocC}}$,
  and by applying \inferrule{\iruleOTLetrec}, we get\linebreak
  $\otJudge{\oeEnv}{\otEnv}
  {\bigl(\oeLetrec
    {\oeRecVar}
    {\otEnvMap{\ocX[1]}{\otT[1]},\ldots,\otEnvMap{\ocX[n]}{\otT[n]}}
    {\oeE[1]}
    {\oeE[2]}
    \bigr)\subst{\ocX}{\ocC}}$.
  \\
\CASE \inferrule{\iruleOTX}.
  $\oeE = \oeCall{\oeRecVar}{\ocC[1],\ldots,\ocC[n]}$ and
    $\otEnvEntails{\otEnv}{\ocC[i]}{\otTi}$ for $i \in \{1, \ldots, n\}$.
  By (1), we have $\otEnvEntails{\otEnv}{\ocC[i]\subst{\ocX}{\ocC}}{\otTi}$ and
  By \inferrule{\iruleOTX}, it follows that
    $\otJudge{\oeEnv}{\otEnv}{\bigl(\oeCall{\oeRecVar}{\ocC[1],\ldots,\ocC[n]}\bigr)\subst{\ocX}{\ocC}}$.\\
\CASE \inferrule{\iruleOTRecv}. We have
    $\otJudge{\oeEnv}{\otEnv}{\oeBranchSingle{\ocY}{\ocCi}{\roleQ}{\oeE}}$.
  By assumption and (1), we have
    $\otEnvEntails{\otEnv}{\ocCi\subst{\ocX}{\ocC}}{%
          \otExtSum{\roleQ}{i \in I}{\otExtChoice{\mpLab[i]}{\otT[i]}{\otTi[i]}}%
    }$, and by assumption and induction hypothesis, we have
    $\otJudge{\oeEnv}{%
          \otEnv \otEnvComp%
          \otEnvMap{\ocY}{\otVar{\otExtChoice{\mpLab[i]}{\otT[i]}{\otTi[i]}}{i \in I}}%
        }{%
          \oeE\subst{\ocX}{\ocC}%
    }$.
  By applying \inferrule{\iruleOTRecv}, we get
    $\otJudge{\oeEnv}{\otEnv}{\bigl(\oeBranchSingle{\ocY}{\ocCi}{\roleQ}{\oeE}\bigr)\subst{\ocX}{\ocC}}$.\\
\CASE \inferrule{\iruleOTMatch}. We have
  $e = \oeMatch{\ocCi}{\mpLab[i]}{\ocX[i]}{\ocY[i]}{\oeE[i]}{i \in I}$.
  By assumption and (1), we have
  $\otEnvEntails{\otEnv}{\ocCi\subst{\ocX}{\ocC}}{\otVar{\otExtChoice{\mpLab[i]}{\otT[i]}{\otTi[i]}}{i \in I}}$.
  Furthermore, by assumption and induction hypothesis, for each $i \in I$, we have
        $\otJudge{\oeEnv}{%
          \otEnv \otEnvComp%
          \otEnvMap{\ocY[i]}{\otT[i]} \otEnvComp%
          \otEnvMap{\ocX[i]}{\otTi[i]}%
        }{\oeE[i]\subst{\ocX}{\ocC}}$.
  By applying \inferrule{\iruleOTMatch}, we get\\
    $\otJudge{\oeEnv}{\otEnv}
      {\bigl(\oeMatch{\ocCi}{\mpLab[i]}{\ocX[i]}{\ocY[i]}{\oeE[i]}{i \in I}\bigr)\subst{\ocX}{\ocC}}$.\\
\CASE \inferrule{\iruleOTSel}.
    $\oeE = \oeSel{\ocX}{\ocC[0]}{\roleQ}{\mpLab}{\ocC[1]}{\oeE}$.
  By assumption and (1), 
    $\otEnvEntails{\otEnv}{\ocC[0]\subst{\ocX}{\ocC}}{
        \otIntSum{\roleQ}{}{\otIntChoice{\mpLab}{\otT}{\otTi}}}$ and
    $\otEnvEntails{\otEnv}{\ocC[1]\subst{\ocX}{\ocC}}{\otT}$ hold.
  By assumption and induction hypothesis,
    $\otJudge{\oeEnv}{\otEnv \otEnvComp \otEnvMap{\ocY}{\otTi}}{\oeE\subst{\ocX}{\ocC}}$.
  By applying \inferrule{\iruleOTSel}, we get
    $\otJudge{\oeEnv}{\otEnv}{
        \bigl(\oeSel{\ocX}{\ocC[0]}{\roleQ}{\mpLab}{\ocC[1]}{\oeE}\bigr)\subst{\ocX}{\ocC}}$.
  \\
\CASE \inferrule{\iruleOTPar}. $\oeE = \oeE[1] \oePar \oeE[2]$.
  By induction hypothesis, we get $\otJudge{\oeEnv}{\otEnv}{\oeE[i]\subst{\ocX}{\ocC}}$ for $i \in \{1,2\}$.
  By applying \inferrule{\iruleOTPar}, we get
    $\otJudge{\oeEnv}{\otEnv}{\bigl(\oeE[1] \oePar \oeE[2]\bigr)\subst{\ocX}{\ocC}}$.\\
\CASE \inferrule{\iruleOTInit}. We have
  $\oeE=\oeInit{\ocX[1],\ldots,\ocX[n]}{\gocaml}{(\oeE[1] \oePar \cdots \oePar \oeE[n])}$.
  By induction hypothesis, we get\\
    $\otJudge{\oeEnv}{\otEnv \otEnvComp \otEnvMap{\ocX[i]}{\otT[i]}}{%
        \oeE[i]\subst{\ocX}{\ocC}}$
  for each
    $i \in \{1,\ldots,n\}$
  (note that $\ocX \notin \setenum{\ocX[i]}_{i \in \{1,\ldots,n\}}$).\\
  By applying \inferrule{\iruleOTInit}, we get
      $\otJudge{\oeEnv}{\otEnv}{
    \oeInit{\ocX[1],\ldots,\ocX[n]}{\gocaml}{(\oeE[1] \oePar \cdots \oePar \oeE[n])\subst{\ocX}{\ocC}}}
  $.
  \\
\CASE \inferrule{\iruleOTRes}.
  We assume $\ocS\notin\fn{c}$.
  By induction hypothesis, 
  $\otJudge{\oeEnv}{\otEnv \otEnvComp \otEnvMap{\ocS}{\otChan{\otT}}}{\oeE\subst{\ocX}{\ocC}}$.
  By applying \inferrule{\iruleOTRes}, we get
  $\otJudge{\oeEnv}{\oeEnv}{\bigl(\oeRes{\otEnvMap{\ocS}{\otChan{\otT}}}\oeE\bigr)\subst{\ocX}{\ocC}}$.
\end{enumerate}
\end{proof}

\begin{LEM}[Inversion]\rm
  Followings hold:
  \begin{enumerate}
  \item If $\otJudge{\oeEnv}{\otEnv}{\oeSend{\ocX}{\ocD}{\ocCi}{\oeE}}$
    and
    $\ocD = \ocPair{\ocS[j]}{\ocC[j]}$
    then,
    $\otEnvEntails{\otEnv}{\ocC[j]}{\otT}$ and
    $\otJudge{\oeEnv}{\otEnv \otEnvComp \otEnvMap{\ocX}{\otT}}{\oeE}$,
    $\otEnv = \otEnvi \otEnvComp \otEnvMap{\ocS[j]}{\otChan{\otTi}}$
    where $j \in I$, and
    $\otEnvEntails{\otEnv}{\ocCi}{\otTii}$ where $\otTii \otSub \otTi$.
  \item If $\otJudge{\oeEnv}{\otEnv}{\oeRecv{\ocX}{\ocD}{\oeE}}$
    and $\ocD = \ocWrapper{\ocS[i]}{\ocH[i]}{i \in I}$
    then,
    $\otEnv = \otEnvi \otEnvComp \{\otEnvMap{\ocS[i]}{\otChan{\otS[i]}}\}_{i \in I}$,
    $\otEnvEntails{\otEnv}{\ocH[i]}{\otWrapper{\otT}{\otT[i]}}$ and $\otS[i] \otSub \otT[i]$ for all $i \in I$, and
    $\otJudge{\oeEnv}{\otEnv \otEnvComp \otEnvMap{\ocX}{\otT}}{\oeE}$.
  \item If $\otJudge{\oeEnv}{\otEnv}{\oeMatch{\ocC}{\mpLab{i}}{\ocX[i]}{\ocY[i]}{\oeE[i]}{i \in I}}$,
    $\ocC=\ocVariant{\mpLab[j]}{\ocPair{\ocC[j]}{\ocCi[j]}}{}$ and $j \in I$,
    then
    for all $i \in I$, $\otJudge{\oeEnv}{\otEnv \otEnvComp \otEnvMap{\ocX[i]}{\otT[i]} \otEnvComp \otEnvMap{\ocY[i]}{\otTi[i]}}{\oeE[i]}$, 
    $\otEnvEntails{\otEnv}{\ocC[j]}{\otS[j]}$, 
    $\otEnvEntails{\otEnv}{\ocCi[j]}{\otSi[j]}$, 
    $\otS[j] \otSub \otT[j]$ and
    $\otSi[j] \otSub \otTi[j]$.
  \item If $\otJudge{\oeEnv}{\otEnv}{\oeLetrec{\oeRecVar}{\tilde{\ocX}}{\oeE[1]}{\oeE[2]}}$, then
    $\otJudge{\oeEnv \oeEnvComp \oeEnvMap{\oeRecVar}{\otT[1],\ldots,\otT[n]}} {\otEnv \otEnvComp \otEnvMap{x_1}{\otT[1]} \otEnvComp \ldots \otEnvComp\otEnvMap{x_n}{\otT[n]}} {\oeE[1]}$, and
    $\otJudge{\oeEnv \oeEnvComp \oeEnvMap{\oeRecVar}{\otT[1],\ldots,\otT[n]}}{\otEnv}{\oeE[2]}$.
  \item If $\otJudge{\oeEnv}{\otEnv}{\oeE[1] \oePar \oeE[2]}$, then
    $\otJudge{\oeEnv}{\otEnv}{\oeE[1]}$ and $\otJudge{\oeEnv}{\otEnv}{\oeE[2]}$.
  \item If $\otJudge{\oeEnv}{\otEnv}{\oeCall{\oeRecVar}{\ocC[1],\ldots,\ocC[n]}}$, then
    $\oeEnv = \oeEnvi \oeEnvComp \oeEnvMap{X}{\otT[1],\ldots,\otT[n]}$, 
    $\forall i \in 1..n$, $\otEnvEntails{\otEnv}{\ocC[i]}{\otS[i]}$ and
    $\otS[i] \otSub \otT[i]$.
  \item
    If
    $\otJudge{\oeEnv}{\otEnv}{\oeInit{\ocX[1],\ldots,\ocX[n]}{\gocaml}{(\oeE[1] \oePar \cdots \oePar \oeE[n])}}$,
    then
    $\otRolesSet{\gocaml} = \{\roleP[1],\ldots,\roleP[n]\}$,
    $\ogtEnvEntails{}{\gocaml}{\otT[1]\otTimes\cdots\otTimes\otT[n]}$ and
    $\otJudge{\oeEnv}{\otEnv \otEnvComp \otEnvMap{\ocX[i]}{\otT[i]}}{\oeE[i]}$.
  \item
    If $\otJudge{\oeEnv}{\otEnv}{\oeRes{\otEnvMap{\ocS}{\otChan{\otT}}}\oeE}$
    then
    $\otJudge{\oeEnv}{\otEnv \otEnvComp \otEnvMap{\ocS}{\otChan{\otT}}}{\oeE}$.
  \end{enumerate}
\end{LEM}
\begin{proof}
  Standard.
\end{proof}

\begin{LEM}[Type preservation for $\equiv$]\rm
  If $\otJudge{\oeEnv}{\otEnv}{\oeE}$ and $\oeE \equiv \oeEi$,
  then $\otJudge{\oeEnv}{\otEnv}{\oeEi}$.
\end{LEM}
\begin{proof}
  Standard.
\end{proof}

The following lemma relates term-level and type-level projection.

\subsection{Type safety for global combinators}

\begin{DEF}\rm
  $\otEnv$ is {\em basic on $\ocC$},
  written $\BASIC{\otEnv}{\ocC}$,
  if, for all $\ocX \in \fv{\ocC}$ there is some
  $\otRecVar[\ocX]$ such that
  $\otEnvEntails{\otEnv}{\ocX}{\otRecVar[\ocX]}$,
  and $\otRecVar[\ocX]\neq\otRecVar[\ocY]$ for any $\ocX, \ocY \in \fv{\ocC}$ s.t. $\ocX\neq\ocY$, .
\end{DEF}

\begin{LEM}\rm
  If $\otEnvEntails{\otEnv}{\ocC[i]}{\otT}$ $(i \in \{1,2\})$ and followings hold:
  \begin{enumerate}
    \item $\BASIC{\otEnv}{\ocC[i]}$ for  $i \in \{1,2\}$.
    \item If $\mapsubst{\ocZ}{(\ocC[1], \ocC[2])} \in \ocChi$, then
      $\otEnvEntails{\otEnv}{\ocC[i]}{\otT}$ $(i \in \{1,2\})$ and $\otEnvEntails{\otEnv}{\ocZ}{\otT}$.
  \end{enumerate}
  Then,
  $\ocC[1] \ocBinMerge[\ocChi] \ocC[2]$ is defined and
  $\otEnvEntails{\otEnv}{\ocC[1] \ocBinMerge[\ocChi] \ocC[2]}{\otT}$ holds.
\end{LEM}
\begin{proof}
  We proceed by the induction on the number of calls of $\ocC[1] \ocBinMerge[\ocChi] \ocC[2]$.
  This induction terminates since the size of the set of pairs $(\ocC[1], \ocC[2])$ accumulated in $\ocChi$ is bounded.
  The interesting cases are ones that involve recursion.
  
  \CASE{\ $\ocC[1] = \ocRec{\ocX}{\ocCi[1]}$}.
  (1) If $\mapsubst{\ocZ}{(\ocRec{\ocX}{\ocCi[1]}, \ocC[2])} \in \ocChi$,
  by the definition, we have $\ocC[1]\ocBinMerge[\ocChi]\ocC[2]=\ocZ$.
  Furthermore, by assumption, we have $\otEnvEntails{\otEnv}{\ocZ}{\otT}$.
  (2) If $\mapsubst{\ocZ}{(\ocRec{\ocX}{\ocCi[1]}, \ocC[2])} \notin \ocChi$,
  by inversion lemma, we have $\otEnvEntails{\otEnv}{\ocRec{\ocX}{\ocCi[1]}}{\otRec{\otRecVar}{\otTi}}$
  for some $\otRecVar, \otTi$ where $\otRec{\otRecVar}{\otTi} \otSub \otT$,
  and by substitution lemma,
  we have $\otEnvEntails{\otEnv}{\ocCi[1]\subst{\ocX}{\ocRec{\ocX}{\ocCi[1]}}}{\otTi\subst{\otRecVar}{\otRec{\otRecVar}{\otTi}}}$.
  Furthermore, since $\otTi\subst{\otRecVar}{\otRec{\otRecVar}{\otTi}} \otSub \otRec{\otRecVar}{\otTi} \otSub \otT$,
  we have $\otEnvEntails{\otEnv}{\ocCi[1]\subst{\ocX}{\ocRec{\ocX}{\ocCi[1]}}}{\otT}$.
  By induction hypothesis,
  we have some
  $\ocCi = \ocCi[1]\subst{\ocX}{\ocRec{\ocX}{\ocCi[1]}} \ocBinMerge[\ocChi \cdot \mapsubst{\ocZ}{(\ocRec{\ocX}{\ocCi[1]}, \ocC[2])}] \ocC[2]$
  defined, and $\otEnvEntails{\otEnv \otEnvComp \otEnvMap{\ocZ}{\otT}}{\ocCi}{\otT}$.
  (Here, beware that both
    $\BASIC{\otEnv \otEnvComp \otEnvMap{\ocZ}{\otT}}{\ocCi[1]\subst{\ocX}{\ocRec{\ocX}{\ocCi[1]}}}$
    and
    $\BASIC{\otEnv \otEnvComp \otEnvMap{\ocZ}{\otT}}{\ocC[2]}$
  hold since $\ocZ$ is fresh, i.e. $\ocZ \notin \left(\fv{\ocCi[1]}\cup\fv{\ocC[2]}\right)$).
  Then, by \inferrule{\iruleOTCRec}, we have
  $\otEnvEntails{\otEnv}{\ocRec{\ocZ}{\ocCi}}{\otT}$.

  \CASE{\ $\ocC[1] = \ocX$}.
  Since $\BASIC{\otEnv}{\ocX}$, we have $\otEnvEntails{\otEnv}{\ocX}{\otRecVar[\ocX]}$, and
  by inversion lemma, $\otT=\otRecVar[\ocX]$.
  Furthermore, since only possible rule to derive $\otEnvEntails{\otEnv}{\ocC[2]}{\otRecVar[\ocX]}$ is
  \inferrule{\iruleOTCVar}, and from $\BASIC{\otEnv}{\ocC[2]}$, we have $\ocC[2] = \ocX$.
  Hence, by the definition of $\ocBinMerge[\ocChi]$, we have $\ocC[1] \ocBinMerge[\ocChi] \ocC[2] = \ocX$.
\end{proof}

\begin{LEM}\rm
  If $\otEnvEntails{\otEnv}{\ocC[i]}{\otT}$ and $\BASIC{\otEnv}{\ocC[i]}$ for all $i \in I$,
  then $\ocMerge{i \in I}{\ocC[i]}$ is defined
  and $\otEnvEntails{\otEnv}{\ocMerge{i \in I}{\ocC[i]}}{\otT}$.
\end{LEM}
\begin{proof}
Straightforward by induction. 
\end{proof}
  
\begin{PROP}\thmstart
If $\ogtEnvEntails{}{\ogtG}{\otT}$, then $\otT$ is closed.
\end{PROP}
\begin{proof}
By induction on $\ogtG$.
\end{proof}

\begin{restatable}[]{LEM}{lemSubjectReductionForGC}\label{lem:SubjectReductionForGC}\thmstart
  If $\ogtEnvEntailsEx{\roleP[1],\ldots,\roleP[n]}{\otEnv}{\ogtG}{\otT[1]\otTimes\cdots\otTimes\otT[n]}$
  then ${\GCCV{\gocaml}} = \ocC$ is defined and
  $\otEnvEntails{\otEnvi}{\ocC}{\otT[1]\otTimes\cdots\otTimes\otT[n]}$
  where
  $\otEnvi = \otEnv \otEnvComp \{\otEnvMap{\ocS[i]}{\otS[i]}\}_{\ocS[i]\in \fn{\ocC}}$
  for some $\{\widetilde{\otS[i]}\}$.
\end{restatable}
\begin{proof}
  We proceed by induction on the structure of $\gocaml$.\\
\CASE{\ $\gocaml=\ogtComm{\roleP[j]}{\roleP[k]}{}{\mpLab}{\otT}{\gocaml}$}.
By inversion, $\ogtEnvEntails{\otEnv}{\gocaml}{\otT[1]\otTimes\cdots\otTimes\otT[n]}$ holds.
By induction hypothesis, we get
$\otEnvEntails{\otEnvi}{\GCCV{\gocaml}}{\otT[1]\otTimes\cdots\otTimes\otT[n]}$
where $\otEnvi = \otEnv\otEnvComp\{\widetilde{\ocS[i]{:}\otS[i]}\}$ for some
$\{\widetilde{\ocS[i]{:}\otS[i]}\}$.
Let
$\otEnvii=\otEnvi\otEnvComp\otEnvMap{\ocS}{\otT}$
where $\ocS=\ocS_{\{\roleP[j],\roleP[k],\mpLab,\blueI\}}$.
For each $\roleP[i] \in \{\roleP[1],\elip,\roleP[n]\}$,
we have $\otEnvEntails{\otEnv}{\NTH{\GCCV{\gocaml}}{i}}{\otT[i]}$.
and see that by \inferrule{\iruleOTCName},
$\otEnvEntails{\otEnvii}{\ocS}{\otT}$.
Then,
by applying typing rules repeatedly, we have:
$\otEnvEntails{\otEnvii}{
  \ocIntSum{\roleP[k]}{}{
    \ocIntChoiceSmall{\mpLab}{\ocS}{\NTH{\GCCV{\ogtG}}{j}}}
}{\otTi[j]
}$
and
$
\otEnvEntails{\otEnvii}
{
  \ocExtSum{\roleP[j]}{}{
    \ocExtChoiceSmall{\mpLab}{\ocS}{\NTH{\GCCV{\ogtG}}{k}}}
}
{\otTi[k]
}$
where
$\otTi[j]=
\otIntSum{\roleP[k]}{}{
    \otIntChoice{\mpLab}{\otT}{\otT[j]},
  }
$
and
$\otTi[k]=
  \otExtSum{\roleP[j]}{}{
    \otExtChoice{\mpLab}{\otT}{\otT[k]}
  }
$.
Then, by using \inferrule{\iruleOTCTuple}, we have\\
$\otEnvEntails{\otEnvii}{\gocaml}{
  \otT[1]{\otTimes}
  \elip{\otTimes}\otTi[j]{\otTimes}\elip
  \elip{\otTimes}\otTi[k]{\otTimes}\elip\otT[n]}$.

\CASE{\ $\gocaml=\ogtChoice{\roleP[a]}{\gocaml[i]}{i \in I}$}.
By inversion, for all $i \in I$ we have
$\otEnvEntails{\otEnv}{\gocaml[i]}{\otT[1]\otTimes\elip\otTimes\otT[n]}$.
Then, applying induction hypothesis, we have
$\otEnvEntails{\otEnvi[i]}{\GCCV{\gocaml[i]}}{\otT[1]\otTimes\elip
\otIntSum{\roleP[a]}{k\in K_{i}}{
    \otIntChoice{\mpLab[k]}{\otT[k]}{\otT[k]}
  }
  \elip\otTimes\otT[n]}$
where $\otEnvi[i] = \otEnv\otEnvComp\{\widetilde{\ocS[ij]{:}\otS[ij]}\}$ for some
$\{\widetilde{\ocS[ij]{:}\otS[ij]}\}$.
By taking $\otEnvi = \otEnv, \bigcup_{i\in I}\{\widetilde{\ocS[ij]{:}\otS[ij]}\}$
and $\ocC[ij] = \NTH{\GCCV{\gocaml[i]}}{j}$,
we have $\otEnvEntails{\otEnvi}{\ocC[ij]}{\otT[j]}$
and $\otEnvEntails{\otEnvi}{\ocMerge{i\in I}{\ocC[ij]}}{\otT[j]}$
for each $j \in \{1,\elip,n\}\setminus\{a\}$, and
$\ocC[ia] = \ocIntSum{\roleP[a]}{k\in K_{i}}{
    \ocIntChoice{\mpLab[k]}{\ocS[ik]}{\ocCi[ik]}
  }
$.
Then, by applying \inferrule{\iruleOTCRecord} for
$\ocIntSum{\roleP[a]}{k\in K}{
    \ocIntChoice{\mpLab[k]}{\ocS[ik]}{\ocCi[ik]}
  }\ (K=\bigcup_{i\in I}K_i)$ and by using
\inferrule{\iruleOTCTuple}, we have the desired typing.
Other cases are trivial or similar.
\end{proof}

\colSubjectReductionForGC*
\begin{proof}
  A special case of the above lemma. 
\end{proof}

\subsection{Proof of Subject Reduction}

\lemSubjectReduction*
\begin{proof}
Induction on derivation of $\oeE \longrightarrow \oeEi$.\\
\CASE \inferrule{\iruleORedComm}.
$\oeE = \oeSend{\ocX}{\ocC[\roleP]\#\roleQ\#\mpLab[k]}{\ocCi}{\oeE[1]} \oePar \oeRecv{\ocY}{\ocC[\roleQ]\#\roleP}{\oeE[2]}$,
$\oeEi = \oeE[1]\subst{\ocX}{\ocC} \oePar \oeE[2]\subst{\ocY}{\ocWrapperApp{\ocH[j]}{\ocCi}}$
where $j \in I$,
$\ocC[\roleP]\#\roleQ\#\mpLab[k] = \ocPair{\ocS[j]}{\ocC}$ and $\ocC[\roleQ]\#\roleP = \ocWrapper{\ocS[i]}{\ocH[i]}{i\in I}$.
By applying inversion lemma for $\oePar$, {\bf send} and {\bf recv}, we have
\begin{itemize}
\item $\otJudge{\oeEnv}{\otEnvi \otEnvComp \otEnvMap{\ocS[j]}{\otChan{\otTi}}}{\oeSend{\ocX}{\ocD[1]}{\ocCi}{\oeE[1]}}$,
\item $\otJudge{\oeEnv}{\otEnvi \otEnvComp \{\otEnvMap{\ocS[i]}{\otChan{\otTi[i]}}\}_{i \in I}}{\oeRecv{\ocY}{\ocD[2]}{\oeE[2]}}$,
  and $\otTi[j] = \otTi$
\item $\otEnvEntails{\otEnv}{\ocC}{\otT}$ and $\otJudge{\oeEnv}{\otEnv \otEnvComp \otEnvMap{\ocX}{\otT}}{\oeE[1]}$,
\item $\otEnvEntails{\otEnv}{\ocCi}{\otTii}$ and $\otTii \otSub \otTi$,
\item For all $i \in I$, $\otEnvEntails{\otEnv}{\ocH[i]}{\otWrapper{\otTiii}{\otT[i]}}$ and $\otJudge{\oeEnv}{\otEnv \otEnvComp \otEnvMap{\ocY}{\otTiii}}{\oeE[2]}$.
\end{itemize}
By applying substitution lemma on $\oeE[1]$, we get $\otJudge{\oeEnv}{\otEnv}{\oeE[1]\subst{\ocX}{\ocC}}$.
Next, by applying \inferrule{\iruleOTCSub} to $\otEnvEntails{\otEnv}{\ocCi}{\otTii}$, we have $\otEnvEntails{\otEnv}{\ocCi}{\otTi = \otTi[j]}$ and 
By applying substitution lemma on $\ocH[j]$, we have $\otEnvEntails{\otEnv}{\ocWrapperApp{\ocH[j]}{\ocCi}}{\otTiii}$
and by substitution lemma on $\oeE[2]$, we get  $\otJudge{\oeEnv}{\otEnv}{\oeE[2]\subst{\ocY}{\ocWrapperApp{\ocH[j]}{\ocCi}}}$.
Then, from \inferrule{\iruleOTPar} we get $\otJudge{\oeEnv}{\otEnv}{\oeE[1]\subst{\ocX}{\ocC} \oePar \oeE[2]\subst{\ocY}{\ocWrapperApp{\ocH[j]}{\ocCi}}}$.\\
\CASE \inferrule{\iruleORedMatch}.
$\oeE = \oeMatch{c}{\mpLab[i]}{\ocX[i]}{\ocY[i]}{\oeE[i]}{i\in I}$,
$\oeEi = \oeE[j]\subst{\ocX[j]}{\ocC[1]}\subst{\ocY[j]}{\ocC[2]}$, and
$\ocC = \ocVariant{\mpLab[j]}{\ocPair{\ocC[1]}{\ocC[2]}}{}$
where $j \in I$.
By inversion lemma for $\oeFmt{match}$, we have
\begin{itemize}
\item
    $\otEnvEntails{\otEnv}{\ocC[j]}{\otT[j]}$, 
    $\otEnvEntails{\otEnv}{\ocCi[j]}{\otTi[j]}$, and
\item for all $i \in I$, $\otJudge{\oeEnv}{\otEnv \otEnvComp \otEnvMap{\ocX[i]}{\otT[i]} \otEnvComp \otEnvMap{\ocY[i]}{\otTi[i]}}{\oeE[i]}$.
\end{itemize}
By applying substitution lemma on $\oeE[j]$ twice, we get $\otJudge{\oeEnv}{\otEnv}{\oeE[j]\subst{\ocX[j]}{\ocC[1]}\subst{\ocY[j]}{\ocC[2]}}$.\\
\CASE \inferrule{\iruleORedRec}.
$\oeE = \oeLetrec{\oeRecVar}{\tilde{\ocX}}{\oeE[1]}{\left(\oeCall{\oeRecVar}{\tilde{\ocC}} \oePar \oeE[2]\right)}$ and
$\oeEi = \oeLetrec{\oeRecVar}{\tilde{\ocX}}{\oeE[1]}{\left(\oeE[1]\subst{\tilde{\ocX}}{\tilde{\ocC}} \oePar \oeE[2]\right)}$.
By inversion of {\bf letrec} and $\oePar$, we have
\begin{itemize}
  \item 
    $\otJudge{\oeEnv \oeEnvComp \oeEnvMap{\oeRecVar}{\otT[1],\ldots,\otT[n]}} {\otEnv \otEnvComp \otEnvMap{x_1}{\otT[1]} \otEnvComp \ldots \otEnvComp\otEnvMap{x_n}{\otT[n]}} {\oeE[1]}$
  \item
    $\forall i \in 1..n$, $\otEnvEntails{\otEnv}{\ocC[i]}{\otS[i]}$ and
    $\otS[i] \otSub \otT[i]$, and
  \item 
    $\otJudge{\oeEnv \oeEnvComp \oeEnvMap{\oeRecVar}{\otT[1],\ldots,\otT[n]}}{\otEnv}{\oeE[2]}$.
\end{itemize}
By rule \inferrule{\iruleOTCSub}, we have
$\otEnvEntails{\otEnv}{\ocC[i]}{\otT[i]}$ for all $i \in \{1,\ldots,n\}$.
By applying substitution lemma on $\oeE[1]$ repeatedly,
we get $\otJudge{\oeEnv \oeEnvComp \oeEnvMap{\oeRecVar}{\otT[1],\ldots,\otT[n]}}{\otEnv}{\oeE[1]\subst{\tilde{\ocX}}{\tilde{\ocC}}}$.
  Finally, By rule \inferrule{\iruleOTPar} and \inferrule{\iruleOTLetrec}, we get
  $\otJudge{\oeEnv}{\otEnv}{\oeLetrec{\oeRecVar}{\tilde{\ocX}}{\oeE[1]}{\left(\oeE[1]\subst{\tilde{\ocX}}{\tilde{\ocC}} \oePar \oeE[2]\right)}}$.\\
\CASE \inferrule{\iruleORedInit}.
   $\oeE = \oeInit{\ocX[1],\ldots,\ocX[n]}{\gocaml}{(\oeE[1] \oePar \cdots \oePar \oeE[n])}$,
  $\oeEi = \oeRes{\tilde{\ocS}}{(\oeE[1]\subst{\ocX[1]}{\ocC[1]} \oePar \cdots \oePar \oeE[n]\subst{\ocX[n]}{\ocC[n]})}$.
  From the premise of the rule, we have:
\begin{itemize}  
\item $\GCCVR{\gocaml}{\{\roleP[1],\elip,\roleP[n]\}}=\ocTuple{\ocC[1],\elip,\ocC[n]}$,
\item $\bigcup{}_{i\in\{1\elip n\}} \fn{\ocC[i]} = \{\tilde{\ocS}\}$, which shares base name $\ocS$ and
\item $\{\tilde{\ocS}\} \cap \bigcup{}_{i\in\{1\elip n\}} \fn{\oeE[i]} = \emptyset$.
\end{itemize}
By inversion, we have
\begin{itemize}
\item for $i \in \{1,\ldots,n\}$, $\NTH{\GCCV{\gocaml}}{i} = \otT[i]$ and
\item $\otJudge{\oeEnv}{\otEnv \otEnvComp \otEnvMap{\ocX[i]}{\otT[i]}}{\oeE[i]}$.
\end{itemize}
From \Cref{col:SubjectReductionForGC}, for
$\{\ocS[j]\}_{j \in J} = \tilde{\ocS}$
we have $\otEnvEntails{\{\otEnvMap{\ocS[j]}{\otChan{\otTi[j]}}\}_{j \in J}}{\ocC[i]}{\otT[i]}$ for all $i \in \{1,\ldots,n\}$.
By weakening, for all $i \in \{1,\ldots,n\}$, we have 
\begin{itemize}
\item $\otEnvEntails{\otEnv \otEnvComp \{\otEnvMap{\ocS[j]}{\otChan{\otTi[j]}}\}_{j \in J}}{\ocC[i]}{\otT[i]}$ and
\item $\otJudge{\oeEnv}{\otEnv \otEnvComp \{\otEnvMap{\ocS[j]}{\otChan{\otTi[j]}}\}_{j \in J} \otEnvComp \otEnvMap{\ocX[i]}{\otT[i]}}{\oeE[i]}$.
\end{itemize}
By substitution lemma, we get
$\otJudge{\oeEnv}{\otEnv \otEnvComp \{\otEnvMap{\ocS[j]}{\otChan{\otTi[j]}}\}_{j \in J}}{\oeE[i]\subst{\ocX[i]}{\ocC[i]}}$.
By applying \inferrule{\iruleOTPar} and \inferrule{\iruleOTRes}, we finally get
$\otJudge{\oeEnv}{\otEnv}{\oeRes{\ocS}{
      \left(\oeE[1]\subst{\ocX[1]}{\ocC[1]} \oePar \ldots \oePar \oeE[n]\subst{\ocX[n]}{\ocC[n]}\right)}}$.
\end{proof}

\section{Implementation: Omitted Type Signatures and Explanations}
\label{sec:appimpl}
This section gives the \OCaml type signatures and implementations of 
the main implementation building blocks using
sophisticated functional programming techniques based on GADT and polymorphic variants.
Namely, first-class methods and labels are explained in \S~\ref{sec:methods_details},
roles and variable-length tuples in \S~{sec:vartup},
input and output channels in \S~\ref{sec:iochannel_details}, 
and global combinators in \S~\ref{sec:gc_details}.

\subsection{First-Class Methods and Labels}
\label{sec:methods_details}
\begin{figure}[t]
{\lstset{numbers=left}
\begin{OCAMLLISTING}
(* the  definition of the type method_*)
type ('obj, 'mt) method_ = {make_obj: 'mt -> 'obj; call_obj: 'obj -> 'mt} ^\label{line:method2}^
(* example usage of _method: *)
val login_method : (<login : 'mt>, 'mt) method_ (* the type of login_method *)
let login_method = 
  {make_obj=(fun v -> object method login = v end); call_obj=(fun obj -> obj#login)} ^\label{line:methodusage2}^

(* the  definition of the type label*)
type ('obj, 'ot, 'var, 'vt) label = {obj: ('obj, 'ot) method_; var: 'vt -> 'var} ^\label{line:lable2}^
(* example usage of label *)
val login : (<login : 'mt>, 'mt, [> `login of 'vt], 'vt) label
let login = {obj=login_method; var=(fun v -> `login(v))} ^\label{line:lableusage2}^
\end{OCAMLLISTING}}
\caption{Implementation of first-class methods and labels}
\label{fig:labels2}
\end{figure}

As we show in \S~\ref{sec:proofsystem}, the definition of roles and labels use {\em methods} of an object.
To enable this encoding, we introduce {\em first-class} methods --
the type \oCODE{method_} defined on Line~\ref{line:method2} in \cfig{fig:labels2}. 
The type is a record with a  {\em constructor function} \oCODE{make_obj} and
a {\em destructor function} \oCODE{call_obj}.
An example usage of the type \oCODE{method_} is given on Line~\ref{line:methodusage2} by defining the type \oCODE{login_method}. 
In \oCODE{make_obj}, the expression \oCODE{(object method login=v end)}
creates an object that  consists of a method \oCODE{login} with no parameter, returning \oCODE{(v: 'mt)}.
Field \oCODE{call_obj} simply implements a method invocation (\oCODE{obj#login}).

Our encoding of local types requires label names to be encoded as an object method (in case of internal choice) and
as a variant tag (in case of external choice).
Hence, the \oCODE{label} type, Line~\ref{line:lable2}, is defined as a pair of a first-class method and
a {\em variant constructor function}.
As in \S~\ref{subsec:gcevalimpl},
while object and variant constructor functions are needed to compose a channel vector in \lstinline!(-->)!,
object destructor functions are used in \lstinline!merge! in \lstinline!choice_at!,
to extract bare channels inside an object.
Variant destructors are not needed, as they are destructed via pattern-matching and merging is
done by \lstinline!Event.choose! of Concurrent ML.
Using the types \oCODE{method_, label} the user can define arbitrary labels.

\subsection{Variable-Length Tuples and Roles}
\label{sec:vartup}
We declare {\em variable-length} tuple type (\oCODE{$t$._tup})
as a Generalised Abstract Data Type \cite{garrigue11gadt}, as follows:
\begin{LISTINGBOX}
type _ tup = Nil : nil tup | Cons : 'hd * 'tl tup -> [`cons of 'hd * 'tl] tup
and nil = [`cons of unit * 'a] as 'a
\end{LISTINGBOX}
The type \oCODE{tup} consists of
two constructors \oCODE{Nil} and \oCODE{Cons} which construct tuples $(c_1,c_2,\elip,c_n)$ as a cons-list
\oCODE{(Cons($c_1$, Cons($c_2$, .., Cons($c_n$,Nil))))}.
The element types can be {\em heterogeneous};
in type \oCODE{($t$._tup)} the argument $t$ denotes tuple type $t_1{\otTimes}{\elipc}{\otTimes}t_n$
by the nested sequence of polymorphic variant types as \oCODE{([`cons of $t_1$ * ... [`cons of $t_n$ * nil]] tup)}.
Here, the auxiliary type \oCODE{nil} is defined by an infinite sequence of \oCODE{unit} types
defined in the second line, \oCODE{([`cons of unit * 'a] as 'a)} where outer \oCODE{'a}
binds the whole \oCODE{nil} type, forming an equi-recursive type which essentially
states that the the rest of roles have a closed session \oCODE{unit}. %
Thus, \oCODE{finish} combinator is defined as \oCODE{let finish : nil tup = Nil}
which has an infinite sequence of \oCODE{unit}s on types,
denoting a terminating protocol for any number of roles.

Then, taking inspiration from \cite[\S~3.2.4]{linocaml},
we define the type-level index type on this
variable-length tuple as polymorphic lenses (see \S~\ref{subsec:chanvecimpl}),
again using GADTs.
The index type \oCODE{idx} has two constructors \oCODE{Zero} and \oCODE{Succ},
making an index in a tuple via Peano numbers.
The constructor \oCODE{Zero} says that the lens refers to the 0-th element i.e. the head of a cons, while \oCODE{Succ} takes a lens and constructs a new lens which refers to a position deeper by one. By applying \oCODE{Succ} repeatedly, elements at arbitrary depths can be referred. We store the lens for each channel vector inside the role object. For example, the roles \oCODE{s} and  \oCODE{c} from the \oCODE{oAuth} protocol in \Sec~\ref{sec:overview} are implemented as the records \oCODE|let c = {index = Zero, ...}}| and \oCODE|let s = {index = Succ(Zero), ...}| respectively. 
\begin{LISTINGBOX}
type (_,_,_,_) idx =
    Zero : (--![`cons of 't * 'tl]!--, --?'t?--, --![`cons of 'u * 'tl]!--, --?'u?--) idx
  | Succ : (--!'tl1!--, --?'t?--, --!'tl2!--, --?'u?--) idx ->
     (--![`cons of 'hd * 'tl1]!--, --?'t?--, --![`cons of 'hd * 'tl2]!--, --?'u?--) idx
val tup_get : --!'ts!-- tup -> (--!'ts!--, --?'t?--, --!'us!--, --?'u?--) idx -> --?'t?--
val tup_put : --!'ts!-- tup -> (--!'ts!--, --?'t?--, --!'us!--, --?'u?--) idx -> --?'u?-- -> --!'us!-- tup 
\end{LISTINGBOX}

\myparagraph{Roles.}
By pairing first-class methods and indices, we develop
the {\em role type}, defined in \cfig{fig:roles}.
The role type is a record with two fields,
\oCODE{role_index} denotes the index of the role within the global combinator sequence,  
while \oCODE{role_label} is a first-class encoding of the role label as a method in an object.
The full declaration of the role \oCODE{s} is given on Line~\ref{line:roleusage2}.

\begin{figure}[t]
{\lstset{numbers=left}
\begin{OCAMLLISTING}
(* the definition of the type role*)
type (--!'ts!--, --?'t?--, --!'us!--, --?'u?--, 'robj, 'mt) role = ^\label{line:role2}^
   {role_index : (--!'ts!--,--?'t?--,--!'us!--,--?'u?--) idx; role_label : ('robj,'mt) method_}
(* example usage of role: *)
val s : (--![`cons of 't0 * 'ts]!--, --?'t0?--, --![`cons of 'u0 * 'ts]!--, --?'u0?--, <role_S:'mt>, 'mt) role
let s = {role_index=Zero; ^\label{line:roleusage2}^ 
         role_label={make_obj=(fun v -> object method role_S=v end); call_obj=(fun o -> o#role_S)}} 
\end{OCAMLLISTING}}
\caption{Implementation of Roles}
\label{fig:roles}
\end{figure}

\subsection{Input and output types} 
\label{sec:iochannel_details}
To represent communication channels, 
we use the OCaml module Event, which provides a synchronous inter-thread
communications over channels. For each communication action we generate a fresh channel
and wrap it in a channel vector structure. The output \lstinline!<$m$:._($v$*$t$)._out>! is an object with a method $m$ proactively called by the sender's side choosing label $m$,
of which return type
\lstinline!($v$*$t$)._out! is just a pair of channel and continuation. 
\lstinline!  type ('v, 't)._._out = 'v Event.channel * 't (* abstract *)!\\
where \oCODE{'v Event.channel} is a standard synchronous channel type of value \lstinline!'v! in OCaml.
Note that this pair structure is abstract i.e., hidden outside the module,
to prevent abusing of the continuation \lstinline!'t! before sending on \lstinline!'v channel!.
 The output on \lstinline!'v._channel! does not transmit any labels, but they are {\em implicitly} passed.
 The transmission of the label \lstinline!m! {\em implicitly} happens, 
 when output labels are proactively chosen by calling a method $m$.
The input \lstinline![>`$m$._of._$v$*$t$]._inp! makes an external choice as an idiomatic pattern-matching on variants,
enabling a case analysis on continuations based on labels.
This is done by \lstinline!Event.wrap! function, which originates from Concurrent ML \cite{concurrentml}.
The \lstinline!wrap! function works as a \lstinline!map! on received values;
thus, by wrapping \lstinline!$v$._channel! with a function \lstinline!$v$._->._[>`m._of._$v$*$t$]!,
we obtain an input of type \lstinline![>`$m$._of._$v$*$t$]._inp!.

\subsection{Global Combinators}
\label{sec:gc_details}
This section gives the types for all global combinators. 

\myparagraph{Communication combinator} is a 4-ary combinator. 
Its type signature has many type variables which are resolved by unification,
as we already observed in \S~\ref{subsec:chanvecimpl}.
The types signature is given below, which realises the typing rule \inferrule{\iruleOTGComm} in \cfig{fig:typingforglobal}
using lenses in \lstinline!role! type and first-class methods in \lstinline!label! type:
\begin{OCAMLLISTING}
val ( --> ) : (--!'g1!--, --?'ti?--, --!'g2!--, --?'ui?--, ('ri as --?'uj?--), 'var inp) role -> (* sending role type *)
     (--!'g0!--, --?'tj?--, --!'g1!--, --?'uj?--, ('rj as --?'ui?--), 'obj) role -> (* receiving role type *)
     ('obj, ('v, 'ti) out, 'var, 'v * 'tj) label  -> (* the type of the label *)
      --!'g0!-- tup ->  (* the type of the initial tuple of channel vectors *)
      --!'g2!-- tup (* the type of the resulting tuple *)
\end{OCAMLLISTING}
In the expression \oCODE{(($\roleR[i]$ --> $\roleR[j]$) $\mpLab$ --!g!--)},
the continuation \oCODE{--!g!--} holds the tuple type \oCODE{--!('g0 tup)!--}.
By index-based update via role types, the tuple type \oCODE{--!('g1 tup)!--} is updated to \oCODE{--!('g2 tup)!--} such that 
\oCODE{$\roleR[i]$}'s channel vector in \oCODE{--!('g1 tup)!--}  is updated to 
\oCODE{<role_$\roleR[j]$: <$\mpLab$: ('v, --?ti?--) out>>}, while 
that of \oCODE{$\roleR[j]$} becomes \oCODE{<role_$\roleR[i]$: [> `$\mpLab$ of 'v * 'tj] inp>}.
Assuming that the indices of \oCODE{$\roleR[i]$} and \oCODE{$\roleR[j]$} are $i$ and $j$ respectively,
\oCODE{--!'g0!--} is updated first to \oCODE{--!'g1!--}
by changing its $j$-th element \oCODE{--?'tj?--} to \oCODE{--?'uj?--}.
Then, it is further updated to \oCODE{--!'g2!--}
by changing $i$-th element \oCODE{--?'ti?--} to \oCODE{--?'ui?--}.
Furthermore, the part \oCODE{('ri as --?'uj?--)} which equates \oCODE{--?'uj?--} and \oCODE{'ri}
determines the form of \oCODE{--?'uj?--} (at role \oCODE{$\roleR[j]$}) being \oCODE{<role_$\roleR[i]$: 'var inp>}.
Type \oCODE{'var} has the form of a variant type \oCODE{[> `$\mpLab$ of 'v * 'tj]}
which results in a faithful encoding of a receiving type,
since it is a part of variant constructor function (specified by the parameters of type \oCODE{label}).
By a similar argument, type \oCODE{--?'ui?--} equated to \oCODE{'rj}
has the form \oCODE{<role_$\roleR[j]$: <$\mpLab$: ('v, --?ti?--) out>>}, which describes the session at \oCODE{$\roleR[i]$}.

\myparagraph{Loops via lazy evaluation}
The signature of the loop combinator \oCODE{fix} is given below.
\begin{LISTINGBOX}
val fix : ('t tup -> 't tup) -> 't tup
let fix f = let rec body = lazy (f (RecVar body)) in Lazy.force body
\end{LISTINGBOX}
Function \oCODE{fix} takes a function \oCODE{f} and returns a fixpoint of it \oCODE{(x = f x)}
by utilising lazy evaluation and a {\em value recursion} which makes a cyclic data structure.
We extend the tuple types for global combinators, i.e \oCODE{'t._tup} type, with a new constructor \oCODE{RecVar} which discriminates recursion variables  from other constructors.
\oCODE{Lazy.force} tries to expand unguarded recursion variables which occurs right under the fixpoint combinator.
This enables the ``fail-fast'' policy, explained in \S~\ref{sec:proofsystem}. %
For example, an unguarded loop like \oCODE{(fix (fun t -> t))} fails with \oCODE{UnguardedLoop} exception.

\myparagraph{Branching combinator: Merging and object concatenation}
In a similar way, from \inferrule{\iruleOTGChoice} the type of the binary branching combinator
\oCODE{(choice_at._$r_a$._$\mrg$._($r_a$,--!gl!--)._($r_a$,--!gr!--))}
is implemented as follows:
\begin{LISTINGBOX}
val choice_at : (--!'g0!--, unit, --!'g!--, --?'tlr?--, 'ra, _) role -> (--?'tlr?--, --?'tl?--, --?'tr?--) disj ->
                (--!'gl!--, --?'tl?--, --!'g0!--, unit, 'ra, _) role * 'gl tup ->
                (--!'gr!--, --?'tr?--, --!'g0!--, unit, 'ra, _) role * 'gr tup -> 'g tup
\end{LISTINGBOX}
The lens part is same as in \S~\ref{subsec:chanvecimpl}.
Additionally, the role-label part \lstinline!'ra! ensures that the three roles are same.
Types \oCODE{--?'tl?--} and \oCODE{--?'tr?--} are output type of form
\oCODE{<role_$\roleQ$: <$\mpLab[{i}]$: ($v_{i}$, $t_{i}$) out>$_{i \in I}$>}
and
\oCODE{<role_$\roleQ$: <$\mpLabi[{j}]$: ($v'_{j}$, $t'_{j}$) out>$_{j \in J}$>}
where \oCODE{$\roleQ$} is the destination role and \oCODE{$\{\mpLab[i]\}$}
and \oCODE{$\{\mpLabi[j]\}$} are the set of output labels which should be disjoint from each other.
The following type \oCODE{($\mathit{lr}$,$l$,$r$)._disj} denotes a constraint
that type $\mathit{lr}$ is the type concatenated from mutually-disjoint $l$ and $r$:
\begin{LISTINGBOX}
type ('lr, 'l, 'r)._disj =
  {disj_merge: 'l -> 'r -> 'lr; disj_split_L: 'lr -> 'l; disj_split_R: 'lr -> 'r}
\end{LISTINGBOX}

\subsection{Example of Concatenating Two Disjoint Objects}
\label{sec:disj_merge}

The functions \oCODE{disj_merge} concatenates two disjoint objects \oCODE{'l} and \oCODE{'r} into one, while
\oCODE{disj_split_$\{$L,R$\}$} splits an object \oCODE{'lr} to \oCODE{'l} and \oCODE{'r}, respectively.
Both are used in the definition of a branching operator.
This constraint must manually be supplied by programmers.
For example, the following \oCODE{left_or_right} states
a concatenation of type \oCODE{<left: 'tl>} and \oCODE{<right: 'tr>} into \oCODE{<left: 'tl; right: 'tr>}:
\begin{LISTING}
val left_or_right : (<left: 'l; right: 'r>, <left: 'l>, <right: 'r>) disj
let left_or_right =
  {disj_merge=(fun l r -> object method left=l#left method right=r#right end);
   disj_split_L=(fun obj -> obj#left); disj_split_R=(fun obj -> obj#right)}
\end{LISTING}

\section{Multiparty Session Types and Processes}\label{sec:mpst}
This section quickly outlines the multiparty session types 
\cite{MPST,scalas19less}.
For the syntax of types, we follow
\cite{BettiniCDLDY08} which is the most widely used syntax in the 
literature.  
A \emph{global type}, written $\gtG,\gtG',..$, 
describes the whole 
conversation scenario of a multiparty session as a type signature, and
a {\em local type}, written by  $\stS,\stS', \cdots$.
Let $\PSet$ be a set of 
participants fixed throughout the section: $\PSet =\{\roleP, \roleQ, \roleR, \cdots\}$, and $\ASigma$ is a set of alphabets.

\begin{DEF}[Global types]\label{def:globaltypes}\rm
  The syntax of a \textbf{global type $\gtG$} is:%
  
  \smallskip%
  \scalebox{0.95}{%
  \centerline{%
  \(%
  \gtG%
  \,\coloncolonequals\,%
  \gtComm{\roleP}{\roleQ}{i \in I}{\gtLab[i]}{\stS[i]}{\gtG[i]}%
  \bnfsep%
  \gtRec{\gtRecVar}{\gtG}%
  \bnfsep%
  \gtRecVar%
  \bnfsep%
  \gtEnd%
  \)%
  \quad%
  \text{%
    with %
    $\roleP \!\neq\! \roleQ$, %
    \;$I \!\neq\! \emptyset$, %
    \;and\; $\forall i \!\in\! I: \fv{\stS[i]} = \emptyset$%
  }%
  }}%
  \smallskip%
  
\noindent
We write\; $\roleP \in \gtRoles{\gtG}$ %
\;(or simply $\roleP \!\in\! \gtG$) \;iff, for some $\roleQ$, %
either $\gtFmt{\roleP {\to} \roleQ}$ %
or $\gtFmt{\roleQ {\to} \roleP}$ %
occurs in $\gtG$.
\end{DEF}

\begin{DEF}[Local types]\rm
  The syntax of \textbf{local types} %
  is:

  \smallskip%
  \scalebox{0.92}{
  \centerline{\(%
    \textstyle%
    \stS, \stT%
    \,\coloncolonequals\,%
    \stExtSum{\roleP}{i \in I}{\stChoice{\stLab[i]}{\stS[i]} \stSeq \stSi[i]}%
    \bnfsep%
    \stIntSum{\roleP}{i \in I}{\stChoice{\stLab[i]}{\stS[i]} \stSeq \stSi[i]}%
    \bnfsep%
    \stEnd%
    \bnfsep%
    \stRec{\stRecVar}{\stS}%
    \bnfsep%
    \stRecVar
    \ %
    \text{with $I \!\neq\! \emptyset$,  
      and $\stLab[i]$ pairwise distinct}%
    \)}}%
  \smallskip%

  \noindent%
  We require types to be closed, %
  and recursion variables to be guarded.
\end{DEF}
The relation between global and local types is formalised by
{\em projection}~\cite{BettiniCDLDY08,HYC08}. 

\begin{DEF}[projection]\rm 
\label{def:projection}
The 
{\em projection of $\gtG$ onto
  $\roleP$} (written $\gtProj{\gtG}{\roleP}$)
is defined as:%

  \smallskip%
  \centerline{\(%
  \begin{array}{c}
    \gtProj{\left(%
      \gtComm{\roleQ}{\roleR}{i \in I}{\gtLab[i]}{\stS[i]}{\gtG[i]}%
      \right)}{\roleP}%
    \;=\;%
    \left\{%
    \begin{array}{@{\hskip 0.5mm}l@{\hskip 5mm}l@{}}
      \stIntSum{\roleR}{i \in I}{%
        \stChoice{\stLab[i]}{\stS[i]} \stSeq (\gtProj{\gtG[i]}{\roleP})%
      }%
      &\text{\footnotesize%
        if\, $\roleP = \roleQ$}%
      \\%
      \stExtSum{\roleQ}{i \in I}{%
        \stChoice{\stLab[i]}{\stS[i]} \stSeq (\gtProj{\gtG[i]}{\roleP})%
      }%
      &\text{\footnotesize%
        if\, $\roleP = \roleR$}%
      \\%
      \stMerge{i \in I}{\gtProj{\gtG[i]}{\roleP}}%
      &\text{\footnotesize%
        if\, $\roleQ \neq \roleP \neq \roleR$}%
    \end{array}
    \right.
    \\[5mm]%
    \gtProj{(\gtRec{\gtRecVar}{\gtG})}{\roleP}%
    \;=\;%
    \left\{%
    \begin{array}{@{\hskip 0.5mm}l@{\hskip 5mm}l@{}}
      \stRec{\stRecVar}{(\gtProj{\gtG}{\roleP})}%
      &%
      \text{\footnotesize%
        if\, $\gtProj{\gtG}{\roleP} \neq \stRecVari$ ($\forall \stRecVari$)%
      }%
      \\%
      \stEnd%
      &%
      \text{\footnotesize%
        otherwise}
    \end{array}
    \right.%
    \qquad%
    \begin{array}{@{}r@{\hskip 1mm}c@{\hskip 1mm}l@{}}
      \gtProj{\gtRecVar}{\roleP}%
      &=&%
      \stRecVar%
      \\%
      \gtProj{\gtEnd}{\roleP}%
      &=&%
      \stEnd%
    \end{array}
  \end{array}
  \)}%
  \smallskip%
  
  \noindent%
For projection of branchings, we appeal to a merge
operator along the lines of \cite{DY11}, written $\stS \stBinMerge \stS'$, ensuring
that if the locally observable behaviour of the local type is 
dependent of the chosen branch then it is identifiable via a unique
choice/branching label. 
The \emph{merging operation} $\stBinMerge{}{}$ is defined as a partial commutative 
operator over two types such that:

  \smallskip%
  \noindent%
  \scalebox{0.80}{%
  {\(%
  \begin{array}{c}%
    \textstyle%
    \stExtSum{\roleP}{i \in I}{\stChoice{\stLab[i]}{\stS[i]} \stSeq \stSi[i]}%
    \,\stBinMerge\,%
    \stExtSum{\roleP}{\!j \in J}{\stChoice{\stLab[j]}{\stS[\!j]} \stSeq \stTi[j]}%
    \;=\;%
    \stExtSum{\roleP}{k \in I \cap J}{\stChoice{\stLab[k]}{\stS[k]} \stSeq%
      (\stSi[k] \!\stBinMerge\! \stTi[k])%
    }%
    \stExtC%
    \stExtSum{\roleP}{i \in I \setminus J}{\stChoice{\stLab[i]}{\stS[i]} \stSeq \stSi[i]}%
    \stExtC%
    \stExtSum{\roleP}{\!j \in J \setminus I}{\stChoice{\stLab[j]}{\stS[\!j]} \stSeq \stTi[j]}%
    \\[1mm]%
    \stIntSum{\roleP}{i \in I}{\stChoice{\stLab[i]}{\stS[i]} \stSeq \stSi[i]}%
    \,\stBinMerge\,%
    \stIntSum{\roleP}{i \in I}{\stChoice{\stLab[i]}{\stS[i]} \stSeq \stSi[i]}%
    \;=\;%
    \stIntSum{\roleP}{i \in I}{\stChoice{\stLab[i]}{\stS[i]} \stSeq \stSi[i]}%
    \\[1mm]%
    \stRec{\stRecVar}{\stS} \,\stBinMerge\, \stRec{\stRecVar}{\stT}%
    \;=\;%
    \stRec{\stRecVar}{(\stS \stBinMerge \stT)}%
    \qquad%
    \stRecVar \,\stBinMerge\, \stRecVar%
    \;=\;%
    \stRecVar%
    \qquad%
    \stEnd \,\stBinMerge\, \stEnd%
    \;=\;%
    \stEnd%
  \end{array}
  \)}}%
\end{DEF}

\newcommand{\roleAlice}{{\color{roleColor}\roleFmt{A}}}%
\newcommand{\roleBob}{{\color{roleColor}\roleFmt{B}}}%
\newcommand{\roleCarol}{{\color{roleColor}\roleFmt{C}}}%

\noindent
We say that $G$ is {\em well-formed} if for all $\roleP\in \PSet$, 
$\gtProj{\gtG}{\roleP}$ is defined. 

Below we define the \emph{multiparty session subtyping relation}, following
\cite{PrecisenessJPSY19}\cite{DGJPY15}.\footnote{%
  \label{footnote:inverted-subt}%
  For convenience, we use the ``channel-oriented'' order %
  of \cite{GH05,scalas17linear} for our subtyping relation. %
  For a comparison with 
  ``process-oriented'' subtyping of \cite{DGJPY15},  
see \cite{Gay2016}.} %
Intuitively, a type $\stS$ is smaller than $\stS'$ %
when $\stS$ is ``less demanding'' than $\stS'$, 
\ie, when $\stS$ imposes to support less external choices 
and allows to perform more internal choices. %
Session subtyping is used in the type system %
to augment its flexibility.

\begin{DEF}[Session subtyping]\rm 
\label{def:session:subtyping}
The subtyping relation $\stSub$ is \emph{co}inductively defined as: 
  \smallskip%

  \centerline{\scalebox{0.9}{\(%
  \begin{array}{c}%
    \cinference[\iruleSubBranch]{
      \forall i \in I%
      &%
      \stS[i] \stSub \stT[i]%
      &%
      \stSi[i] \stSub \stTi[i]%
    }{%
      \stExtSum{\roleP}{i \in I}{\stChoice{\stLab[i]}{\stS[i]} \stSeq \stSi[i]}%
      \stSub%
      \stExtSum{\roleP}{i \in I \cup J}{\stChoice{\stLab[i]}{\stT[i]} \stSeq \stTi[i]}%
    }%
    \qquad%
    \cinference[\iruleSubSel]{
      \forall i \in I%
      &%
      \stT[i] \stSub \stS[i]%
      &%
      \stSi[i] \stSub \stTi[i]%
    }{%
      \stIntSum{\roleP}{i \in I \cup J}{\stChoice{\stLab[i]}{\stS[i]} \stSeq \stSi[i]}%
      \stSub%
      \stIntSum{\roleP}{i \in I}{\stChoice{\stLab[i]}{\stT[i]} \stSeq \stTi[i]}%
    }%
    \\[2mm]%
    \cinference[\iruleSubEnd]{%
      \phantom{X}%
    }{%
      \stEnd \stSub \stEnd%
    }%
    \qquad%
    \cinference[\iruleSubRecL]{%
      \stS\subst{\stRecVar}{\stRec{\stRecVar}{\stS}} \stSub \stT%
    }{%
      \stRec{\stRecVar}{\stS} \stSub \stT%
    }%
    \qquad%
    \cinference[\iruleSubRecR]{%
      \stS \stSub \stT\subst{\stRecVar}{\stRec{\stRecVar}{\stT}}%
    }{%
      \stS \stSub \stRec{\stRecVar}{\stT}
    }%
  \end{array}
  \)}}%
\end{DEF}

 }{}

\end{document}